%% file: main.tex
\pdfoutput=1
\documentclass[english]{jssst_ppl} 

\usepackage{preamble}

\title{Type-based Qubit Allocation for a First-Order Quantum Programming Language}
\author{Ryo Wakizaka$^1$, Atsushi Igarashi$^2$}
\inst{%
$^1$Graduate of Informatics, Kyoto University\\
\texttt{wakizaka@fos.kuis.kyoto-u.ac.jp}
\medskip\par
$^2$Graduate of Informatics, Kyoto University\\
\texttt{igarashi@kuis.kyoto-u.ac.jp}
}
\begin{document}
\maketitle
\begin{abstract}
  Qubit allocation is a process to assign physical qubits to logical
  qubits in a quantum program.  Since some quantum computers have
  \emph{connectivity constraints} on applications of two-qubit
  operations, it is mainly concerned with finding an assignment and
  inserting instructions to satisfy the connectivity constraints.
  Many methods have been proposed for the qubit allocation problem for low-level quantum programs.

  This paper presents a type-based framework of qubit allocation
  for a quantum programming language with first-order functions.  In our
  framework, the connectivity constraints are expressed by a simple
  graph of qubits called a coupling graph.  We formalize (1) the source
  language, whose type system verifies that the number of qubits
  required for a given program to run does not exceed the number of
  nodes of the coupling graph, (2) the target language, whose qualified
  type system verifies that a well-typed program satisfies the
  connectivity constraints, and (3) an algorithm to translate a source
  program into a target program.  We prove that both languages are
  type-safe and that the translation algorithm is type preserving.
\end{abstract}

\input{mysystem}

\input{languages_ott_defs}

\input{introduction}

\input{background}
\input{srclang}

\input{tgtlang}

\input{algorithm}

\input{related_work}

\input{conclusion}

\bibliographystyle{junsrt}
\bibliography{references}

\newpage
\appendix
\input{appendix}

\input{appendix-tgt}
\input{appendix-alg}

\end{document}

%% file: mysystem.tex
\NewDocumentCommand\STJudgeExp{}{\Theta  \pipe  N  \pipe  \Gamma  \vdash  e  \languagessym{:}  T}

\NewDocumentCommand\STReturn{}{
  \infrule[T-Return]{
  }{
    \Theta  \pipe  N  \pipe  \languagesmv{x_{{\mathrm{1}}}}  \languagessym{:}  \tau_{{\mathrm{1}}}  \languagessym{,} \, .. \, \languagessym{,}  \languagesmv{x_{\languagesmv{n}}}  \languagessym{:}  \tau_{\languagesmv{n}}  \vdash  \languagessym{(}  \languagesmv{x_{{\mathrm{1}}}}  \languagessym{,} \, .. \, \languagessym{,}  \languagesmv{x_{\languagesmv{n}}}  \languagessym{)}  \languagessym{:}  \tau_{{\mathrm{1}}}  \languagessym{*} \, .. \, \languagessym{*}  \tau_{\languagesmv{n}}
  }
}

\NewDocumentCommand\STInit{}{
  \infrule[T-Init]{
    \languagesmv{x} \not\in \dom(\Gamma)
    \andalso N \geq 1
    \andalso \Theta  \pipe  N  \languagessym{-}  \languagessym{1}  \pipe  \Gamma  \languagessym{,}  \languagesmv{x}  \languagessym{:}  \languageskw{qbit}  \vdash  e  \languagessym{:}  T
  }{
    \Theta  \pipe  N  \pipe  \Gamma  \vdash  \languageskw{let} \, \languagesmv{x}  \languagessym{=} \, \languageskw{init} \, \languagessym{()} \, \languageskw{in} \, e  \languagessym{:}  T
  }
}

\NewDocumentCommand\STDiscard{}{
  \infrule[T-Discard]{
    \Theta  \pipe  N  \languagessym{+}  \languagessym{1}  \pipe  \Gamma  \vdash  e  \languagessym{:}  T
  }{
    \Theta  \pipe  N  \pipe  \Gamma  \languagessym{,}  \languagesmv{x}  \languagessym{:}  \languageskw{qbit}  \vdash  \languageskw{discard} \, \languagesmv{x}  \languagessym{;}  e  \languagessym{:}  T
  }
}

\NewDocumentCommand\STLet{}{
  \infrule[T-Let]{
    \dom(\Gamma_{{\mathrm{1}}}) \cap \dom(\Gamma_{{\mathrm{2}}}) = \emptyset
    \andalso    N' = N + |\dom(\Gamma_{{\mathrm{1}}})| - n
    \andalso \languagesmv{x_{{\mathrm{1}}}}  \languagessym{,} \, .. \, \languagessym{,}  \languagesmv{x_{\languagesmv{n}}} \not\in \dom(\Gamma_{{\mathrm{2}}}) \\
    \Theta  \pipe  N  \pipe  \Gamma_{{\mathrm{1}}}  \vdash  e_{{\mathrm{1}}}  \languagessym{:}  \tau_{{\mathrm{1}}}  \languagessym{*} \, .. \, \languagessym{*}  \tau_{\languagesmv{n}}
    \andalso \Theta  \pipe  N'  \pipe  \Gamma_{{\mathrm{2}}}  \languagessym{,}  \languagesmv{x_{{\mathrm{1}}}}  \languagessym{:}  \tau_{{\mathrm{1}}}  \languagessym{,} \, .. \, \languagessym{,}  \languagesmv{x_{\languagesmv{n}}}  \languagessym{:}  \tau_{\languagesmv{n}}  \vdash  e_{{\mathrm{2}}}  \languagessym{:}  T
  }{
    \Theta  \pipe  N  \pipe  \Gamma_{{\mathrm{1}}}  \languagessym{,}  \Gamma_{{\mathrm{2}}}  \vdash  \languageskw{let} \, \languagessym{(}  \languagesmv{x_{{\mathrm{1}}}}  \languagessym{,} \, .. \, \languagessym{,}  \languagesmv{x_{\languagesmv{n}}}  \languagessym{)}  \languagessym{=}  e_{{\mathrm{1}}} \, \languageskw{in} \, e_{{\mathrm{2}}}  \languagessym{:}  T
  }
}

\NewDocumentCommand\STCnot{}{
  \infrule[T-Cnot]{
    \languagesmv{x_{{\mathrm{1}}}}  \languagessym{,}  \languagesmv{x_{{\mathrm{2}}}} \not\in \dom(\Gamma)
    \andalso \Theta  \pipe  N  \pipe  \Gamma  \languagessym{,}  \languagesmv{x_{{\mathrm{1}}}}  \languagessym{:}  \languageskw{qbit}  \languagessym{,}  \languagesmv{x_{{\mathrm{2}}}}  \languagessym{:}  \languageskw{qbit}  \vdash  e  \languagessym{:}  T
  }{
    \Theta  \pipe  N  \pipe  \Gamma  \languagessym{,}  \languagesmv{y_{{\mathrm{1}}}}  \languagessym{:}  \languageskw{qbit}  \languagessym{,}  \languagesmv{y_{{\mathrm{2}}}}  \languagessym{:}  \languageskw{qbit}  \vdash  \languageskw{let} \, \languagessym{(}  \languagesmv{x_{{\mathrm{1}}}}  \languagessym{,}  \languagesmv{x_{{\mathrm{2}}}}  \languagessym{)}  \languagessym{=} \, \languageskw{cnot} \, \languagessym{(}  \languagesmv{y_{{\mathrm{1}}}}  \languagessym{,}  \languagesmv{y_{{\mathrm{2}}}}  \languagessym{)} \, \languageskw{in} \, e  \languagessym{:}  T
  }
}

\NewDocumentCommand\STIf{}{
  \infrule[T-If]{
    \Gamma  \languagessym{(}  \languagesmv{x}  \languagessym{)} = \languageskw{qbit}
    \andalso \Theta  \pipe  N  \pipe  \Gamma  \vdash  e_{{\mathrm{1}}}  \languagessym{:}  T
    \andalso \Theta  \pipe  N  \pipe  \Gamma  \vdash  e_{{\mathrm{2}}}  \languagessym{:}  T
  }{
    \Theta  \pipe  N  \pipe  \Gamma  \vdash  \languageskw{if} \, \languagesmv{x} \, \languageskw{then} \, e_{{\mathrm{1}}} \, \languageskw{else} \, e_{{\mathrm{2}}}  \languagessym{:}  T
  }
}

\NewDocumentCommand\STCall{}{
  \infrule[T-Call]{
    \Theta  \languagessym{(}  \languagesmv{f}  \languagessym{)} =  \tau'_{{\mathrm{1}}}  \languagessym{*} \, .. \, \languagessym{*}  \tau'_{\languagesmv{m}}  \xrightarrow{ N' }  \tau''_{{\mathrm{1}}}  \languagessym{*} \, .. \, \languagessym{*}  \tau''_{\languagesmv{n}}  \\
    \languagesmv{x_{{\mathrm{1}}}}  \languagessym{,} \, .. \, \languagessym{,}  \languagesmv{x_{\languagesmv{n}}} \not\in \dom(\Gamma)
    \andalso N \geq N'
    \andalso N'' = N - n + m \\
    \Theta  \pipe  N''  \pipe  \Gamma  \languagessym{,}  \languagesmv{x_{{\mathrm{1}}}}  \languagessym{:}  \tau''_{{\mathrm{1}}}  \languagessym{,} \, .. \, \languagessym{,}  \languagesmv{x_{\languagesmv{n}}}  \languagessym{:}  \tau''_{\languagesmv{n}}  \vdash  e  \languagessym{:}  T
  }{
    \Theta  \pipe  N  \pipe  \Gamma  \languagessym{,}  \languagesmv{y_{{\mathrm{1}}}}  \languagessym{:}  \tau_{{\mathrm{1}}}  \languagessym{,} \, .. \, \languagessym{,}  \languagesmv{y_{\languagesmv{m}}}  \languagessym{:}  \tau_{\languagesmv{m}}  \vdash  \languageskw{let} \, \languagessym{(}  \languagesmv{x_{{\mathrm{1}}}}  \languagessym{,} \, .. \, \languagessym{,}  \languagesmv{x_{\languagesmv{n}}}  \languagessym{)}  \languagessym{=}  \languagesmv{f}  \languagessym{(}  \languagesmv{y_{{\mathrm{1}}}}  \languagessym{,} \, .. \, \languagessym{,}  \languagesmv{y_{\languagesmv{m}}}  \languagessym{)} \, \languageskw{in} \, e  \languagessym{:}  T
  }
}

\NewDocumentCommand\STJudgeFunDef{}{\Theta  \vdash  D}

\NewDocumentCommand\STFunDef{}{
  \infrule[T-FunDef]{
    \Theta  \languagessym{,}  \languagesmv{f}  \languagessym{:}   \tau_{{\mathrm{1}}}  \languagessym{*} \, .. \, \languagessym{*}  \tau_{\languagesmv{m}}  \xrightarrow{ N }  T   \pipe  N  \pipe  \languagesmv{x_{{\mathrm{1}}}}  \languagessym{:}  \tau_{{\mathrm{1}}}  \languagessym{,} \, .. \, \languagessym{,}  \languagesmv{x_{\languagesmv{m}}}  \languagessym{:}  \tau_{\languagesmv{m}}  \vdash  e  \languagessym{:}  T \\
  }{
    \Theta  \languagessym{,}  \languagesmv{f}  \languagessym{:}   \tau_{{\mathrm{1}}}  \languagessym{*} \, .. \, \languagessym{*}  \tau_{\languagesmv{m}}  \xrightarrow{ N }  T   \vdash   \languagesmv{f}  \mapsto ( \languagesmv{x_{{\mathrm{1}}}}  \languagessym{,} \, .. \, \languagessym{,}  \languagesmv{x_{\languagesmv{m}}} )  e 
  }
}

\NewDocumentCommand\SJudgeProg{}{N  \vdash   \braket{  D  ,  e  } }

\NewDocumentCommand\STProg{}{
  \infrule[T-Prog]{
    \andalso \Theta  \vdash  D
    \andalso \Theta  \pipe  N  \pipe   \emptyset   \vdash  e  \languagessym{:}  T
  }{
    N  \vdash   \braket{  D  ,  e  } 
  }
}


\NewDocumentCommand\SETrans{}{\qclosure{X, \rho, e}  \rightarrow _{ D }  \qclosure{X', \rho', e'}}
                           
\NewDocumentCommand\SELetI{}{
  \infax[E-Let1]{
    \Braket{X, \rho, \languageskw{let} \, \languagessym{(}  \languagesmv{x_{{\mathrm{1}}}}  \languagessym{,} \, .. \, \languagessym{,}  \languagesmv{x_{\languagesmv{n}}}  \languagessym{)}  \languagessym{=}  \languagessym{(}  \languagesmv{y_{{\mathrm{1}}}}  \languagessym{,} \, .. \, \languagessym{,}  \languagesmv{y_{\languagesmv{n}}}  \languagessym{)} \, \languageskw{in} \, e}  \rightarrow _{ D }  \Braket{X, \rho, \languagessym{[}  \languagesmv{y_{{\mathrm{1}}}}  \slash  \languagesmv{x_{{\mathrm{1}}}}  \languagessym{,} \, .. \, \languagessym{,}  \languagesmv{y_{\languagesmv{n}}}  \slash  \languagesmv{x_{\languagesmv{n}}}  \languagessym{]} \, e}
  }
}

\NewDocumentCommand\SELetII{}{
  \infrule[E-Let2]{
    \Braket{X, \rho, e_{{\mathrm{1}}}}  \rightarrow _{ D }  \Braket{X', \rho', e'_{{\mathrm{1}}}}
  }{
    \Braket{X, \rho, \languageskw{let} \, \languagessym{(}  \languagesmv{x_{{\mathrm{1}}}}  \languagessym{,} \, .. \, \languagessym{,}  \languagesmv{x_{\languagesmv{n}}}  \languagessym{)}  \languagessym{=}  e_{{\mathrm{1}}} \, \languageskw{in} \, e_{{\mathrm{2}}}}  \rightarrow _{ D }  \Braket{X', \rho', \languageskw{let} \, \languagessym{(}  \languagesmv{x_{{\mathrm{1}}}}  \languagessym{,} \, .. \, \languagessym{,}  \languagesmv{x_{\languagesmv{n}}}  \languagessym{)}  \languagessym{=}  e'_{{\mathrm{1}}} \, \languageskw{in} \, e_{{\mathrm{2}}}}
  }
}

\NewDocumentCommand\SEInit{}{
  \infrule[E-Init]{
    \languagesmv{x'} \in X
  }{
    \Braket{X, \rho, \languageskw{let} \, \languagesmv{x}  \languagessym{=} \, \languageskw{init} \, \languagessym{()} \, \languageskw{in} \, e}
            \rightarrow _{ D }  \braket{X \setminus \{\languagesmv{x'}\}, \rho, \languagessym{[}  \languagesmv{x'}  \slash  \languagesmv{x}  \languagessym{]} \, e}
  }
}

\NewDocumentCommand\SEDiscard{}{
  \infax[E-Discard]{
    \braket{X, \rho, \languageskw{discard} \, \languagesmv{x}  \languagessym{;}  e}  \rightarrow _{ D }  \braket{X \uplus \{x\}, \sum_{b \in \{0, 1\}} \ketbra{0}{b}{x'}\rho\ketbra{b}{0}{x'}, e}
  }
}

\NewDocumentCommand\SECnot{}{
  \infrule[E-Cnot]{
    U = \mathit{CNOT}_{y_1, y_2}
  }{
    \braket{X, \rho, \languageskw{let} \, \languagessym{(}  \languagesmv{x_{{\mathrm{1}}}}  \languagessym{,}  \languagesmv{x_{{\mathrm{2}}}}  \languagessym{)}  \languagessym{=} \, \languageskw{cnot} \, \languagessym{(}  \languagesmv{y_{{\mathrm{1}}}}  \languagessym{,}  \languagesmv{y_{{\mathrm{2}}}}  \languagessym{)} \, \languageskw{in} \, e}
            \rightarrow _{ D }  \braket{X, U\rho U^\dagger, \languagessym{[}  \languagesmv{y_{{\mathrm{1}}}}  \slash  \languagesmv{x_{{\mathrm{1}}}}  \languagessym{,}  \languagesmv{y_{{\mathrm{2}}}}  \slash  \languagesmv{x_{{\mathrm{2}}}}  \languagessym{]} \, e}
  }
}

\NewDocumentCommand\SECall{}{
  \infrule[E-Call]{
     \languagesmv{f}  \mapsto ( \languagesmv{y'_{{\mathrm{1}}}}  \languagessym{,} \, .. \, \languagessym{,}  \languagesmv{y'_{\languagesmv{m}}} )  e'  \in D
  }{
    \braket{X, \rho, \languageskw{let} \, \languagessym{(}  \languagesmv{x_{{\mathrm{1}}}}  \languagessym{,} \, .. \, \languagessym{,}  \languagesmv{x_{\languagesmv{n}}}  \languagessym{)}  \languagessym{=}  \languagesmv{f}  \languagessym{(}  \languagesmv{y_{{\mathrm{1}}}}  \languagessym{,} \, .. \, \languagessym{,}  \languagesmv{y_{\languagesmv{m}}}  \languagessym{)} \, \languageskw{in} \, e} \\
            \rightarrow _{ D }  \braket{X, \rho, \languageskw{let} \, \languagessym{(}  \languagesmv{x_{{\mathrm{1}}}}  \languagessym{,} \, .. \, \languagessym{,}  \languagesmv{x_{\languagesmv{n}}}  \languagessym{)}  \languagessym{=}  \languagessym{[}  \languagesmv{y_{{\mathrm{1}}}}  \slash  \languagesmv{y'_{{\mathrm{1}}}}  \languagessym{,} \, .. \, \languagessym{,}  \languagesmv{y_{\languagesmv{m}}}  \slash  \languagesmv{y'_{\languagesmv{m}}}  \languagessym{]} \, e' \, \languageskw{in} \, e}
  }
}

\NewDocumentCommand\SEIfTrue{}{
  \infrule[E-IfTrue]{
    M_1 = \frac{I - Z_x}{2}
  }{
    \braket{X, \rho, \languageskw{if} \, \languagesmv{x} \, \languageskw{then} \, e_{{\mathrm{1}}} \, \languageskw{else} \, e_{{\mathrm{2}}}}  \rightarrow _{ D }  \braket{X, M_1\rho M_1^\dagger, e_{{\mathrm{1}}}}
  }
}

\NewDocumentCommand\SEIfFalse{}{
  \infrule[E-IfFalse]{
    M_0 = \frac{I + Z_x}{2}
  }{
    \braket{X, \rho, \languageskw{if} \, \languagesmv{x} \, \languageskw{then} \, e_{{\mathrm{1}}} \, \languageskw{else} \, e_{{\mathrm{2}}}}  \rightarrow _{ D }  \braket{X, M_0\rho M_0^\dagger, e_{{\mathrm{2}}}}
  }
}


\NewDocumentCommand\TTJudgeExp{}{\Theta  \pipe  \Phi  \pipe  \Gamma  \vdash  e  \languagessym{:}  T}

\NewDocumentCommand\TTReturn{}{
  \infrule[T-Return]{
  }{
    \Theta  \pipe  \Phi  \pipe  \languagesmv{x_{{\mathrm{1}}}}  \languagessym{:}  \tau_{{\mathrm{1}}}  \languagessym{,} \, .. \, \languagessym{,}  \languagesmv{x_{\languagesmv{n}}}  \languagessym{:}  \tau_{\languagesmv{n}}  \vdash  \languagessym{(}  \languagesmv{x_{{\mathrm{1}}}}  \languagessym{,} \, .. \, \languagessym{,}  \languagesmv{x_{\languagesmv{n}}}  \languagessym{)}  \languagessym{:}  \tau_{{\mathrm{1}}}  \languagessym{*} \, .. \, \languagessym{*}  \tau_{\languagesmv{n}}
  }
}

\NewDocumentCommand\TTInit{}{
  \infrule[T-Init]{
    \Gamma  \languagessym{(}  \languagesmv{x}  \languagessym{)} =  \texttt{q}( \alpha ) 
    \andalso \Theta  \pipe  \Phi  \pipe  \Gamma  \vdash  e  \languagessym{:}  T
  }{
    \Theta  \pipe  \Phi  \pipe  \Gamma  \vdash  \languageskw{init} \, \languagesmv{x}  \languagessym{;}  e  \languagessym{:}  T
  }
}

\NewDocumentCommand\TTSwap{}{
  \infrule[T-Swap]{
     \alpha_{{\mathrm{1}}}  \sim  \alpha_{{\mathrm{2}}}  \in \Phi
    \andalso \Theta  \pipe  \Phi  \pipe  \Gamma  \languagessym{,}  \languagesmv{x_{{\mathrm{1}}}}  \languagessym{:}   \texttt{q}( \alpha_{{\mathrm{1}}} )   \languagessym{,}  \languagesmv{x_{{\mathrm{2}}}}  \languagessym{:}   \texttt{q}( \alpha_{{\mathrm{2}}} )   \vdash  e  \languagessym{:}  T
  }{
    \Theta  \pipe  \Phi  \pipe  \Gamma  \languagessym{,}  \languagesmv{y_{{\mathrm{1}}}}  \languagessym{:}   \texttt{q}( \alpha_{{\mathrm{1}}} )   \languagessym{,}  \languagesmv{y_{{\mathrm{2}}}}  \languagessym{:}   \texttt{q}( \alpha_{{\mathrm{2}}} )   \vdash  \languageskw{let} \, \languagessym{(}  \languagesmv{x_{{\mathrm{1}}}}  \languagessym{,}  \languagesmv{x_{{\mathrm{2}}}}  \languagessym{)}  \languagessym{=} \, \languageskw{swap} \, \languagessym{(}  \languagesmv{y_{{\mathrm{1}}}}  \languagessym{,}  \languagesmv{y_{{\mathrm{2}}}}  \languagessym{)} \, \languageskw{in} \, e  \languagessym{:}  T
  }
}

\NewDocumentCommand\TTCnot{}{
\infrule[T-Cnot]{
     \alpha_{{\mathrm{1}}}  \sim  \alpha_{{\mathrm{2}}}  \in \Phi
    \andalso \Theta  \pipe  \Phi  \pipe  \Gamma  \languagessym{,}  \languagesmv{x_{{\mathrm{1}}}}  \languagessym{:}   \texttt{q}( \alpha_{{\mathrm{1}}} )   \languagessym{,}  \languagesmv{x_{{\mathrm{2}}}}  \languagessym{:}   \texttt{q}( \alpha_{{\mathrm{2}}} )   \vdash  e  \languagessym{:}  T
  }{
    \Theta  \pipe  \Phi  \pipe  \Gamma  \languagessym{,}  \languagesmv{y_{{\mathrm{1}}}}  \languagessym{:}   \texttt{q}( \alpha_{{\mathrm{1}}} )   \languagessym{,}  \languagesmv{y_{{\mathrm{2}}}}  \languagessym{:}   \texttt{q}( \alpha_{{\mathrm{2}}} )   \vdash  \languageskw{let} \, \languagessym{(}  \languagesmv{x_{{\mathrm{1}}}}  \languagessym{,}  \languagesmv{x_{{\mathrm{2}}}}  \languagessym{)}  \languagessym{=} \, \languageskw{cnot} \, \languagessym{(}  \languagesmv{y_{{\mathrm{1}}}}  \languagessym{,}  \languagesmv{y_{{\mathrm{2}}}}  \languagessym{)} \, \languageskw{in} \, e  \languagessym{:}  T
  }
}

\NewDocumentCommand\TTIf{}{
  \infrule[T-If]{
    \languagesmv{x} \in \dom(\Gamma)
    \andalso \Theta  \pipe  \Phi  \pipe  \Gamma  \vdash  e_{{\mathrm{1}}}  \languagessym{:}  T
    \andalso \Theta  \pipe  \Phi  \pipe  \Gamma  \vdash  e_{{\mathrm{2}}}  \languagessym{:}  T
  }{
    \Theta  \pipe  \Phi  \pipe  \Gamma  \vdash  \languageskw{if} \, \languagesmv{x} \, \languageskw{then} \, e_{{\mathrm{1}}} \, \languageskw{else} \, e_{{\mathrm{2}}}  \languagessym{:}  T
  }
}

\NewDocumentCommand\TTLet{}{
  \infrule[T-Let]{
    \Theta  \pipe  \Phi  \pipe  \Gamma_{{\mathrm{1}}}  \vdash  e_{{\mathrm{1}}}  \languagessym{:}  \tau_{{\mathrm{1}}}  \languagessym{*} \, .. \, \languagessym{*}  \tau_{\languagesmv{m}} \\
    \Theta  \pipe  \Phi  \pipe  \Gamma_{{\mathrm{2}}}  \languagessym{,}  \languagesmv{x_{{\mathrm{1}}}}  \languagessym{:}  \tau_{{\mathrm{1}}}  \languagessym{,} \, .. \, \languagessym{,}  \languagesmv{x_{\languagesmv{m}}}  \languagessym{:}  \tau_{\languagesmv{m}}  \vdash  e_{{\mathrm{2}}}  \languagessym{:}  T
  }{
    \Theta  \pipe  \Phi  \pipe  \Gamma_{{\mathrm{1}}}  \languagessym{,}  \Gamma_{{\mathrm{2}}}  \vdash  \languageskw{let} \, \languagessym{(}  \languagesmv{x_{{\mathrm{1}}}}  \languagessym{,} \, .. \, \languagessym{,}  \languagesmv{x_{\languagesmv{m}}}  \languagessym{)}  \languagessym{=}  e_{{\mathrm{1}}} \, \languageskw{in} \, e_{{\mathrm{2}}}  \languagessym{:}  T
  }
}

\NewDocumentCommand\TTCall{}{
  \infrule[T-Call]{
    \Theta  \languagessym{(}  \languagesmv{f}  \languagessym{)} =  \forall   \overline{ \alpha }   .   \Phi'  \Rightarrow  \tau'_{{\mathrm{1}}}  \languagessym{*} \, .. \, \languagessym{*}  \tau'_{\languagesmv{n}}  \rightarrow  \tau''_{{\mathrm{1}}}  \languagessym{*} \, .. \, \languagessym{*}  \tau''_{\languagesmv{m}}   \\
     \sigma _{ \alpha }  = \languagessym{[}   \overline{ \alpha' } / \overline{ \alpha }   \languagessym{]}
    \andalso  \sigma _{ \alpha } \Phi' \subseteq \Phi
    \andalso \forall i \in \{1, \dots, n\}.  \sigma _{ \alpha }  \tau'_{\languagesmv{i}} = \tau_{\languagesmv{i}} \\
    \Theta  \pipe  \Phi  \pipe  \Gamma  \languagessym{,}  \languagesmv{x_{{\mathrm{1}}}}  \languagessym{:}   \sigma _{ \alpha }  \, \tau''_{{\mathrm{1}}}  \languagessym{,} \, .. \, \languagessym{,}  \languagesmv{x_{\languagesmv{m}}}  \languagessym{:}   \sigma _{ \alpha }  \, \tau''_{\languagesmv{m}}  \vdash  e  \languagessym{:}  T
  }{
    \Theta  \pipe  \Phi  \pipe  \Gamma  \languagessym{,}  \languagesmv{y_{{\mathrm{1}}}}  \languagessym{:}  \tau_{{\mathrm{1}}}  \languagessym{,} \, .. \, \languagessym{,}  \languagesmv{y_{\languagesmv{n}}}  \languagessym{:}  \tau_{\languagesmv{n}}  \vdash  \languageskw{let} \, \languagessym{(}  \languagesmv{x_{{\mathrm{1}}}}  \languagessym{,} \, .. \, \languagessym{,}  \languagesmv{x_{\languagesmv{m}}}  \languagessym{)}  \languagessym{=}  \languagesmv{f}  \languagessym{(}  \languagesmv{y_{{\mathrm{1}}}}  \languagessym{,} \, .. \, \languagessym{,}  \languagesmv{y_{\languagesmv{n}}}  \languagessym{)} \, \languageskw{in} \, e  \languagessym{:}  T
  }
}

\NewDocumentCommand\TTJudgeFunDef{}{\Theta  \vdash  D}

\NewDocumentCommand\TTFunDef{}{
  \infrule[T-FunDef]{
    \Theta  \languagessym{,}  \languagesmv{f}  \languagessym{:}   \forall   \overline{ \alpha }   .   \Phi  \Rightarrow  \tau_{{\mathrm{1}}}  \languagessym{*} \, .. \, \languagessym{*}  \tau_{\languagesmv{m}}  \rightarrow  T    \pipe  \Phi  \pipe  \languagesmv{x_{{\mathrm{1}}}}  \languagessym{:}  \tau_{{\mathrm{1}}}  \languagessym{,} \, .. \, \languagessym{,}  \languagesmv{x_{\languagesmv{m}}}  \languagessym{:}  \tau_{\languagesmv{m}}  \vdash  e  \languagessym{:}  T \\
  }{
    \Theta  \languagessym{,}  \languagesmv{f}  \languagessym{:}   \forall   \overline{ \alpha }   .   \Phi  \Rightarrow  \tau_{{\mathrm{1}}}  \languagessym{*} \, .. \, \languagessym{*}  \tau_{\languagesmv{m}}  \rightarrow  T    \vdash   \languagesmv{f}  \mapsto ( \languagesmv{x_{{\mathrm{1}}}}  \languagessym{,} \, .. \, \languagessym{,}  \languagesmv{x_{\languagesmv{m}}} )  e 
  }
}

\NewDocumentCommand\TTProg{}{
  \infrule[T-Prog]{
    \Theta  \vdash  D
    \andalso \Theta  \pipe  \Phi  \pipe  \Gamma  \vdash  e  \languagessym{:}  T
  }{
    \Phi  \pipe  \Gamma  \vdash   \braket{  D  ,  e  } 
  }
}


\NewDocumentCommand\TETrans{}{\qclosure{\rho, e}  \rightarrow _{D, G} \qclosure{\rho', e'}}

\NewDocumentCommand\TEInit{}{
  \infax[E-Init]{
    \qclosure{\rho, \languageskw{init} \, \languagesmv{x}  \languagessym{;}  e}
       \rightarrow _{D, G} \qclosure{\sum_{b \in \{0, 1\}} \ketbra{0}{b}{x}\rho\ketbra{b}{0}{x}, e}
  }
}

\NewDocumentCommand\TELetI{}{
  \infax[E-Let1]{
    \braket{\rho, \languageskw{let} \, \languagessym{(}  \languagesmv{x_{{\mathrm{1}}}}  \languagessym{,} \, .. \, \languagessym{,}  \languagesmv{x_{\languagesmv{n}}}  \languagessym{)}  \languagessym{=}  \languagessym{(}  \languagesmv{y_{{\mathrm{1}}}}  \languagessym{,} \, .. \, \languagessym{,}  \languagesmv{y_{\languagesmv{n}}}  \languagessym{)} \, \languageskw{in} \, e}
       \rightarrow _{D, G} \braket{\rho, \languagessym{[}  \languagesmv{y_{{\mathrm{1}}}}  \slash  \languagesmv{x_{{\mathrm{1}}}}  \languagessym{,} \, .. \, \languagessym{,}  \languagesmv{y_{\languagesmv{n}}}  \slash  \languagesmv{x_{\languagesmv{n}}}  \languagessym{]} \, e}
  }
}

\NewDocumentCommand\TELetII{}{
  \infrule[E-Let2]{
    \braket{\rho, e_{{\mathrm{1}}}}  \rightarrow _{D, G} \braket{\rho', e'_{{\mathrm{1}}}}
  }{
    \braket{\rho, \languageskw{let} \, \languagessym{(}  \languagesmv{x_{{\mathrm{1}}}}  \languagessym{,} \, .. \, \languagessym{,}  \languagesmv{x_{\languagesmv{n}}}  \languagessym{)}  \languagessym{=}  e_{{\mathrm{1}}} \, \languageskw{in} \, e_{{\mathrm{2}}}}
       \rightarrow _{D, G} \braket{\rho', \languageskw{let} \, \languagessym{(}  \languagesmv{x_{{\mathrm{1}}}}  \languagessym{,} \, .. \, \languagessym{,}  \languagesmv{x_{\languagesmv{n}}}  \languagessym{)}  \languagessym{=}  e'_{{\mathrm{1}}} \, \languageskw{in} \, e_{{\mathrm{2}}}}
  }
}

\NewDocumentCommand\TESwap{}{
  \infrule[E-Swap]{
    (\languagesmv{y_{{\mathrm{1}}}}, \languagesmv{y_{{\mathrm{2}}}}) \in G
    \andalso U = \mathit{SWAP}_{y_1, y_2}
  }{
    \braket{\rho, \languageskw{let} \, \languagessym{(}  \languagesmv{x_{{\mathrm{1}}}}  \languagessym{,}  \languagesmv{x_{{\mathrm{2}}}}  \languagessym{)}  \languagessym{=} \, \languageskw{swap} \, \languagessym{(}  \languagesmv{y_{{\mathrm{1}}}}  \languagessym{,}  \languagesmv{y_{{\mathrm{2}}}}  \languagessym{)} \, \languageskw{in} \, e}
       \rightarrow _{D, G} \braket{U\rho U^\dagger, e}
  }
}

\NewDocumentCommand\TECnot{}{
  \infrule[E-Cnot]{
    (y_1, y_2) \in G
    \andalso U = \mathit{CNOT}_{y_1, y_2}
  }{
    \qclosure{\rho, \languageskw{let} \, \languagessym{(}  \languagesmv{x_{{\mathrm{1}}}}  \languagessym{,}  \languagesmv{x_{{\mathrm{2}}}}  \languagessym{)}  \languagessym{=} \, \languageskw{cnot} \, \languagessym{(}  \languagesmv{y_{{\mathrm{1}}}}  \languagessym{,}  \languagesmv{y_{{\mathrm{2}}}}  \languagessym{)} \, \languageskw{in} \, e}
       \rightarrow _{D, G} \qclosure{U\rho U^\dagger, \languagessym{[}  \languagesmv{y_{{\mathrm{1}}}}  \slash  \languagesmv{x_{{\mathrm{1}}}}  \languagessym{,}  \languagesmv{y_{{\mathrm{2}}}}  \slash  \languagesmv{x_{{\mathrm{2}}}}  \languagessym{]} \, e}
  }
}

\NewDocumentCommand\TECall{}{
  \infrule[T-Call]{
     \languagesmv{f}  \mapsto ( \languagesmv{y'_{{\mathrm{1}}}}  \languagessym{,} \, .. \, \languagessym{,}  \languagesmv{y'_{\languagesmv{m}}} )  e_{{\mathrm{1}}}  \in D
  }{
    \braket{\rho, \languageskw{let} \, \languagessym{(}  \languagesmv{x_{{\mathrm{1}}}}  \languagessym{,} \, .. \, \languagessym{,}  \languagesmv{x_{\languagesmv{n}}}  \languagessym{)}  \languagessym{=}  \languagesmv{f}  \languagessym{(}  \languagesmv{y_{{\mathrm{1}}}}  \languagessym{,} \, .. \, \languagessym{,}  \languagesmv{y_{\languagesmv{m}}}  \languagessym{)} \, \languageskw{in} \, e_{{\mathrm{2}}}} \\
       \rightarrow _{D, G} \braket{\rho, \languageskw{let} \, \languagessym{(}  \languagesmv{x_{{\mathrm{1}}}}  \languagessym{,} \, .. \, \languagessym{,}  \languagesmv{x_{\languagesmv{n}}}  \languagessym{)}  \languagessym{=}  \languagessym{[}  \languagesmv{y_{{\mathrm{1}}}}  \slash  \languagesmv{y'_{{\mathrm{1}}}}  \languagessym{,} \, .. \, \languagessym{,}  \languagesmv{y_{\languagesmv{m}}}  \slash  \languagesmv{y'_{\languagesmv{m}}}  \languagessym{]} \, e_{{\mathrm{1}}} \, \languageskw{in} \, e_{{\mathrm{2}}}}
  }
}

%% file: languages_ott_defs.tex
\newcommand{\languagesdrule}[4][]{{\displaystyle\frac{\begin{array}{l}#2\end{array}}{#3}\quad\languagesdrulename{#4}}}
\newcommand{\languagesusedrule}[1]{\[#1\]}
\newcommand{\languagespremise}[1]{ #1 \\}
\newenvironment{languagesdefnblock}[3][]{ \framebox{\mbox{#2}} \quad #3 \\[0pt]}{}
\newenvironment{languagesfundefnblock}[3][]{ \framebox{\mbox{#2}} \quad #3 \\[0pt]\begin{displaymath}\begin{array}{l}}{\end{array}\end{displaymath}}
\newcommand{\languagesfunclause}[2]{ #1 \equiv #2 \\}
\newcommand{\languagesnt}[1]{\mathit{#1}}
\newcommand{\languagesmv}[1]{\mathit{#1}}
\newcommand{\languageskw}[1]{\mathbf{#1}}
\newcommand{\languagessym}[1]{#1}
\newcommand{\languagescom}[1]{\text{#1}}
\newcommand{\languagesdrulename}[1]{\textsc{#1}}
\newcommand{\languagescomplu}[5]{\overline{#1}^{\,#2\in #3 #4 #5}}
\newcommand{\languagescompu}[3]{\overline{#1}^{\,#2<#3}}
\newcommand{\languagescomp}[2]{\overline{#1}^{\,#2}}
\newcommand{\languagesgrammartabular}[1]{\begin{supertabular}{llcllllll}#1\end{supertabular}}
\newcommand{\languagesmetavartabular}[1]{\begin{supertabular}{ll}#1\end{supertabular}}
\newcommand{\languagesrulehead}[3]{$#1$ & & $#2$ & & & \multicolumn{2}{l}{#3}}
\newcommand{\languagesprodline}[6]{& & $#1$ & $#2$ & $#3 #4$ & $#5$ & $#6$}
\newcommand{\languagesfirstprodline}[6]{\languagesprodline{#1}{#2}{#3}{#4}{#5}{#6}}
\newcommand{\languageslongprodline}[2]{& & $#1$ & \multicolumn{4}{l}{$#2$}}
\newcommand{\languagesfirstlongprodline}[2]{\languageslongprodline{#1}{#2}}
\newcommand{\languagesbindspecprodline}[6]{\languagesprodline{#1}{#2}{#3}{#4}{#5}{#6}}
\newcommand{\languagesprodnewline}{\\}
\newcommand{\languagesinterrule}{\\[5.0mm]}
\newcommand{\languagesafterlastrule}{\\}
\newcommand{\languagesmetavars}{
\languagesmetavartabular{
 $ \languagesmv{termvar} ,\, \languagesmv{x} ,\, \languagesmv{y} ,\, \languagesmv{z} $ &  \\
 $ \languagesmv{funvar} ,\, \languagesmv{f} $ &  \\
 $ \languagesmv{declvar} ,\, \languagesmv{d} $ &  \\
 $ \languagesmv{qidxvar} ,\, \alpha ,\, \beta ,\, \gamma $ &  \\
 $ \languagesmv{substvar} ,\, \sigma $ &  \\
 $ \languagesmv{permvar} ,\, \languagesmv{p} $ &  \\
 $ \languagesmv{connectivity\_var} ,\, \languagesmv{c} $ &  \\
 $ \languagesmv{index} ,\, \languagesmv{n} ,\, \languagesmv{m} ,\, \languagesmv{l} ,\, \languagesmv{i} ,\, \languagesmv{j} ,\, \languagesmv{k} ,\, \languagesmv{L} ,\, \languagesmv{N} $ &  \\
}}

\newcommand{\languagesNat}{
\languagesrulehead{N}{::=}{\languagescom{natural_numbers}}\languagesprodnewline
\languagesfirstprodline{|}{\languagessym{1}}{}{}{}{}\languagesprodnewline
\languagesprodline{|}{N_{{\mathrm{1}}}  \languagessym{+}  N_{{\mathrm{2}}}}{}{}{}{}\languagesprodnewline
\languagesprodline{|}{N_{{\mathrm{1}}}  \languagessym{-}  N_{{\mathrm{2}}}}{}{}{}{}}

\newcommand{\languagesdeclvarXXseq}{
\languagesrulehead{\languagesnt{declvar\_seq}}{::=}{}\languagesprodnewline
\languagesfirstprodline{|}{\languagesmv{declvar}}{}{}{}{}\languagesprodnewline
\languagesprodline{|}{\languagesnt{declvar\_seq_{{\mathrm{1}}}}  \languagessym{,} \, .. \, \languagessym{,}  \languagesnt{declvar\_seq_{\languagesmv{n}}}}{}{}{}{}}

\newcommand{\languagesO}{
\languagesrulehead{\Theta}{::=}{\languagescom{fun_context}}\languagesprodnewline
\languagesfirstprodline{|}{ \emptyset }{}{}{}{}\languagesprodnewline
\languagesprodline{|}{\languagesmv{f}  \languagessym{:}  \theta}{}{}{}{}\languagesprodnewline
\languagesprodline{|}{\languagesmv{f}  \languagessym{:}  \theta}{}{}{}{}\languagesprodnewline
\languagesprodline{|}{\Theta_{{\mathrm{1}}}  \languagessym{,}  \Theta_{{\mathrm{2}}}}{}{}{}{}\languagesprodnewline
\languagesprodline{|}{\languagesnt{qidx\_subst} \, \Theta}{}{}{}{}}

\newcommand{\languagesOXXapply}{
\languagesrulehead{\languagesnt{O\_apply}}{::=}{}\languagesprodnewline
\languagesfirstprodline{|}{\Theta  \languagessym{(}  \languagesmv{funvar}  \languagessym{)}}{}{}{}{}}

\newcommand{\languagesG}{
\languagesrulehead{\Gamma}{::=}{\languagescom{context}}\languagesprodnewline
\languagesfirstprodline{|}{ \emptyset }{}{}{}{}\languagesprodnewline
\languagesprodline{|}{\languagesmv{x}  \languagessym{:}  \tau}{}{}{}{}\languagesprodnewline
\languagesprodline{|}{\languagesmv{x}  \languagessym{:}  \tau}{}{}{}{}\languagesprodnewline
\languagesprodline{|}{\Gamma_{{\mathrm{1}}}  \languagessym{,}  \Gamma_{{\mathrm{2}}}}{}{}{}{}\languagesprodnewline
\languagesprodline{|}{\Gamma_{{\mathrm{1}}}  \languagessym{,} \, .. \, \languagessym{,}  \Gamma_{{\mathrm{2}}}}{}{}{}{}\languagesprodnewline
\languagesprodline{|}{\languagesnt{var\_subst} \, \Gamma}{}{}{}{}\languagesprodnewline
\languagesprodline{|}{\languagesnt{qidx\_subst} \, \Gamma}{}{}{}{}\languagesprodnewline
\languagesprodline{|}{ \Psi (  \Gamma  ) }{}{}{}{}}

\newcommand{\languagesGXXapply}{
\languagesrulehead{\languagesnt{G\_apply}}{::=}{}\languagesprodnewline
\languagesfirstprodline{|}{\Gamma  \languagessym{(}  \languagesmv{x}  \languagessym{)}}{}{}{}{}}

\newcommand{\languagesunitary}{
\languagesrulehead{\languagesnt{unitary}}{::=}{}\languagesprodnewline
\languagesfirstprodline{|}{\languageskw{H}}{}{}{}{}}

\newcommand{\languagesqidx}{
\languagesrulehead{\languagesnt{qidx}}{::=}{}\languagesprodnewline
\languagesfirstprodline{|}{\languagesmv{qidxvar}}{}{}{}{}\languagesprodnewline
\languagesprodline{|}{\languagesnt{qidx\_subst} \, \languagesnt{qidx}}{}{}{}{}\languagesprodnewline
\languagesprodline{|}{ \phi ( \languagesnt{qidx} ) }{}{}{}{}}

\newcommand{\languagesqidxXXseq}{
\languagesrulehead{\languagesnt{qidx\_seq}}{::=}{}\languagesprodnewline
\languagesfirstprodline{|}{ \epsilon }{}{}{}{}\languagesprodnewline
\languagesprodline{|}{\languagesnt{qidx}}{}{}{}{}\languagesprodnewline
\languagesprodline{|}{\languagesnt{qidx\_seq_{{\mathrm{1}}}}  \languagessym{,}  \languagesnt{qidx\_seq_{{\mathrm{2}}}}}{}{}{}{}\languagesprodnewline
\languagesprodline{|}{\languagesnt{qidx_{{\mathrm{1}}}}  \languagessym{,} \, .. \, \languagessym{,}  \languagesnt{qidx_{{\mathrm{2}}}}}{}{}{}{}\languagesprodnewline
\languagesprodline{|}{ \overline{ \languagesnt{qidx} } }{}{}{}{}\languagesprodnewline
\languagesprodline{|}{\languagesnt{qidx\_subst} \, \languagesnt{qidx\_seq}}{}{}{}{}\languagesprodnewline
\languagesprodline{|}{ \mathit{QV}( T ) }{}{}{}{}\languagesprodnewline
\languagesprodline{|}{ \mathit{QV}( \Gamma ) }{}{}{}{}}

\newcommand{\languagesqidxXXsubstXXinner}{
\languagesrulehead{\languagesnt{qidx\_subst\_inner}}{::=}{}\languagesprodnewline
\languagesfirstprodline{|}{\languagesnt{qidx}  \slash  \languagesnt{qidx'}}{}{}{}{}\languagesprodnewline
\languagesprodline{|}{\languagesnt{qidx\_subst\_inner_{{\mathrm{1}}}}  \languagessym{,} \, .. \, \languagessym{,}  \languagesnt{qidx\_subst\_inner_{\languagesmv{n}}}}{}{}{}{}\languagesprodnewline
\languagesprodline{|}{ \overline{ \languagesnt{qidx} } / \overline{ \languagesnt{qidx'} } }{}{}{}{}}

\newcommand{\languagesqidxXXsubst}{
\languagesrulehead{\languagesnt{qidx\_subst}}{::=}{}\languagesprodnewline
\languagesfirstprodline{|}{\languagesmv{substvar}}{}{}{}{}\languagesprodnewline
\languagesprodline{|}{ \languagesmv{substvar} _{ \languagesnt{qidx} } }{}{}{}{}\languagesprodnewline
\languagesprodline{|}{\languagessym{[}  \languagesnt{qidx\_subst\_inner}  \languagessym{]}}{}{}{}{}\languagesprodnewline
\languagesprodline{|}{ \languagesnt{qidx\_subst}  \circ  \languagesnt{qidx\_subst'} }{}{}{}{}}

\newcommand{\languagesvar}{
\languagesrulehead{\languagesnt{var}}{::=}{}\languagesprodnewline
\languagesfirstprodline{|}{\languagesmv{x}}{}{}{}{}\languagesprodnewline
\languagesprodline{|}{\languagesnt{var\_subst} \, \languagesnt{var}}{}{}{}{}}

\newcommand{\languagesvarXXseq}{
\languagesrulehead{\languagesnt{var\_seq}}{::=}{}\languagesprodnewline
\languagesfirstprodline{|}{\languagesnt{var}}{}{}{}{}\languagesprodnewline
\languagesprodline{|}{\languagesnt{var}  \languagessym{,}  \languagesnt{var\_seq}}{}{}{}{}\languagesprodnewline
\languagesprodline{|}{\languagesnt{var_{{\mathrm{1}}}}  \languagessym{,} \, .. \, \languagessym{,}  \languagesnt{var_{\languagesmv{n}}}}{}{}{}{}\languagesprodnewline
\languagesprodline{|}{ \overline{ \languagesnt{var} } }{}{}{}{}\languagesprodnewline
\languagesprodline{|}{\languagesnt{var\_subst} \, \languagesnt{var\_seq}}{}{}{}{}}

\newcommand{\languagesvarXXsubst}{
\languagesrulehead{\languagesnt{var\_subst}}{::=}{}\languagesprodnewline
\languagesfirstprodline{|}{\sigma}{}{}{}{}\languagesprodnewline
\languagesprodline{|}{ \sigma _{ \languagesmv{x} } }{}{}{}{}\languagesprodnewline
\languagesprodline{|}{\languagessym{[}  \languagesmv{x}  \slash  \languagesmv{x'}  \languagessym{]}}{}{}{}{}\languagesprodnewline
\languagesprodline{|}{\languagesnt{var\_subst} \, \languagesnt{var\_subst'}}{}{}{}{}\languagesprodnewline
\languagesprodline{|}{\languagessym{[}  \languagesmv{x_{{\mathrm{1}}}}  \slash  \languagesmv{x'_{{\mathrm{1}}}}  \languagessym{,} \, .. \, \languagessym{,}  \languagesmv{x_{\languagesmv{n}}}  \slash  \languagesmv{x'_{\languagesmv{n}}}  \languagessym{]}}{}{}{}{}\languagesprodnewline
\languagesprodline{|}{\languagessym{[}  \languagesmv{x_{{\mathrm{1}}}}  \slash  \languagesmv{x'_{{\mathrm{1}}}}  \languagessym{]} \, .. \, \languagessym{[}  \languagesmv{x_{\languagesmv{n}}}  \slash  \languagesmv{x'_{\languagesmv{n}}}  \languagessym{]}}{}{}{}{}\languagesprodnewline
\languagesprodline{|}{\languagessym{[}  v  \slash  \languagesmv{x}  \languagessym{]}}{}{}{}{}}

\newcommand{\languagessty}{
\languagesrulehead{\tau}{::=}{\languagescom{src_types}}\languagesprodnewline
\languagesfirstprodline{|}{\languageskw{qbit}}{}{}{}{}}

\newcommand{\languagessfunty}{
\languagesrulehead{\theta}{::=}{\languagescom{function_types}}\languagesprodnewline
\languagesfirstprodline{|}{ T  \xrightarrow{ N }  T' }{}{}{}{}}

\newcommand{\languagesstupleXXty}{
\languagesrulehead{T}{::=}{\languagescom{src_tuple_types}}\languagesprodnewline
\languagesfirstprodline{|}{\languageskw{ty}}{}{}{}{}\languagesprodnewline
\languagesprodline{|}{T_{{\mathrm{1}}}  \languagessym{*}  T_{{\mathrm{2}}}}{}{}{}{}\languagesprodnewline
\languagesprodline{|}{\tau_{{\mathrm{1}}}  \languagessym{*} \, .. \, \languagessym{*}  \tau_{\languagesmv{n}}}{}{}{}{}}

\newcommand{\languagessdecl}{
\languagesrulehead{d}{::=}{\languagescom{declaration}}\languagesprodnewline
\languagesfirstprodline{|}{\languagesmv{declvar}}{}{}{}{}\languagesprodnewline
\languagesprodline{|}{ \languagesmv{funvar}  \mapsto ( \languagesnt{var\_seq} )  e }{}{}{}{}}

\newcommand{\languagessdecls}{
\languagesrulehead{D}{::=}{\languagescom{declarations}}\languagesprodnewline
\languagesfirstprodline{|}{\languagessym{\{}  d_{{\mathrm{1}}}  \languagessym{,} \, .. \, \languagessym{,}  d_{\languagesmv{n}}  \languagessym{\}}}{}{}{}{}\languagesprodnewline
\languagesprodline{|}{d_{{\mathrm{1}}}  \languagessym{,} \, .. \, \languagessym{,}  d_{\languagesmv{n}}}{}{}{}{}\languagesprodnewline
\languagesprodline{|}{D_{{\mathrm{1}}}  \languagessym{,}  D_{{\mathrm{2}}}}{}{}{}{}}

\newcommand{\languagessrcXXprogram}{
\languagesrulehead{P}{::=}{\languagescom{program}}\languagesprodnewline
\languagesfirstprodline{|}{ \braket{  D  ,  e  } }{}{}{}{}}

\newcommand{\languagesse}{
\languagesrulehead{e}{::=}{\languagescom{expression}}\languagesprodnewline
\languagesfirstprodline{|}{v}{}{}{}{}\languagesprodnewline
\languagesprodline{|}{\languageskw{let} \, \languagesnt{var}  \languagessym{=} \, \languageskw{init} \, \languagessym{()} \, \languageskw{in} \, e}{}{}{}{}\languagesprodnewline
\languagesprodline{|}{\languageskw{discard} \, \languagesnt{var}  \languagessym{;}  e}{}{}{}{}\languagesprodnewline
\languagesprodline{|}{\languageskw{let} \, \languagessym{(}  \languagesnt{var_{{\mathrm{1}}}}  \languagessym{,}  \languagesnt{var_{{\mathrm{2}}}}  \languagessym{)}  \languagessym{=} \, \languageskw{cnot} \, \languagessym{(}  \languagesnt{var'_{{\mathrm{1}}}}  \languagessym{,}  \languagesnt{var'_{{\mathrm{2}}}}  \languagessym{)} \, \languageskw{in} \, e}{}{}{}{}\languagesprodnewline
\languagesprodline{|}{\languageskw{let} \, \languagessym{(}  \languagesnt{var\_seq}  \languagessym{)}  \languagessym{=}  \languagesmv{f}  \languagessym{(}  \languagesnt{var\_seq'}  \languagessym{)} \, \languageskw{in} \, e}{}{}{}{}\languagesprodnewline
\languagesprodline{|}{\languageskw{if} \, \languagesnt{var} \, \languageskw{then} \, e_{{\mathrm{1}}} \, \languageskw{else} \, e_{{\mathrm{2}}}}{}{}{}{}\languagesprodnewline
\languagesprodline{|}{\languageskw{let} \, \languagesmv{f}  \languagessym{(} \, \languageskw{variable\_seq} \, \languagessym{)}  \languagessym{=}  e_{{\mathrm{1}}} \, \languageskw{in} \, e_{{\mathrm{2}}}}{}{}{}{}\languagesprodnewline
\languagesprodline{|}{\languageskw{let} \, \languagessym{(}  \languagesnt{var\_seq}  \languagessym{)}  \languagessym{=}  e_{{\mathrm{1}}} \, \languageskw{in} \, e_{{\mathrm{2}}}}{}{}{}{}\languagesprodnewline
\languagesprodline{|}{\languagesnt{var\_subst} \, e}{}{}{}{}\languagesprodnewline
\languagesprodline{|}{\languageskw{let} \, \languagesnt{var}  \languagessym{=}  \languagesnt{unitary}  \languagessym{(}  \languagesnt{var'}  \languagessym{)} \, \languageskw{in} \, e}{}{}{}{}\languagesprodnewline
\languagesprodline{|}{\languageskw{let} \, \languagessym{(}  \languagesnt{var\_seq_{{\mathrm{1}}}}  \languagessym{)}  \languagessym{=}  \languagesnt{unitary}  \languagessym{(}  \languagesnt{var\_seq_{{\mathrm{2}}}}  \languagessym{)} \, \languageskw{in} \, e}{}{}{}{}}

\newcommand{\languagessval}{
\languagesrulehead{v}{::=}{\languagescom{values}}\languagesprodnewline
\languagesfirstprodline{|}{\languagesmv{x}}{}{}{}{}\languagesprodnewline
\languagesprodline{|}{\languagessym{(}  \languagessym{)}}{}{}{}{}\languagesprodnewline
\languagesprodline{|}{\languagesnt{var\_seq}}{}{}{}{}\languagesprodnewline
\languagesprodline{|}{\languagessym{(}  \languagesnt{var\_seq}  \languagessym{)}}{}{}{}{}}

\newcommand{\languagessrcXXjudge}{
\languagesrulehead{\languagesnt{src\_judge}}{::=}{}\languagesprodnewline
\languagesfirstprodline{|}{\Theta  \pipe  N  \pipe  \Gamma  \vdash  e  \languagessym{:}  T}{}{}{}{}}

\newcommand{\languagessdeclsXXjudge}{
\languagesrulehead{\languagesnt{sdecls\_judge}}{::=}{}\languagesprodnewline
\languagesfirstprodline{|}{\Theta  \vdash  D}{}{}{}{}}

\newcommand{\languagessprogXXjudge}{
\languagesrulehead{\languagesnt{sprog\_judge}}{::=}{}\languagesprodnewline
\languagesfirstprodline{|}{N  \vdash  P}{}{}{}{}}

\newcommand{\languagestyqclosureXXaxiom}{
\languagesrulehead{\languagesnt{tyqclosure\_axiom}}{::=}{}\languagesprodnewline
\languagesfirstprodline{|}{ N  \pipe  \Gamma  \vdash_{ D } }{}{}{}{}}

\newcommand{\languageswfXXCG}{
\languagesrulehead{\languagesnt{wf\_CG}}{::=}{}\languagesprodnewline
\languagesfirstprodline{|}{ \vdash _{ \text{WF} }  \Phi  \pipe  \Gamma }{}{}{}{}}

\newcommand{\languagestdecl}{
\languagesrulehead{d}{::=}{\languagescom{declaration}}\languagesprodnewline
\languagesfirstprodline{|}{ \languagesmv{funvar}  \mapsto ( \languagesnt{var\_seq} )  e }{}{}{}{}}

\newcommand{\languagestdecls}{
\languagesrulehead{D}{::=}{\languagescom{declarations}}\languagesprodnewline
\languagesfirstprodline{|}{\languagessym{\{}  d_{{\mathrm{1}}}  \languagessym{,} \, .. \, \languagessym{,}  d_{\languagesmv{n}}  \languagessym{\}}}{}{}{}{}\languagesprodnewline
\languagesprodline{|}{d_{{\mathrm{1}}}  \languagessym{,} \, .. \, \languagessym{,}  d_{\languagesmv{n}}}{}{}{}{}\languagesprodnewline
\languagesprodline{|}{D_{{\mathrm{1}}}  \languagessym{,}  D_{{\mathrm{2}}}}{}{}{}{}}

\newcommand{\languagestdeclXXjudge}{
\languagesrulehead{\languagesnt{tdecl\_judge}}{::=}{}\languagesprodnewline
\languagesfirstprodline{|}{\Theta  \vdash  D}{}{}{}{}}

\newcommand{\languagestty}{
\languagesrulehead{\tau}{::=}{\languagescom{target_types}}\languagesprodnewline
\languagesfirstprodline{|}{ \texttt{q}( \languagesnt{qidx} ) }{}{}{}{}\languagesprodnewline
\languagesprodline{|}{\languagesnt{qidx\_subst} \, \tau}{}{}{}{}\languagesprodnewline
\languagesprodline{|}{\Gamma  \languagessym{(}  \languagesmv{x}  \languagessym{)}}{}{}{}{}}

\newcommand{\languagesttupleXXty}{
\languagesrulehead{T}{::=}{}\languagesprodnewline
\languagesfirstprodline{|}{\tau}{}{}{}{}\languagesprodnewline
\languagesprodline{|}{T_{{\mathrm{1}}}  \languagessym{*}  T_{{\mathrm{2}}}}{}{}{}{}\languagesprodnewline
\languagesprodline{|}{\tau_{{\mathrm{1}}}  \languagessym{*} \, .. \, \languagessym{*}  \tau_{\languagesmv{n}}}{}{}{}{}\languagesprodnewline
\languagesprodline{|}{\languagesnt{qidx\_subst} \, T}{}{}{}{}}

\newcommand{\languageste}{
\languagesrulehead{e}{::=}{\languagescom{expression}}\languagesprodnewline
\languagesfirstprodline{|}{v}{}{}{}{}\languagesprodnewline
\languagesprodline{|}{\languageskw{let} \, \languagessym{(}  \languagesnt{var\_seq}  \languagessym{)}  \languagessym{=}  \languagessym{(}  \languagesnt{var\_seq'}  \languagessym{)} \, \languageskw{in} \, e}{}{}{}{}\languagesprodnewline
\languagesprodline{|}{\languageskw{init} \, \languagesmv{x}  \languagessym{;}  e}{}{}{}{}\languagesprodnewline
\languagesprodline{|}{\languageskw{discard} \, \languagesmv{x}  \languagessym{;}  e}{}{}{}{}\languagesprodnewline
\languagesprodline{|}{\languageskw{let} \, \languagessym{(}  \languagesmv{x_{{\mathrm{1}}}}  \languagessym{,}  \languagesmv{x_{{\mathrm{2}}}}  \languagessym{)}  \languagessym{=} \, \languageskw{swap} \, \languagessym{(}  \languagesmv{y_{{\mathrm{1}}}}  \languagessym{,}  \languagesmv{y_{{\mathrm{2}}}}  \languagessym{)} \, \languageskw{in} \, e}{}{}{}{}\languagesprodnewline
\languagesprodline{|}{\languageskw{let} \, \languagessym{(}  \languagesmv{x_{{\mathrm{1}}}}  \languagessym{,}  \languagesmv{x_{{\mathrm{2}}}}  \languagessym{)}  \languagessym{=} \, \languageskw{cnot} \, \languagessym{(}  \languagesmv{y_{{\mathrm{1}}}}  \languagessym{,}  \languagesmv{y_{{\mathrm{2}}}}  \languagessym{)} \, \languageskw{in} \, e}{}{}{}{}\languagesprodnewline
\languagesprodline{|}{\languageskw{let} \, \languagessym{(}  \languagesnt{var\_seq}  \languagessym{)}  \languagessym{=}  \languagesmv{f}  \languagessym{(}  \languagesnt{var\_seq'}  \languagessym{)} \, \languageskw{in} \, e}{}{}{}{}\languagesprodnewline
\languagesprodline{|}{\languageskw{if} \, \languagesmv{x} \, \languageskw{then} \, e_{{\mathrm{1}}} \, \languageskw{else} \, e_{{\mathrm{2}}}}{}{}{}{}\languagesprodnewline
\languagesprodline{|}{\languageskw{let} \, \languagessym{(}  \languagesnt{var\_seq}  \languagessym{)}  \languagessym{=}  e_{{\mathrm{1}}} \, \languageskw{in} \, e_{{\mathrm{2}}}}{}{}{}{}\languagesprodnewline
\languagesprodline{|}{\languagesnt{var\_subst} \, e}{}{}{}{}}

\newcommand{\languagestval}{
\languagesrulehead{v}{::=}{}\languagesprodnewline
\languagesfirstprodline{|}{\languagessym{(}  \languagesnt{var\_seq}  \languagessym{)}}{}{}{}{}}

\newcommand{\languagesC}{
\languagesrulehead{\Phi}{::=}{\languagescom{constraint}}\languagesprodnewline
\languagesfirstprodline{|}{ \emptyset }{}{}{}{}\languagesprodnewline
\languagesprodline{|}{\languagesmv{connectivity\_var}}{}{}{}{}\languagesprodnewline
\languagesprodline{|}{ \languagesnt{qidx}  \sim  \languagesnt{qidx'} }{}{}{}{}\languagesprodnewline
\languagesprodline{|}{\Phi_{{\mathrm{1}}}  \languagessym{,}  \Phi_{{\mathrm{2}}}}{}{}{}{}\languagesprodnewline
\languagesprodline{|}{\languagesnt{qidx\_subst} \, \Phi}{}{}{}{}\languagesprodnewline
\languagesprodline{|}{\languageskw{perm} \, \Phi}{}{}{}{}}

\newcommand{\languagestfunty}{
\languagesrulehead{\theta}{::=}{\languagescom{function_types}}\languagesprodnewline
\languagesfirstprodline{|}{ \Phi  \Rightarrow  T  \rightarrow  T' }{}{}{}{}\languagesprodnewline
\languagesprodline{|}{ \forall  \languagesnt{qidx\_seq}  .  \theta }{}{}{}{}}

\newcommand{\languagestjudge}{
\languagesrulehead{\languagesnt{tjudge}}{::=}{}\languagesprodnewline
\languagesfirstprodline{|}{\Theta  \pipe  \Phi  \pipe  \Gamma  \vdash  e  \languagessym{:}  T}{}{}{}{}}

\newcommand{\languagestgtXXprogram}{
\languagesrulehead{P}{::=}{\languagescom{program}}\languagesprodnewline
\languagesfirstprodline{|}{ \braket{  D  ,  e  } }{}{}{}{}}

\newcommand{\languagestprogXXjudge}{
\languagesrulehead{\languagesnt{tprog\_judge}}{::=}{}\languagesprodnewline
\languagesfirstprodline{|}{\Phi  \pipe  \Gamma  \vdash  P}{}{}{}{}}

\newcommand{\languagestyqclosureXXaxiomXXtgt}{
\languagesrulehead{\languagesnt{tyqclosure\_axiom\_tgt}}{::=}{}\languagesprodnewline
\languagesfirstprodline{|}{ \Phi  \pipe  \Gamma  \vdash_{ D } }{}{}{}{}}

\newcommand{\languagesarrows}{
\languagesrulehead{\languagesnt{arrows}}{::=}{}\languagesprodnewline
\languagesfirstprodline{|}{ \rightarrow _{ D } }{}{}{}{}\languagesprodnewline
\languagesprodline{|}{ \xrightarrow{s} }{}{}{}{}}

\newcommand{\languagesterminals}{
\languagesrulehead{\languagesnt{terminals}}{::=}{}\languagesprodnewline
\languagesfirstprodline{|}{ \rightarrow }{}{}{}{}\languagesprodnewline
\languagesprodline{|}{ \Rightarrow }{}{}{}{}\languagesprodnewline
\languagesprodline{|}{\languagessym{;}}{}{}{}{}\languagesprodnewline
\languagesprodline{|}{ \vdash }{}{}{}{}\languagesprodnewline
\languagesprodline{|}{ \slash }{}{}{}{}\languagesprodnewline
\languagesprodline{|}{ \neq }{}{}{}{}\languagesprodnewline
\languagesprodline{|}{ \pipe }{}{}{}{}}

\newcommand{\languagesformula}{
\languagesrulehead{\languagesnt{formula}}{::=}{}\languagesprodnewline
\languagesfirstprodline{|}{\languagesnt{judgement}}{}{}{}{}}

\newcommand{\languagesjudgement}{
\languagesrulehead{\languagesnt{judgement}}{::=}{}}

\newcommand{\languagesuserXXsyntax}{
\languagesrulehead{\languagesnt{user\_syntax}}{::=}{}\languagesprodnewline
\languagesfirstprodline{|}{\languagesmv{termvar}}{}{}{}{}\languagesprodnewline
\languagesprodline{|}{\languagesmv{funvar}}{}{}{}{}\languagesprodnewline
\languagesprodline{|}{\languagesmv{declvar}}{}{}{}{}\languagesprodnewline
\languagesprodline{|}{\languagesmv{qidxvar}}{}{}{}{}\languagesprodnewline
\languagesprodline{|}{\languagesmv{substvar}}{}{}{}{}\languagesprodnewline
\languagesprodline{|}{\languagesmv{permvar}}{}{}{}{}\languagesprodnewline
\languagesprodline{|}{\languagesmv{connectivity\_var}}{}{}{}{}\languagesprodnewline
\languagesprodline{|}{\languagesmv{index}}{}{}{}{}\languagesprodnewline
\languagesprodline{|}{N}{}{}{}{}\languagesprodnewline
\languagesprodline{|}{\languagesnt{declvar\_seq}}{}{}{}{}\languagesprodnewline
\languagesprodline{|}{\Theta}{}{}{}{}\languagesprodnewline
\languagesprodline{|}{\languagesnt{O\_apply}}{}{}{}{}\languagesprodnewline
\languagesprodline{|}{\Gamma}{}{}{}{}\languagesprodnewline
\languagesprodline{|}{\languagesnt{G\_apply}}{}{}{}{}\languagesprodnewline
\languagesprodline{|}{\languagesnt{unitary}}{}{}{}{}\languagesprodnewline
\languagesprodline{|}{\languagesnt{qidx}}{}{}{}{}\languagesprodnewline
\languagesprodline{|}{\languagesnt{qidx\_seq}}{}{}{}{}\languagesprodnewline
\languagesprodline{|}{\languagesnt{qidx\_subst\_inner}}{}{}{}{}\languagesprodnewline
\languagesprodline{|}{\languagesnt{qidx\_subst}}{}{}{}{}\languagesprodnewline
\languagesprodline{|}{\languagesnt{var}}{}{}{}{}\languagesprodnewline
\languagesprodline{|}{\languagesnt{var\_seq}}{}{}{}{}\languagesprodnewline
\languagesprodline{|}{\languagesnt{var\_subst}}{}{}{}{}\languagesprodnewline
\languagesprodline{|}{\tau}{}{}{}{}\languagesprodnewline
\languagesprodline{|}{\theta}{}{}{}{}\languagesprodnewline
\languagesprodline{|}{T}{}{}{}{}\languagesprodnewline
\languagesprodline{|}{d}{}{}{}{}\languagesprodnewline
\languagesprodline{|}{D}{}{}{}{}\languagesprodnewline
\languagesprodline{|}{P}{}{}{}{}\languagesprodnewline
\languagesprodline{|}{e}{}{}{}{}\languagesprodnewline
\languagesprodline{|}{v}{}{}{}{}\languagesprodnewline
\languagesprodline{|}{\languagesnt{src\_judge}}{}{}{}{}\languagesprodnewline
\languagesprodline{|}{\languagesnt{sdecls\_judge}}{}{}{}{}\languagesprodnewline
\languagesprodline{|}{\languagesnt{sprog\_judge}}{}{}{}{}\languagesprodnewline
\languagesprodline{|}{\languagesnt{tyqclosure\_axiom}}{}{}{}{}\languagesprodnewline
\languagesprodline{|}{\languagesnt{wf\_CG}}{}{}{}{}\languagesprodnewline
\languagesprodline{|}{d}{}{}{}{}\languagesprodnewline
\languagesprodline{|}{D}{}{}{}{}\languagesprodnewline
\languagesprodline{|}{\languagesnt{tdecl\_judge}}{}{}{}{}\languagesprodnewline
\languagesprodline{|}{\tau}{}{}{}{}\languagesprodnewline
\languagesprodline{|}{T}{}{}{}{}\languagesprodnewline
\languagesprodline{|}{e}{}{}{}{}\languagesprodnewline
\languagesprodline{|}{v}{}{}{}{}\languagesprodnewline
\languagesprodline{|}{\Phi}{}{}{}{}\languagesprodnewline
\languagesprodline{|}{\theta}{}{}{}{}\languagesprodnewline
\languagesprodline{|}{\languagesnt{tjudge}}{}{}{}{}\languagesprodnewline
\languagesprodline{|}{P}{}{}{}{}\languagesprodnewline
\languagesprodline{|}{\languagesnt{tprog\_judge}}{}{}{}{}\languagesprodnewline
\languagesprodline{|}{\languagesnt{tyqclosure\_axiom\_tgt}}{}{}{}{}\languagesprodnewline
\languagesprodline{|}{\languagesnt{arrows}}{}{}{}{}\languagesprodnewline
\languagesprodline{|}{\languagesnt{terminals}}{}{}{}{}}

\newcommand{\languagesgrammar}{\languagesgrammartabular{
\languagesNat\languagesinterrule
\languagesdeclvarXXseq\languagesinterrule
\languagesO\languagesinterrule
\languagesOXXapply\languagesinterrule
\languagesG\languagesinterrule
\languagesGXXapply\languagesinterrule
\languagesunitary\languagesinterrule
\languagesqidx\languagesinterrule
\languagesqidxXXseq\languagesinterrule
\languagesqidxXXsubstXXinner\languagesinterrule
\languagesqidxXXsubst\languagesinterrule
\languagesvar\languagesinterrule
\languagesvarXXseq\languagesinterrule
\languagesvarXXsubst\languagesinterrule
\languagessty\languagesinterrule
\languagessfunty\languagesinterrule
\languagesstupleXXty\languagesinterrule
\languagessdecl\languagesinterrule
\languagessdecls\languagesinterrule
\languagessrcXXprogram\languagesinterrule
\languagesse\languagesinterrule
\languagessval\languagesinterrule
\languagessrcXXjudge\languagesinterrule
\languagessdeclsXXjudge\languagesinterrule
\languagessprogXXjudge\languagesinterrule
\languagestyqclosureXXaxiom\languagesinterrule
\languageswfXXCG\languagesinterrule
\languagestdecl\languagesinterrule
\languagestdecls\languagesinterrule
\languagestdeclXXjudge\languagesinterrule
\languagestty\languagesinterrule
\languagesttupleXXty\languagesinterrule
\languageste\languagesinterrule
\languagestval\languagesinterrule
\languagesC\languagesinterrule
\languagestfunty\languagesinterrule
\languagestjudge\languagesinterrule
\languagestgtXXprogram\languagesinterrule
\languagestprogXXjudge\languagesinterrule
\languagestyqclosureXXaxiomXXtgt\languagesinterrule
\languagesarrows\languagesinterrule
\languagesterminals\languagesinterrule
\languagesformula\languagesinterrule
\languagesjudgement\languagesinterrule
\languagesuserXXsyntax\languagesafterlastrule
}}

\newcommand{\languagesdefnss}{
}

\newcommand{\languagesall}{\languagesmetavars\\[0pt]
\languagesgrammar\\[5.0mm]
\languagesdefnss}

%% file: introduction.tex
\section{Introduction}\label{sec:introduction}

In past years, quantum computation has received attention for quantum speedups of some algorithms,
and quantum computing technologies have advanced rapidly.
The \emph{quantum circuit model} is a fundamental model of quantum computation.
A quantum circuit is described by a sequence of quantum gates and measurements applied to qubits,
just as classical circuits are built by a sequence of logical gates applied to classical bits.
Many toolkits and quantum programming languages have been presented to construct quantum circuits.

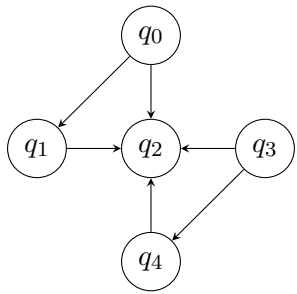
\begin{wrapfigure}[12]{l}[0mm]{70mm}
  \centering
  \begin{tikzpicture}
    \node[shape=circle, draw] (q0) at (0, 1.5) {$q_0$};
    \node[shape=circle, draw] (q1) at (-1.5, 0) {$q_1$};
    \node[shape=circle, draw] (q2) at (0, 0) {$q_2$};
    \node[shape=circle, draw] (q3) at (1.5, 0) {$q_3$};
    \node[shape=circle, draw] (q4) at (0, -1.5) {$q_4$};

    \draw [->, >=stealth] (q0) -- (q1);
    \draw [->, >=stealth] (q0) -- (q2);
    \draw [->, >=stealth] (q1) -- (q2);
    \draw [->, >=stealth] (q3) -- (q2);
    \draw [->, >=stealth] (q4) -- (q2);
    \draw [->, >=stealth] (q3) -- (q4);
  \end{tikzpicture}
  \caption{The coupling graph of IBM QX2}
  \label{fig:qx2-coupling-graph}
\end{wrapfigure}

Some near-term quantum computers have several architectural restrictions, such as the connectivity between qubits.
The connectivity constraints prevent us from applying a two-qubit gate to an arbitrary pair of qubits.
These constraints are represented by a \emph{coupling graph} whose nodes correspond to qubits.
We can apply two-qubit gates to a pair of qubits only when the pair is directly connected in the coupling graph.
For example, the IBM QX2 quantum computer~\cite{qiskitteamIBMQXBackend} has a graph given in \cref{fig:qx2-coupling-graph}.
This graph illustrates that the IBM QX2 computer has five qubits, and we can, for example, apply a two-qubit gate to $(q_0, q_1)$
but not to $(q_1, q_4)$ or $(q_1, q_0)$.

Due to the connectivity constraints, we often have to transform a quantum program to satisfy them.
Since this transformation may need additional gates, it is desirable to minimize the number of them.
This optimization problem is called \emph{Qubit Allocation Problem}, and many studies have been conducted~\cite{siraichiQubitAllocation2018,siraichiQubitAllocationCombination2019,dengCodarContextualDurationAware2020,nishioExtractingSuccessIBM2020,muraliFormalConstraintbasedCompilation2019a}.

These studies, however, formalize qubit allocation as a transformation of a low-level quantum program.
In the future, we would like to write quantum programs in high-level quantum programming languages
equipped with (recursive) functions, loops, and branching instructions.
It is not trivial to check statically whether a high-level quantum program satisfies connectivity constraints
because the qubit data flows in them are complex.
Thus, we need to develop a static verification method for connectivity constraints
and a correct qubit allocation algorithm for high-level quantum programs.

\subsection*{Contribution}
This paper introduces a type-based framework of qubit allocation for a first-order quantum programming language.
\cref{fig:workflow} gives the workflow of our framework, which consists of a source language, a target language,
and a qubit allocation algorithm.
We briefly explain the source language, the target language, and the algorithm in this order.

\begin{figure}[tb]
  \centering
  \includegraphics[width=14cm]{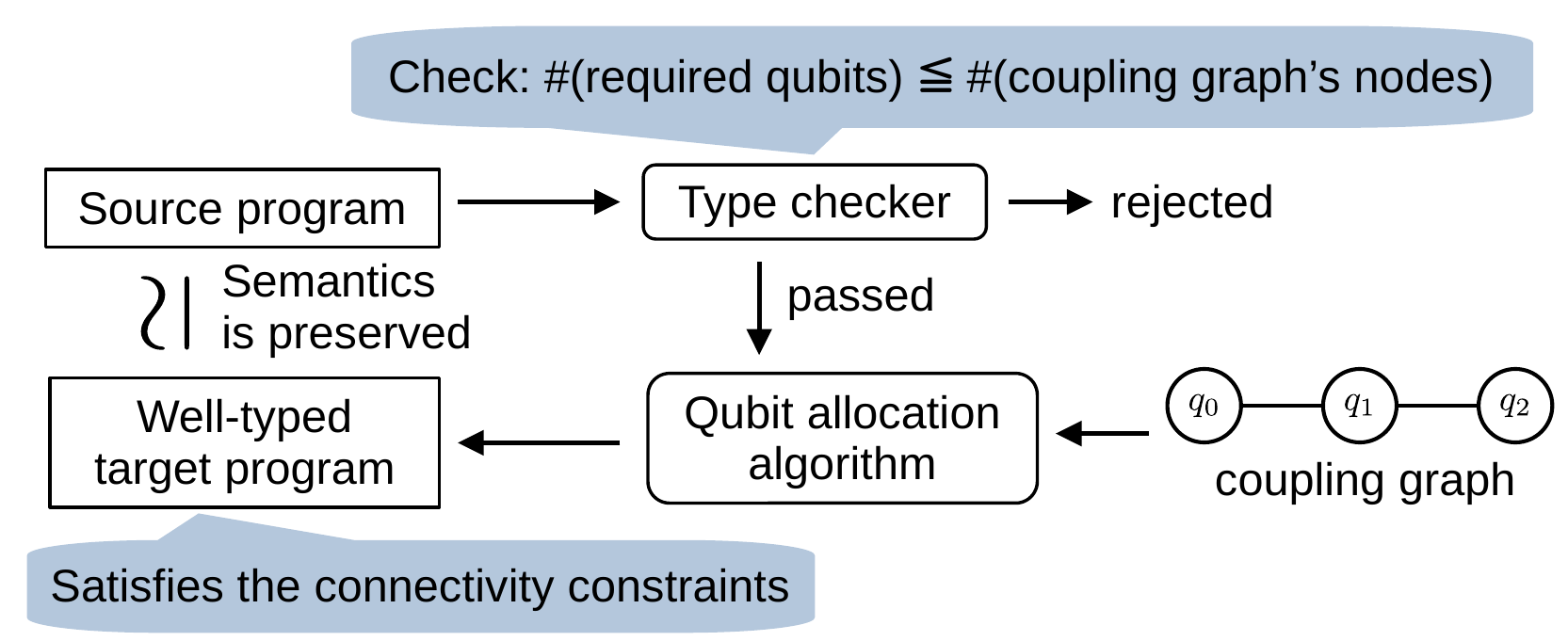}
  \caption{The workflow of our framework}
  \label{fig:workflow}
\end{figure}

We write a quantum program in the source language in the first step.
We design the source language and its linear type system~\cite{turnerOnceType1995}
to estimate the number of qubits required for a given program to run.
If the number of qubits required by the program exceeds the size of nodes of a coupling graph,
the program cannot run on the architecture even though we ignore the connectivity constraints.
Thus, we have to write a well-typed program in the source language before qubit allocation.
We prove the correctness of the estimation as type soundness.

The target language statically checks whether a given program satisfies the connectivity constraints or not.
It uses qualified types~\cite{jonesTheoryQualifiedTypes1992} for the connectivity verification.
We prove the type soundness that a well-typed program does not violate connectivity constraints.

The qubit allocation algorithm takes as input a well-typed program written in the source language
and transforms it so that the resulting program satisfies connectivity constraints.
In other words, the algorithm has to output a well-typed program because of the type soundness of the target language.
We design our allocation algorithm always to obtain a well-typed target program.
We prove the type-preserving property of our algorithm.
Furthermore, we conjecture the semantic-preserving property of our algorithm, whose proof is left as future work.

A type-based approach has two merits compared to the other approaches.
First, it provides a precise and intuitive formalization of correct qubit allocation for high-level quantum programming language.
Second, we can use the type system of the target language as a verifier for connectivity constraints.

\subsection*{The Organization of This Paper}
The rest of this paper is organized as follows;
\cref{sec:preliminaries} explains a basic of quantum computation and qubit allocation.
\cref{sec:src-lang} defines the source language, its type system, and sementics.
\cref{sec:tgt-lang} defines the target language, its type system and semantics.
\cref{sec:algorithm} formalizes qubit allocation as a type-preserving translation from the source language to the target language.
\cref{sec:related-work} discusses related work.
\cref{sec:conclusion} gives concluding remarks and discusses future work.

%% file: background.tex
\section{Preliminaries}\label{sec:preliminaries}

This section explains a basic of quantum computation and qubit allocation briefly.
For a more detailed introduction of these topics, please see~\cite{nielsenQuantumComputationQuantum2011}.

\subsection{Quantum Computation}
In the quantum computation, a quantum bit or \emph{qubit} is the unit of information.
A \textit{quantum state} of a qubit is a normalized vector of the 2-dimentional Hilbert space $\mathbb{C}^2$.
We use \emph{Dirac notation}, which encloses a variable with $|$ and $\rangle$, to denote quantum states.
$\{\ket{0}, \ket{1}\}$ denotes the standard basis vectors of $\mathbb{C}^2$.
We describe the state of a single qubit in the form of
$\alpha\ket{0} + \beta\ket{1}$ with $\alpha, \beta \in \mathbb{C}$ and $|\alpha|^2 + |\beta|^2 = 1$.
The state of $n$-qubits is a normalized vector of $\bigotimes_{i=1}^n \mathbb{C}^2 \cong \mathbb{C}^{2^n}$.
For example, if $\ket{\psi} = \ket{0}, \ket{\phi} = \frac{1}{\sqrt{2}}(\ket{0} + \ket{1})$,
then $\ket{\psi\phi} = \ket{\psi}\ket{\phi} = \frac{1}{\sqrt{2}}(\ket{00} + \ket{01})$.
We call a quantum state $\ket{\psi} \in \mathbb{C}^2$ \emph{pure state}.
On the other hand, when we do not know the state completely,
but know that it is one of the pure states $\{\ket{\psi_i}\}$ with probabilities $p_i$,
we call the state $\{p_i, \ket{\psi_i}\}$ \emph{mixed state}.
A mixed state is described by a \emph{density operator} or a \emph{density matrix}
$\rho = \sum_i p_i \ket{\psi_i}\bra{\psi_i}$.
Here, $\bra{\psi_i}$ is an adjoint of $\ket{\psi_i}$.
We often use a density operator to describe quantum states because it can uniformly express both pure and mixed states.

We change a quantum state by unitary operators called \emph{quantum gates}.
Quantum programming languages provide a certain set of quantum gates, such as the Hadamard gate and the CNOT (controlled not) gate:
\begin{gather*}
  H\ket{x} = \frac{1}{\sqrt{2}}(\ket{0} + (-1)^x\ket{1}), \andalso \CNOT \ket{x}\ket{y} = \ket{x}\ket{x \oplus y},
\end{gather*}
where $x, y \in \{0, 1\}$ and $\oplus$ denotes XOR operation.

We get the result of quantum computation by performing a \emph{quantum measurement}.
In general, quantum measurements consist of projections $M_1, M_2, \dots, M_n$.
When performing measurements to a quantum state $\rho$,
we get an outcome corresponding one of measurements $M_i$ with probability $p_i = \Tr(M_i^\dagger M_i\rho)$
and then the quantum state is changed to $\frac{M_i\rho M_i^\dagger}{p_i}$.

\subsection{Qubit Allocation Problem}

As described in \cref{sec:introduction}, we can apply a CNOT gate to a pair of adjacent qubits on a coupling graph.
In other words, a given program can run safely when the coupling graph has a subgraph isomorphic to a graph whose edge represents a pair of qubits to which the CNOT gate is applied in the program.

The \emph{swap gate}, which consists of three CNOT gates (\cref{fig:swap-gate}), plays an important role when we cannot find an isomorphic graph in a coupling graph.
Since the swap gate swaps the states of adjacent qubits, when we want to apply a CNOT gate to distant qubits,
we move one of them to an adjacent node to another qubit by applying several swap gates.
\cref{fig:example-of-inserting-swaps} illustrates these swap insertions.
In this example, we consider a path graph with three nodes as a coupling graph,
and we try to apply a CNOT gate to qubits $q_0$ and $q_2$ separated by a distance of two.
We can achieve this in three steps:
(1) swap the states of $q_0$ and $q_1$,
(2) apply a CNOT gate to $q_1$ and $q_2$,
(3) swap the states of $q_0$ and $q_1$ again.
Note that the swap gate itself requires adjacent arguments because it is composed of CNOT gates.
\begin{figure}[htbp]
  \begin{tabular}{cc}
    \begin{minipage}{0.4\hsize}
      \leavevmode\centering
      \Qcircuit @C=1em @R=.7em {
        \lstick{} & \qswap     & \qw & \raisebox{-1.7em}{$\equiv$} & & \ctrl{1}  & \targ     & \ctrl{1} & \qw \\
        \lstick{} & \qswap\qwx & \qw &                             & & \targ     & \ctrl{-1} & \targ    & \qw
      }
      \caption{The swap gate}
      \label{fig:swap-gate}
    \end{minipage}

    \begin{minipage}{0.6\hsize}
      \leavevmode\centering
      \Qcircuit @C=1em @R=.7em {
        \lstick{} & \ctrl{1} & \qw &
          \raisebox{-2.2em}{$\equiv$} &
        & \gate{H} & \targ     & \gate{H} & \qw  \\
        \lstick{} & \targ    & \qw &
        &
        & \gate{H} & \ctrl{-1} & \gate{H} & \qw
      }
      \caption{The CNOT gate in the reverse direction}
      \label{fig:rev-cnot}
    \end{minipage}
  \end{tabular}
\end{figure}

\begin{figure}[htbp]
  \begin{tabular}{cc}
    \begin{minipage}{0.4\hsize}
      \centering
      \begin{tikzpicture}
        \node[shape=circle, draw] (q0) at (0, 0) {$q_0$};
        \node[shape=circle, draw] (q1) at (1.5, 0) {$q_1$};
        \node[shape=circle, draw] (q2) at (3.0, 0) {$q_2$};
        \draw [>=stealth] (q0) -- (q1);
        \draw [>=stealth] (q1) -- (q2);
      \end{tikzpicture}
    \end{minipage}
    \begin{minipage}{0.6\hsize}
      \leavevmode\centering
      \Qcircuit @C=1em @R=.7em {
        \lstick{\ket{q_0}} & \ctrl{2} & \qw & \raisebox{-2.2em}{$\rightarrow$} & & \qswap      & \qw      & \qswap     & \qw \\
        \lstick{\ket{q_1}} & \qw      & \qw &                                  & & \qswap\qwx  & \ctrl{1} & \qswap\qwx & \qw \\
        \lstick{\ket{q_2}} & \targ    & \qw &                                  & & \qw         & \targ    & \qw        & \qw
      }
    \end{minipage}
  \end{tabular}
  \caption{Applying a CNOT gate to distant qubits $q_0$ and $q_2$}
  \label{fig:example-of-inserting-swaps}
\end{figure}
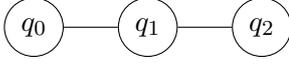

A coupling graph $G$ is a directed graph.
For example, if $(q_1, q_2) \in G$ and $(q_2, q_1) \not\in G$, then we can apply a CNOT gate to $(q_1, q_2)$ but not $(q_2, q_1)$.
However, we can reversely implement a CNOT gate with four additional Hadamard gates (\cref{fig:rev-cnot}).
Therefore, as long as we do not care about the cost of the Hadamard gate,
we can consider a coupling graph as an undirected graph without loss of generality.
In actual quantum computers, the cost of the Hadamard gate is cheaper than that of the CNOT gate.

%% file: srclang.tex
\section{Source Language}\label{sec:src-lang}

This section defines the source language equipped with first-order (recursive) functions, if-expression,
and instructions for creating and discarding a qubit.
The goal of the source language is to estimate the number of qubits required for a program to run
because quantum computers have a limited number of qubits.
We achieve this by using linear types~\cite{turnerOnceType1995}.
Note that the source language does not take care of connectivity constraints.


\subsection{Syntax}

\begin{definition}{(Syntax of the source language)}
  The source language is defined by the grammer:
  \begin{align*}
    \languagesmv{d}\ (\text{function definition})
      &\Coloneqq  \languagesmv{f}  \mapsto ( \languagesmv{x_{{\mathrm{1}}}}  \languagessym{,} \, .. \, \languagessym{,}  \languagesmv{x_{\languagesmv{n}}} )  e  \\
    e\ (\text{expression})
      &\Coloneqq \languagessym{(}  \languagesmv{x_{{\mathrm{1}}}}  \languagessym{,} \, .. \, \languagessym{,}  \languagesmv{x_{\languagesmv{n}}}  \languagessym{)} \pipe \languageskw{let} \, \languagesmv{x}  \languagessym{=} \, \languageskw{init} \, \languagessym{()} \, \languageskw{in} \, e \pipe \languageskw{discard} \, \languagesmv{x}  \languagessym{;}  e \\
      &\quad \pipe \languageskw{let} \, \languagessym{(}  \languagesmv{x_{{\mathrm{1}}}}  \languagessym{,}  \languagesmv{x_{{\mathrm{2}}}}  \languagessym{)}  \languagessym{=} \, \languageskw{cnot} \, \languagessym{(}  \languagesmv{y_{{\mathrm{1}}}}  \languagessym{,}  \languagesmv{y_{{\mathrm{2}}}}  \languagessym{)} \, \languageskw{in} \, e \pipe \languageskw{let} \, \languagessym{(}  \languagesmv{x_{{\mathrm{1}}}}  \languagessym{,} \, .. \, \languagessym{,}  \languagesmv{x_{\languagesmv{n}}}  \languagessym{)}  \languagessym{=}  \languagesmv{f}  \languagessym{(}  \languagesmv{y_{{\mathrm{1}}}}  \languagessym{,} \, .. \, \languagessym{,}  \languagesmv{y_{\languagesmv{m}}}  \languagessym{)} \, \languageskw{in} \, e \\
      &\quad \pipe \languageskw{let} \, \languagessym{(}  \languagesmv{x_{{\mathrm{1}}}}  \languagessym{,} \, .. \, \languagessym{,}  \languagesmv{x_{\languagesmv{n}}}  \languagessym{)}  \languagessym{=}  e_{{\mathrm{1}}} \, \languageskw{in} \, e_{{\mathrm{2}}} \pipe \languageskw{if} \, \languagesmv{x} \, \languageskw{then} \, e_{{\mathrm{1}}} \, \languageskw{else} \, e_{{\mathrm{2}}} \\
    P\ (\text{program})
      &\Coloneqq  \braket{  \languagessym{\{}  \languagesmv{d_{{\mathrm{1}}}}  \languagessym{,} \, .. \, \languagessym{,}  \languagesmv{d_{\languagesmv{n}}}  \languagessym{\}}  ,  e  } 
  \end{align*}
\end{definition}

The expression $\languageskw{let} \, \languagesmv{x}  \languagessym{=} \, \languageskw{init} \, \languagessym{()} \, \languageskw{in} \, e$ creates a qubit in a basis state $\ket{0}$, binds $x$ to it and runs $e$.
In contrast, the expression $\languageskw{discard} \, \languagesmv{x}  \languagessym{;}  e$ discards qubit $\languagesmv{x}$.
The expression $\languageskw{let} \, \languagessym{(}  \languagesmv{x_{{\mathrm{1}}}}  \languagessym{,}  \languagesmv{x_{{\mathrm{2}}}}  \languagessym{)}  \languagessym{=} \, \languageskw{cnot} \, \languagessym{(}  \languagesmv{y_{{\mathrm{1}}}}  \languagessym{,}  \languagesmv{y_{{\mathrm{2}}}}  \languagessym{)} \, \languageskw{in} \, e$ applies the CNOT gate to a pair of qubits $\languagesmv{y_{{\mathrm{1}}}}$ and $\languagesmv{y_{{\mathrm{2}}}}$.
The function call expression $\languageskw{let} \, \languagessym{(}  \languagesmv{x_{{\mathrm{1}}}}  \languagessym{,} \, .. \, \languagessym{,}  \languagesmv{x_{\languagesmv{n}}}  \languagessym{)}  \languagessym{=}  \languagesmv{f}  \languagessym{(}  \languagesmv{y_{{\mathrm{1}}}}  \languagessym{,} \, .. \, \languagessym{,}  \languagesmv{y_{\languagesmv{m}}}  \languagessym{)} \, \languageskw{in} \, e$ calls the function $\languagesmv{f}$ with arguments $\languagesmv{y_{{\mathrm{1}}}}  \languagessym{,} \, .. \, \languagessym{,}  \languagesmv{y_{\languagesmv{m}}}$
and binds $\languagesmv{x_{{\mathrm{1}}}}  \languagessym{,} \, .. \, \languagessym{,}  \languagesmv{x_{\languagesmv{n}}}$ to the results of $f(y_1, \dots, y_m)$.
The if-expression $\languageskw{if} \, \languagesmv{x} \, \languageskw{then} \, e_{{\mathrm{1}}} \, \languageskw{else} \, e_{{\mathrm{2}}}$ performs measurements on $\languagesmv{x}$,
and then chooses the subsequent expression $e_{{\mathrm{1}}}$ or $e_{{\mathrm{2}}}$ depending on the outcome of the measurement.

The form of a function definition is $ \languagesmv{f}  \mapsto ( \languagesmv{x_{{\mathrm{1}}}}  \languagessym{,} \, .. \, \languagessym{,}  \languagesmv{x_{\languagesmv{n}}} )  e $, where $f$ is a function name, $\languagesmv{x_{{\mathrm{1}}}}  \languagessym{,} \, .. \, \languagessym{,}  \languagesmv{x_{\languagesmv{n}}}$ are arguments, and $e$ is its body expression.
A program $ \braket{  \languagessym{\{}  \languagesmv{d_{{\mathrm{1}}}}  \languagessym{,} \, .. \, \languagessym{,}  \languagesmv{d_{\languagesmv{n}}}  \languagessym{\}}  ,  e  } $ consists of function definitions $d_1, \dots, d_n$ and the entry point $e$.
In the $i$-th function definition $\languagesmv{d_{\languagesmv{i}}} =  \languagesmv{f_{\languagesmv{i}}}  \mapsto ( \languagesmv{x_{{\mathrm{1}}}}  \languagessym{,} \, .. \, \languagessym{,}  \languagesmv{x_{\languagesmv{n}}} )  e_{\languagesmv{i}} $,
the expression $e_{\languagesmv{i}}$ can call previously defined functions $f_1, f_2, \dots, f_{i-1}$ and $\languagesmv{f_{\languagesmv{i}}}$ itself recursively.

For simplicity, we omit classical data and useful unitary gates.
We can easily extend the language with these primitives because they do not affect qubit allocation.
We also omit mutual recursion, but we can add it by modifying the type system a little.

\begin{definition}{(Types and contexts)}
  \begin{align*}
    \tau\ (\text{simple type})
      &\Coloneqq \languageskw{qbit} \\
    T\ (\text{tuple type})
      &\Coloneqq \tau_{{\mathrm{1}}}  \languagessym{*} \, .. \, \languagessym{*}  \tau_{\languagesmv{n}} \\
    \theta\ (\text{function type})
      &\Coloneqq  T_{{\mathrm{1}}}  \xrightarrow{ N }  T_{{\mathrm{2}}}  \\
    \Theta\ (\text{function environment})
      &\Coloneqq \emptyset \pipe \Theta  \languagessym{,}  \languagesmv{f}  \languagessym{:}  \theta \\
    \Gamma\ (\text{type environment})
      &\Coloneqq  \emptyset  \pipe \Gamma  \languagessym{,}  \languagesmv{x}  \languagessym{:}  \tau 
  \end{align*}
\end{definition}

A function type $ T_{{\mathrm{1}}}  \xrightarrow{ N }  T_{{\mathrm{2}}} $ indicates that the argument type is $T_{{\mathrm{1}}}$,
the function generates $N$ qubits internally, and the type of the return value is $T_{{\mathrm{2}}}$.
A function environment $\Theta$ has pairs of a function name and its type.
We assume that all variables in $\Gamma$ are pairwise distinct.
We can consider all the contexts to be set.


\subsection{Typing Rules}

\begin{definition}{(Judgments)}
  The type system of the source language has the following judgments:
  \begin{itemize}
  \item (Typing judgment) $\STJudgeExp$,
  \item (Well-typed function definitions) $\STJudgeFunDef$, and
  \item (Well-typed program) $\SJudgeProg$.
  \end{itemize}
\end{definition}

The typing judgment $\STJudgeExp$ indicates that
$e$ is given type $T$ under the function type environment $\Theta$, the number of free qubits $N$, and typing context $\Gamma$.

$\STJudgeFunDef$ means that for all function definitions $\languagesmv{d_{\languagesmv{i}}} \in D$ is well-typed.

$\SJudgeProg$ indicates that the function definitions $D$ is well-typed under some $\Theta$,
and the expression $e$ is well-typed under the function definitions $D$ and $N$.

\begin{figure}[htbp]

\fbox{$\STJudgeExp$}

  \STReturn

  \STInit

  \STDiscard

  \STLet

  \STCnot

  \STIf

  \STCall

\fbox{$\Theta  \vdash  D$}

  \STFunDef

  \begin{tabular}{c}
    \begin{minipage}{0.3\hsize}
      \infrule[T-Empty]{
      }{
        \emptyset  \vdash  \emptyset
      }
    \end{minipage}
    \begin{minipage}{0.7\hsize}
      \infrule[T-FunDecl]{
        \Theta  \vdash  D
        \andalso \Theta  \languagessym{,}  \languagesmv{f}  \languagessym{:}  \theta  \vdash   \languagesmv{f}  \mapsto ( \languagesmv{x_{{\mathrm{1}}}}  \languagessym{,} \, .. \, \languagessym{,}  \languagesmv{x_{\languagesmv{n}}} )  e 
      }{
        \Theta  \languagessym{,}  \languagesmv{f}  \languagessym{:}  \theta  \vdash  D  \languagessym{,}   \languagesmv{f}  \mapsto ( \languagesmv{x_{{\mathrm{1}}}}  \languagessym{,} \, .. \, \languagessym{,}  \languagesmv{x_{\languagesmv{n}}} )  e 
      }
    \end{minipage}
  \end{tabular}

\fbox{$\SJudgeProg$}

  \STProg

  \caption{Typing judgments of the source language}
  \label{fig:src-typing-rules}
\end{figure}

The typing rules are given in \cref{fig:src-typing-rules}.
We explain some of them.

In (\textsc{T-Return}), we must return all variables in the type environment.
That is, we cannot discard quantum variables implicitly.

In (\textsc{T-Init}), a quantum variable $\languagesmv{x}$ is introduced to the type environment,
and then $N$ is reduced by 1.
In (\textsc{T-Discard}), a quantum variable $\languagesmv{x}$ is discarded and $N$ is increased by 1.

In (\textsc{T-Call}), the function $f$ requires that the caller has enough qubits to call $f$ ($N \geq N'$).
After calling $f$, the number of surviving qubits can be calculated by checking the type of $f$ because of linear types.

The rule (\textsc{T-FunDef}) judges whether the body expression of a function is well-typed.
If the expression is well-typed, then its function definition is well-typed.

Our language prohibits variable shadowing because it causes implicit discarding of qubits.
For example, the rule (\textsc{T-Init}) requires that the introduced variable $x$ does not appear in the current typing context $x \not\in \dom(\Gamma)$.

As shown in (\textsc{T-Empty}) and (\textsc{T-FunDecl}), a well-formed function type environment
comprises a sequence of well-typed function definitions.
The type system constructs a well-formed function definitions $d_1, \dots, d_n$ in this order.


\subsection{Operational Semantics}

\begin{definition}{(Runtime state)}
  A runtime state is represented by a triple $\braket{X, \rho, e}$,
  which consists of a set of free qubit variables, a density matrix, and a currently reducing expression.
\end{definition}

We label a density operator with variables appearing in $e$ for convenience.
This notation allows us to easily describe a quantum gate application to specific qubits.
We write $\Var(\rho)$ for the set of labels in $\rho$.
Basically, the equation $X = \Var(\rho)$ holds.

\begin{definition}{(Operational semantics)}
  The operational semantics of the source language is defined as a transition relation $ \rightarrow _{ D } $ on runtime states,
  where $D$ denotes function definitions.
  We write $ \rightarrow _{ D } ^*$ for the transitive and reflexive closure of $ \rightarrow _{ D } $.
  Each rule is displayed in \cref{fig:src-operational-semantics}.
\end{definition}

In (\textsc{E-Init}), an expression consumes one qubit $\languagesmv{x'} \in X$.
The program gets stuck if all the qubits are used ($X = \emptyset$).

In (\textsc{E-Discard}), it resets the state of $x$ and returns $x$ to $X$.
Here, $\ketbra{0}{b}{x}$ is the abbreviation of $\ketbra{0}{b}{x} \otimes I_{\tilde{x}}$,
where $\tilde{x}$ means all the qubits except $x$.
The resetting operation is not unitary,
and thus, strictly speaking, we have to reset a state by \emph{uncomputation} composed of unitary operations.
We assume that we can reset a quantum state in one step because we want to focus on qubit allocation in this paper.

In (\textsc{E-IfTrue}) and (\textsc{E-IfFalse}), it first performs \emph{z-basis measurements} $\{M_s\}_{s = 0, 1}$,
where $M_s = (I + (-1)^sZ)/2$, and $Z\ket{b} = (-1)^b \ket{b}\ (b \in \{0, 1\})$.
The operator $Z_x$ means that it manipulates the quantum state of the variable $x$.
Then it evaluates either $e_1$ or $e_2$ depending on the outcome of the measurements.

\begin{figure}[tb]
\fbox{$\SETrans$}

  \SELetI

  \SELetII

  \SEInit

  \SEDiscard

  \SECnot

  \SECall

  \SEIfTrue

  \SEIfFalse

  \caption{The operational semantics of the source language}
  \label{fig:src-operational-semantics}
\end{figure}


\subsection{Type Soundness}

The type system can estimate the number of qubits required by a given program via $N$.
We show the correctness of this estimation as type soundness
by combining the progress lemma and the subject reduction lemma.
The detailed proofs are in \cref{sec:src-proofs}.

\begin{theorem}{(Type soundness)}\label{thm:src-soundness}
  If $N  \vdash   \braket{  D  ,  e  } $ and $|X| \geq N$,
  then $\qclosure{X, \rho, e}$ does not reach a stuck state.
  In other words, it does not fail to allocate a qubit in (\textsc{E-Init}) during its execution.
\end{theorem}

This theorem states that if a quantum computer architecture $A$ has the coupling graph $G$ and $|V(G)| \geq |N|$,
then $e$ can run safely on $A$ ignoring connectivity constraints.

%% file: tgtlang.tex
\section{Target Language}\label{sec:tgt-lang}

Next, we define the target language, which aims to ensure that
a well-typed target program does not violate connectivity constraints during its execution.
We achieve this goal by using qualified types~\cite{jonesTheoryQualifiedTypes1992}.
Since the definition of the target language is similar to the source language,
we show the main differences between them.

\begin{definition}{(Syntax, contexts and types)}
  The target language is defined by the grammar:
  \begin{align*}
    \alpha  \languagessym{,}  \beta   &\in \mathit{Qidx} \\
    e     &\Coloneqq \languageskw{init} \, \languagesmv{x}  \languagessym{;}  e \pipe \languageskw{let} \, \languagessym{(}  \languagesmv{x_{{\mathrm{1}}}}  \languagessym{,}  \languagesmv{x_{{\mathrm{2}}}}  \languagessym{)}  \languagessym{=} \, \languageskw{swap} \, \languagessym{(}  \languagesmv{y_{{\mathrm{1}}}}  \languagessym{,}  \languagesmv{y_{{\mathrm{2}}}}  \languagessym{)} \, \languageskw{in} \, e \\
               &\quad \pipe \languagessym{(}  \languagesmv{x_{{\mathrm{1}}}}  \languagessym{,} \, .. \, \languagessym{,}  \languagesmv{x_{\languagesmv{n}}}  \languagessym{)} \pipe \languageskw{let} \, \languagessym{(}  \languagesmv{x_{{\mathrm{1}}}}  \languagessym{,}  \languagesmv{x_{{\mathrm{2}}}}  \languagessym{)}  \languagessym{=} \, \languageskw{cnot} \, \languagessym{(}  \languagesmv{y_{{\mathrm{1}}}}  \languagessym{,}  \languagesmv{y_{{\mathrm{2}}}}  \languagessym{)} \, \languageskw{in} \, e \pipe \languageskw{let} \, \languagessym{(}  \languagesmv{x_{{\mathrm{1}}}}  \languagessym{,} \, .. \, \languagessym{,}  \languagesmv{x_{\languagesmv{n}}}  \languagessym{)}  \languagessym{=}  \languagesmv{f}  \languagessym{(}  \languagesmv{y_{{\mathrm{1}}}}  \languagessym{,} \, .. \, \languagessym{,}  \languagesmv{y_{\languagesmv{m}}}  \languagessym{)} \, \languageskw{in} \, e \\
               &\quad \pipe \languageskw{let} \, \languagessym{(}  \languagesmv{x_{{\mathrm{1}}}}  \languagessym{,} \, .. \, \languagessym{,}  \languagesmv{x_{\languagesmv{n}}}  \languagessym{)}  \languagessym{=}  e_{{\mathrm{1}}} \, \languageskw{in} \, e_{{\mathrm{2}}} \pipe \languageskw{if} \, \languagesmv{x} \, \languageskw{then} \, e_{{\mathrm{1}}} \, \languageskw{else} \, e_{{\mathrm{2}}} \\
    \tau    &\Coloneqq  \texttt{q}( \alpha )  \\
    \Phi      &\Coloneqq  \emptyset  \pipe \Phi  \languagessym{,}   \alpha_{{\mathrm{1}}}  \sim  \alpha_{{\mathrm{2}}}  \\
    \theta &\Coloneqq  \forall   \overline{ \alpha }   .   \Phi  \Rightarrow  T_{{\mathrm{1}}}  \rightarrow  T_{{\mathrm{2}}}  
  \end{align*}
\end{definition}

The main difference is that the qubit types are indexed by \emph{qidx}, which corresponds to a node of a coupling graph.
Qidxs allows us to analyze the data flows of qubits and connections between qubits in the type system, as described later.

Unlike the source language, all qubits are allocated statically like classical registers.
Hence, the expression $\languageskw{init} \, \languagesmv{x}  \languagessym{;}  e$ resets the state of $\languagesmv{x}$ instead of creating a qubit.
We introduce the swap gate as a primitive gate for simplicity, which three CNOT gates can implement.

$ \alpha_{{\mathrm{1}}}  \sim  \alpha_{{\mathrm{2}}} $ indicates the connectivity between qubits associated with qidx $\alpha_{{\mathrm{1}}}, \alpha_{{\mathrm{2}}}$.
A connectivity context $\Phi$ has the connections between qidxs.
We can use qidxs and a connectivity context to represent a coupling graph.

A function may call a CNOT gate in its body, so it requires that its arguments satisfy the connectivity constraints.
We add the required connectivity constraints $\Phi$ to the function type $ \forall   \overline{ \alpha }   .   \Phi  \Rightarrow  T_{{\mathrm{1}}}  \rightarrow  T_{{\mathrm{2}}}  $.
When calling a function $f$ whose type has $\Phi$, a caller must arrange arguments to meet $\Phi$.
Qidxs are polymorphic in function types to enhance the reusability of functions.

A function environment and a type environment are the same as the source language's environments.


\subsection{Typing Rules and Operational Semantics}

The typing judgment has the form $\TTJudgeExp$.
In the typing judgment, the connectivity constraints $\Phi$ restrict applying the two-qubit gates.
In other words, the expression $e$ must satisfy the connectivity constraints $\Phi$ to be well-typed.

The typing rules are given in \cref{fig:tgt-typing-rules}.
As seen in (\textsc{T-Cnot}), we can apply a CNOT gate only to directly connected qubits in $\Phi$.
Similarly, in (\textsc{T-Call}), the current connectivity context must meet the connectivity constraints appearing in the type of a called function.

In (\textsc{T-FunDef}), the target language allows polymorphic recursion. 
The other rules are almost the same as the source language.

\begin{figure}[tbp]
  \fbox{$\TTJudgeExp$}

  \TTInit 

  \TTCnot

  \TTCall

  \fbox{$\TTJudgeFunDef$}

  \TTFunDef

  \TTProg

  \caption{Typing judgments of the target language (excerpt)}
  \label{fig:tgt-typing-rules}
\end{figure}

Next, we define the operational semantics of the target language.
In the target language, the runtime state is represented by $\qclosure{\rho, e}$
because the number of qubits does not change in reducing an expression.

\begin{definition}{(Operational semantics)}
  The operational semantics of the target language is defined as a transition relation $ \rightarrow _{D, G}$ on runtime states,
  where $D$ and $G$ denote function definitions and a coupling graph respectively.
  Transition rules are displayed in \cref{fig:tgt-operational-semantics}.
\end{definition}

In (\textsc{E-Cnot}), the reduction proceeds when the arguments are connected in a coupling graph.
A program that does not satisfy connectivity constraints may reach a stuck state.
The transition rule for the swap gate is almost the same as (\textsc{E-Cnot}).
 
\begin{figure}[tb]
  \fbox{$\TETrans$}

  \TEInit

  \TECnot

  \caption{The operational semantics of the target language (excerpt)}
  \label{fig:tgt-operational-semantics}
\end{figure}


\subsection{Type Soundness}

To conclude this section, we show that a well-typed program satisfies the connectivity constraints.
We must prove the type soundness that a well-typed program does not reach a stuck state.
The proof is easily done by the progress lemma and the subject reduction lemma similarly to the source language.
For more detailed proof, please see \cref{sec:appendix-tgt}.

\begin{theorem}{(Type soundness)}\label{thm:tgt-soundness}
  If $\Phi  \pipe  \Gamma  \vdash   \braket{  D  ,  e  } $, then $\qclosure{\rho, e}$ does not reach a stuck state on a coupling graph $G$
  such that $\dom(\Gamma) \subseteq V(G)$ and $\languagesmv{x_{{\mathrm{1}}}}  \languagessym{:}   \texttt{q}( \alpha_{{\mathrm{1}}} )   \languagessym{,}  \languagesmv{x_{{\mathrm{2}}}}  \languagessym{:}   \texttt{q}( \alpha_{{\mathrm{2}}} )  \in \Gamma \land  \alpha_{{\mathrm{1}}}  \sim  \alpha_{{\mathrm{2}}}  \Rightarrow (x_1, x_2) \in E(G)$.
\end{theorem}

%% file: algorithm.tex
\section{Type-based Qubit Allocation Algorithm}\label{sec:algorithm}

This section introduces a type-based algorithm of qubit allocation for the language defined in \cref{sec:src-lang}.
As described later, our algorithm uses heuristics because the qubit allocation problem has NP-hard subproblems.

The algorithm \textsc{QubitAlloc} is the entry point of our algorithm (\cref{alg:qalloc}).
\textsc{QubitAlloc} takes as input a well-typed source program $ \braket{  \languagessym{\{}  \languagesmv{d_{{\mathrm{1}}}}  \languagessym{,} \, .. \, \languagessym{,}  \languagesmv{d_{\languagesmv{n}}}  \languagessym{\}}  ,  e  } $ and a coupling graph,
and returns a target program, which satisfies connectivity constraints.

\cref{fig:example-alg} illustrates how our algorithm works.
In this example, the algorithm replaces the instructions for creating and discarding a qubit with the resetting operation
and extends the arguments and the return values of \texttt{func} with the third qubit
because all the qubits are allocated statically in the target language, unlike the source language.
Finally, the function \texttt{func} takes as input 3-qubits and requires that all the pairs of the arguments are connected.
The caller has to insert swap gates so that the arguments passed to \texttt{func} meet the requirement.
We can consider this requirement for the caller to be a calling convention in connectivity-aware quantum programming languages.
After the transformation of functions, the algorithm transforms the entry expression.
In the first call of \texttt{func}, we give to \texttt{func} $q_1, q_2, q_3$, which have all-to-all connections,
and thus the algorithm does not insert swap gates.
In contrast, in the second call of \texttt{func}, the arguments $q_2, q_4, q_3$ does not meet the requirements of \texttt{func}.
Hence, the algorithm searches a subgraph from the coupling graph isomorphic to the requirement
and then moves $q_2, q_4, q_3$ to the positions by inserting swap gates.
Before returning the values, the algorithm inserts swap gates to preserve the semantics of the given program.

\begin{figure}[tb]
  \centering
  \begin{tikzpicture}
    \node[shape=circle, draw] (q0) at (0, 1.5) {$q_1$};
    \node[shape=circle, draw] (q1) at (-1.5, 0) {$q_2$};
    \node[shape=circle, draw] (q2) at (0, 0) {$q_3$};
    \node[shape=circle, draw] (q3) at (1.5, 0) {$q_4$};

    \draw [>=stealth] (q0) -- (q1);
    \draw [>=stealth] (q0) -- (q2);
    \draw [>=stealth] (q1) -- (q2);
    \draw [>=stealth] (q3) -- (q2);
  \end{tikzpicture}

  \vspace{4mm}
  \begin{tabular}{cc}
    \begin{minipage}{0.45\hsize}
\begin{lstlisting}[language=target, basicstyle=\ttfamily\scriptsize]
def func(x1, x2) in 
  let x3 = init() in
  let (x1, x2) = cnot(x1, x2) in
  let (x2, x3) = cnot(x2, x3) in
  let (x3, x1) = cnot(x3, x1) in
  discard x3;
  (x1, x2)
}

let y1 = init() in
let y2 = init() in
let y3 = init() in
let (y1, y2) = func(y1, y2) in
let (y2, y3) = func(y2, y3) in
(y1, y2, y3)
\end{lstlisting}
    \end{minipage}
    \begin{minipage}{0.55\hsize}
\begin{lstlisting}[language=target, basicstyle=\ttfamily\scriptsize]
def func(x1, x2, x3) in 
  let (x1, x2) = cnot(x1, x2) in
  let (x2, x3) = cnot(x2, x3) in
  let (x3, x1) = cnot(x3, x1) in
  init x3;
  (x1, x2, x3)
}

let (q1, q2, q3) = func(q1, q2, q3) in
let (q3, q1) = swap(q1, q3) in
let (q4, q3) = swap(q3, q4) in
let (q2, q4, q3) = func(q2, q4, q3) in
let (q4, q3) = swap(q3, q4) in
let (q3, q1) = swap(q1, q3) in
(q1, q2, q3, q4)
\end{lstlisting}
    \end{minipage}
  \end{tabular}

  \caption{
    An example of the qubit allocation algorithm.
    On the coupling graph displayed above,
    it transforms the source program (left) into the target program (right).
  }
  \label{fig:example-alg}
\end{figure}

Our algorithm works in three steps:
(1) it assigns a subgraph of the coupling graph to each function (\textsc{AssignSubgraph}),
(2) it solves the qubit allocation problem on each function independently (\textsc{QubitAllocFunc}), and
(3) it solve the qubit allocation problem for the main expression $e$ (\textsc{QubitAllocExp}).
In the following subsections, we describe each procedure in detail.

%

\begin{algorithm}[tb]
  \caption{Type-based Qubit Allocation}
  \label{alg:qalloc}
  \begin{algorithmic}[1]
    \Require{$\Delta$ : the derivation tree of a source program, \quad $G$ : a coupling graph}
    \Function{QubitAlloc}{$\Delta, G$}
      \State
        \AxiomC{$\Theta  \vdash  D$}\AxiomC{$\Theta  \pipe  N  \pipe   \emptyset   \vdash  e  \languagessym{:}  T$}
        \BinaryInfC{$N  \vdash   \braket{  D  ,  e  } $}\DisplayProof
        $\gets \Delta$
      \State $S \gets$ \Call{AssignSubgraphs}{$\Theta, G$} \label{lst:line:assign-subgraphs}
      \State $\Theta', D' \gets$ \Call{QubitAllocFunc}{$\Theta, D, S$}
      \State $\Gamma \gets \{\languagesmv{x_{\languagesmv{i}}} :  \texttt{q}( \alpha_{\languagesmv{i}} )  \pipe \alpha_{\languagesmv{i}} \in V(G)\}$, where $\languagesmv{x_{\languagesmv{i}}}$ are fresh
      \State $e' \gets$ \Call{QubitAllocExp}{$e, \Theta', E(G), \emptyset, \Gamma$}
      \State \Return $\Theta', D', e'$
    \EndFunction
  \end{algorithmic}
\end{algorithm}

\subsection{Assignment of a Subgraph to Each Function}\label{subsec:assignment}

Consider the case that a function $f$ calls another function $g$.
As shown in \cref{fig:example-alg},
the function $f$ have to insert swap gates to satisfy the connectivity constraints of $g$.
In other words, the graph composed of the arguments of $f$ has a subgraph isomorphic to
the graph of the arguments of $g$.

Our algorithm assigns a subgraph of a coupling graph to each function.
We can regard an assigned subgraph as a workspace of the function,
that is, the graph composed of the function's arguments.
We have to apply the CNOT gate and the swap gate in the subgraph.

As one possible assignment, we define \textsc{AssignSubgraph} in \cref{alg:assign-subgraphs}.
First, we construct a sequence of graphs $G_n \supseteq \dots \supseteq G_1$ by \textsc{ConstructSubgraphs},
which uses a heauristic and greedy algorithm that repeatedly removes a non-articulation point with the minimum degree (\textsc{NonArticulationPointWithMinDeg}).
The reason why we adopt this strategy is that it may increase the number of qubit pairs to which the swap gate can be applied.
Moreover, keeping subgraphs connected makes the algorithm simple
because we can always move a qubit to a node adjacent to any other qubit in a connected graph.
At last, \textsc{AssignSubgraphs} assigns $G_k$ to functions in which the number of qubits used is $k$.
This assignment ensures that all functions meet the arguments' requirements by inserting swap gates correctly.

\textbf{Complexity Analysis:}
The algorithm \textsc{ConstructSubgraph} calls the algorithm \textsc{NonArticulationPointWithMinDeg} in every iteration of the for-loop
which searches a non-articulation point with the minimum degree.
There is an implementation of \textsc{NonArticulationPointWithMinDeg} working in $O(V + E)$~\cite{hopcroftAlgorithm447Efficient1973}.
Thus, the time complexity of this step is $O(V(V + E))$.
The space complexity is $O(V + E)$ because it is sufficient to remember the differences between $G_i$ and $G_{i+1}$.

\begin{algorithm}[tb]
  \caption{Subgraph Assignments}
  \label{alg:assign-subgraphs}
  \begin{algorithmic}[1]
    \Function{AssignSubgraphs}{$\Theta, G$}
      \State $G_1, \dots, G_n \gets$ \Call{ConstructSubgraphs}{$G$}
      \State \Return $\{f \mapsto G_{N + |T_{{\mathrm{1}}}|} \pipe \forall \languagesmv{f}  \languagessym{:}   T_{{\mathrm{1}}}  \xrightarrow{ N }  T_{{\mathrm{2}}}  \in \Theta\}$
    \EndFunction

    \Function{ConstructSubgraphs}{$G$}
      \State $G_s \gets \bullet$
      \State $G_n \gets G$
      \State \texttt{/*} $d(v)$ is degrees of a vertex $v$. \texttt{*/}
      \For{$i \gets n$ downto $1$}
        \State $u \gets$ \Call{NonArticulationPointWithMinDeg}{$G_i$}
        \For{$w \in \mathop{\mathrm{adj}}(u)$}
          \State $d(w) \gets d(w) - 1$
        \EndFor
        \State $G_{i - 1} \gets G_i \backslash \{u\}$
        \State $G_s \gets G_s, G_i$
      \EndFor
      \State \Return $G_s$
    \EndFunction
  \end{algorithmic}
\end{algorithm}

\subsection{Qubit Allocation for a Function} 
In the next step, \textsc{QubitAllocFunc} translates well-typed function definitions $d_1, \dots, d_n$ (\cref{alg:qalloc-func}).
\textsc{QubitAllocFunc} reconstructs a function type environment $\Theta$ and function definitions $\languagesmv{d_{{\mathrm{1}}}}  \languagessym{,} \, .. \, \languagessym{,}  \languagesmv{d_{\languagesmv{n}}}$ in this order.

Unlike the source language, the target language cannot dynamically initialize and discard a qubit.
Hence, we extend the argument and the return value of a function with qubits initialized or discarded in the source language.
Moreover, we also extend the function type with the connectivity constraints $\Phi$, which \textsc{AssignSubgraph} assigns.

Next, we add the updated function type to the type environment, the result of \textsc{QubitAllocFunc} on previously defined functions.
\textsc{QubitAllocExp} (\cref{alg:qalloc-exp}) translates the body expression under the transformed environment, as explained in the next section.

\begin{algorithm}[tbp]
  \caption{Qubit Allocation Algorithm for Functions}
  \label{alg:qalloc-func}
  \begin{algorithmic}[1]
    \Function{QubitAllocFunc}{$\Theta, D, S$}
      \If{$\Theta = \bullet \land D = \bullet$}
        \State \Return $\bullet, \bullet$
      \ElsIf{$\Theta = \Theta', \languagesmv{f}  \languagessym{:}   \tau_{{\mathrm{1}}}  \languagessym{*} \, .. \, \languagessym{*}  \tau_{\languagesmv{n}}  \xrightarrow{ N }  T  \land D = D'  \languagessym{,}   \languagesmv{f}  \mapsto ( \languagesmv{x_{{\mathrm{1}}}}  \languagessym{,} \, .. \, \languagessym{,}  \languagesmv{x_{\languagesmv{n}}} )  e $}
        \State $\Theta_{{\mathrm{1}}}, D_{{\mathrm{1}}} \gets$ \Call{QubitAllocFunc}{$\Theta', D', S$}
        \State Generate fresh qidx variables $\alpha_{{\mathrm{1}}}, \dots, \alpha_N$
        \State $G_f \gets S(f); ( \overline{ \alpha' } , \Phi) \gets G_f$~~~($ \overline{ \alpha' }  = \alpha'_{{\mathrm{1}}}  \languagessym{,} \, .. \, \languagessym{,}  \alpha'_{\languagesmv{n}}$)
        \State $\Gamma_{{\mathrm{1}}} \gets \{ \languagesmv{x_{\languagesmv{i}}}  \languagessym{:}   \texttt{q}( \alpha'_{\languagesmv{i}} )  \pipe 1 \leq i \leq n \}$
        \State $\Gamma_{{\mathrm{2}}} \gets \{ \languagesmv{y_{\languagesmv{i}}}  \languagessym{:}   \texttt{q}( \beta_{\languagesmv{i}} )  \pipe 1 \leq i \leq N \land \languagesmv{y_{\languagesmv{i}}}, \beta_{\languagesmv{i}}\ \text{are fresh} \}$
        \State $T' \gets  \texttt{q}( \alpha'_{{\mathrm{1}}} )   \languagessym{*} \, .. \, \languagessym{*}   \texttt{q}( \alpha'_{\languagesmv{n}} )   \languagessym{*}   \texttt{q}( \beta_{{\mathrm{1}}} )   \languagessym{*} \, .. \, \languagessym{*}   \texttt{q}( \beta_{\languagesmv{N}} ) $.
        \State $ \overline{ \alpha }  \gets  \mathit{QV}( T' ) $
        \State $\Theta_{{\mathrm{2}}} \gets \Theta_{{\mathrm{1}}}  \languagessym{,}  \languagesmv{f}  \languagessym{:}   \forall   \overline{ \alpha }   .   \Phi  \Rightarrow  T'  \rightarrow  T'  $
        \State $e' \gets$ \Call{QubitAllocExp}{$e, \Theta_{{\mathrm{2}}}, \Phi, \Gamma_{{\mathrm{1}}}, \Gamma_{{\mathrm{2}}}$}
        \State $D_{{\mathrm{2}}} \gets D_{{\mathrm{1}}}  \languagessym{,}   \languagesmv{f}  \mapsto ( \languagesmv{x_{{\mathrm{1}}}}  \languagessym{,} \, .. \, \languagessym{,}  \languagesmv{x_{\languagesmv{n}}}  \languagessym{,}  \languagesmv{y_{{\mathrm{1}}}}  \languagessym{,} \, .. \, \languagessym{,}  \languagesmv{y_{\languagesmv{k}}} )  e' $
        \State \Return $\Theta_{{\mathrm{2}}}, D_{{\mathrm{2}}}$
      \EndIf
    \EndFunction
  \end{algorithmic}
\end{algorithm}

\subsection{Qubit Allocation for an Expression} 
\textsc{QubitAllocExp} is the central part of the algorithm.
This algorithm takes a source expression and the target language's contexts and returns a target expression.
Under the given contexts, the result of \textsc{QubitAllocExp} is well-typed.
We explain briefly how \textsc{QubitAllocExp} works by case analysis.

In the case of returning values, the algorithm adds qubits to the return values which do not appear in the source type environment.
If the expression is the body expression of a function, the algorithm may insert swap gates before returning values so that the return type is consistent with the type of the function.

In the case of initializing a qubit, the algorithm consumes an unused qubit in $\Gamma_{{\mathrm{2}}}$
and transforms the subsequent expression where $z$ is substituted for $x$.
Here, $[z / x]\Delta$ replaces $x$ with $z$ in the type environments and the expressions of $\Delta$.
This substitution is a type-preserving operation because if the typing derivation tree $\Delta$ is derivable,
$[z / x]\Delta$ is also derivable.

The discarding expression corresponds to the initialization expression in the target language.
Hence, the algorithm changes $\languageskw{discard} \, \languagesmv{x}  \languagessym{;}  e$ to $\languageskw{init} \, \languagesmv{x}  \languagessym{;}  e'$,
where $e'$ is the result of $\textsc{QubitAllocExp}$ with $e$.
The algorithm also moves $x$ to $\Gamma_{{\mathrm{2}}}$, whose variables are unused.

The case of calling a function is critical.
When calling a function $f$, we have to solve two subproblems:
(1) finding a subgraph of the current workspace isomorphic to $\Phi_f$, and
(2) moving the arguments to the isomorphic subgraph.
The problem (1) is called \emph{Subgraph Isomorphism Problem}, which is a well-known NP-complete problem.
From the way to construct subgraphs in step 1, we can immediately find one of the subgraphs
isomorphic to $\Phi_f$ from the current graph.
Of course, we can choose another subgraph if we want a more efficient solution, but we use the trivial subgraph currently.
The problem (2) is called \emph{Token Swapping Problem}~\cite{yamanakaSwappingLabeledTokens2014}, which is proved to be NP-hard.
We adopt a 4-approximation algorithm~\cite{miltzowApproximationHardnessToken2016} for Token Swapping Problem
so that our algorithm works in polynomial time.
We emphasize that \textsc{SubgraphIsomorphism} and \textsc{TokenSwapping} always successes
because of the way of assigning a subgraph to each function as explained in \cref{subsec:assignment}.
After solving these subproblems, \textsc{QubitAlloc} inserts swap gates by using the result of the token swapping problem and transforms contexts.
The algorithm has to give additional qubits to the functions because functions have additional arguments and return values, as explained in the previous section.
We can use any variables as free qubits which appear in $\Gamma_{{\mathrm{2}}}$
because \textsc{QubitAllocExp} keeps the condition that no used qubits appear in the $\Gamma_{{\mathrm{2}}}$.

We give the definitions of the algorithms for the Subgraph Isomorphism problem and the Token Swapping problem.
We do not give the implementation of these algorithms but assume that these algorithms are implemented correctly.

\begin{definition}{(Subgraph Isomorphism)}
  The algorithm \textsc{SubgraphIsomorphism} takes two graphs $G_1$ and $G_2$,
  and outputs an embbeding injection $\phi : G_2 \rightarrow G_1$
  such that $\phi (G_2) \simeq G_2$
  if $G_1$ has a subgraph isomorphic to $G_2$.
\end{definition}

\begin{definition}{(Token Swapping)}
  The algorithm \textsc{TokenSwapping} takes a graph $G$ and
  a partial injection $p : V(G) \rightarrow V(G)$,
  and returns a sequence of pairs of vertices $\Psi = (v_{11}, v_{12}), \dots, (v_{L1}, v_{L2})$
  such that $\forall i \in \{1, \dots, L\}.\ (v_{i1}, v_{i2}) \in G$,
  and $\forall v \in G.\ p(v) = \bot \lor p(v) = \sigma (v)$,
  where $\sigma = (v_{L1}\ v_{L2}) \circ \dots \circ (v_{11}\ v_{12})$ and $\bot$ denotes the undefined value.
\end{definition}

\textbf{Complexity Analysis:}
Let $M$ be the size of the given program. 
The most time-consuming case is a function call, where \textsc{SubgraphIsomorphism} and \textsc{TokenSwapping} are invoked.
If the time complexities of \textsc{SubgraphIsomorphism} and \textsc{TokenSwapping} are $O(r)$ and $O(t)$ respectively,
then $\textsc{QubitAllocExp}$ works in $O(M(s + r))$.
In our settings, there is an implementation of \textsc{SubgraphIsomorphism} working in $O(1)$.
We can also implement the 4-approximation algorithm \textsc{TokenSwapping} running in $O(V^3)$, where $V$ is the size of nodes of a coupling graph.

\begin{algorithm}[tbp]
  \caption{Qubit Allocation for Expressions}
  \label{alg:qalloc-exp}
  \begin{algorithmic}[1]
    \Function{QubitAllocExp}{$e, \Theta, \Phi, \Gamma_{{\mathrm{1}}}, \Gamma_{{\mathrm{2}}}$}
      \If {$e \equiv \languagessym{(}  \languagesmv{x_{{\mathrm{1}}}}  \languagessym{,} \, .. \, \languagessym{,}  \languagesmv{x_{\languagesmv{n}}}  \languagessym{)}$}
        \State (The expected return type is $ \texttt{q}( \alpha_{{\mathrm{1}}} )   \languagessym{*} \, .. \, \languagessym{*}   \texttt{q}( \alpha_{\languagesmv{k}} ) $)
        \Assert{$\Gamma_{{\mathrm{1}}}, \Gamma_{{\mathrm{2}}} = \languagesmv{x_{{\mathrm{1}}}}  \languagessym{:}   \texttt{q}( \beta_{{\mathrm{1}}} )   \languagessym{,} \, .. \, \languagessym{,}  \languagesmv{x_{\languagesmv{k}}}  \languagessym{:}   \texttt{q}( \beta_{\languagesmv{k}} ) $}
        \State $\Psi \gets$ \Call{TokenSwapping}{$( \mathit{QV}( \Gamma_{{\mathrm{1}}} ) , \Phi), \{\alpha_{\languagesmv{i}} \mapsto \beta_{\languagesmv{i}}\}_{i=1}^k$}
        \State \Return \Call{InsertSwaps}{$\languagessym{(}  \languagesmv{x_{{\mathrm{1}}}}  \languagessym{,} \, .. \, \languagessym{,}  \languagesmv{x_{\languagesmv{k}}}  \languagessym{)}, \Gamma_{{\mathrm{1}}} \cup \Gamma_{{\mathrm{2}}}, \Psi$}

      \ElsIf {$e \equiv \languageskw{let} \, \languagesmv{x}  \languagessym{=} \, \languageskw{init} \, \languagessym{()} \, \languageskw{in} \, e'$}
        \State Take $\languagesmv{x'}  \languagessym{:}   \texttt{q}( \alpha )  \in \Gamma_{{\mathrm{2}}}$
        \State \Return \Call{QubitAllocExp}{$[x' / x]\Delta', \Theta, \Phi, \Gamma_{{\mathrm{1}}} \uplus \{\languagesmv{x'}  \languagessym{:}   \texttt{q}( \alpha ) \}, \Gamma_{{\mathrm{2}}} \backslash \languagesmv{x'}$}

      \ElsIf {$e \equiv \languageskw{discard} \, \languagesmv{x}  \languagessym{;}  e'$}
        \State $e' \gets$ \Call{QubitAllocExp}{$\Delta', \Theta, \Phi, \Gamma_{{\mathrm{1}}}, \Gamma_{{\mathrm{2}}}$}
        \State \Return $\languageskw{init} \, \languagesmv{x}  \languagessym{;}  e'$

      \ElsIf {$e \equiv \languageskw{let} \, \languagessym{(}  \languagesmv{x_{{\mathrm{1}}}}  \languagessym{,} \, .. \, \languagessym{,}  \languagesmv{x_{\languagesmv{n}}}  \languagessym{)}  \languagessym{=}  \languagesmv{f}  \languagessym{(}  \languagesmv{y_{{\mathrm{1}}}}  \languagessym{,} \, .. \, \languagessym{,}  \languagesmv{y_{\languagesmv{m}}}  \languagessym{)} \, \languageskw{in} \, e'$}
        \Assert{$\Theta  \languagessym{(}  \languagesmv{f}  \languagessym{)} =  \forall   \overline{ \alpha }   .   \Phi'  \Rightarrow  \tau_{{\mathrm{1}}}  \languagessym{*} \, .. \, \languagessym{*}  \tau_{\languagesmv{k}}  \rightarrow  \tau'_{{\mathrm{1}}}  \languagessym{*} \, .. \, \languagessym{*}  \tau'_{\languagesmv{k}}  $}
        \State $G \gets ( \mathit{QV}( \Gamma_{{\mathrm{1}}} ) , \Phi); G_f \gets ( \overline{ \alpha } , \Phi')$
        \State $\phi \gets$ \Call{SubgraphIsomorphism}{$G, G_f$}
        \State Take $y_{m+1} : \texttt{q}(\beta_{m+1}), \dots, y_k :  \texttt{q}( \beta_{\languagesmv{k}} )  \in \Gamma_{{\mathrm{2}}}$
        \Assert{$\tau_{{\mathrm{1}}} =  \texttt{q}( \alpha_{{\mathrm{1}}} )  \land \cdots \land \tau_{\languagesmv{k}} =  \texttt{q}( \alpha_{\languagesmv{k}} )  \land \languagesmv{y_{{\mathrm{1}}}}  \languagessym{:}   \texttt{q}( \beta_{{\mathrm{1}}} )   \languagessym{,} \, .. \, \languagessym{,}  \languagesmv{y_{\languagesmv{m}}}  \languagessym{:}   \texttt{q}( \beta_{\languagesmv{m}} )  \in \Gamma_{{\mathrm{1}}}$}
        \State $\Psi \gets$ \Call{TokenSwapping}{$G, \{\beta_{\languagesmv{i}} \mapsto \phi (\alpha_{\languagesmv{i}})\}_{i=1}^k$}
        \State $ \sigma _{ \alpha }  \gets \languagessym{[}   \phi ( \alpha_{{\mathrm{1}}} )   \slash  \alpha_{{\mathrm{1}}}  \languagessym{,} \, .. \, \languagessym{,}   \phi ( \alpha_{\languagesmv{k}} )   \slash  \alpha_{\languagesmv{k}}  \languagessym{]}$
        \State $\Gamma_{{\mathrm{3}}} \gets (\Psi(\Gamma_{{\mathrm{1}}}) \backslash \{\languagesmv{y_{{\mathrm{1}}}}  \languagessym{,} \, .. \, \languagessym{,}  \languagesmv{y_{\languagesmv{m}}}\}) \cup \{\languagesmv{x_{{\mathrm{1}}}}  \languagessym{:}   \sigma _{ \alpha }  \, \tau'_{{\mathrm{1}}}  \languagessym{,} \, .. \, \languagessym{,}  \languagesmv{x_{\languagesmv{n}}}  \languagessym{:}   \sigma _{ \alpha }  \, \tau'_{\languagesmv{n}}\}$
        \State $\Gamma_{{\mathrm{4}}} \gets (\Psi(\Gamma_{{\mathrm{2}}}) \backslash \{y_{m+1}, \dots, y_k\}) \cup \{x_{n+1} :  \sigma _{ \alpha } \tau_{n+1}', \dots, \languagesmv{x_{\languagesmv{k}}}  \languagessym{:}   \sigma _{ \alpha }  \, \tau'_{\languagesmv{k}}\}$, where $x_{n+1}, \dots, x_k$ are fresh variables.
        \State $e' \gets$ \Call{QubitAllocExp}{$e', \Theta, \Phi, \Gamma_{{\mathrm{3}}}, \Gamma_{{\mathrm{4}}}$}
        \State $e \gets \languageskw{let} \, \languagessym{(}  \languagesmv{x_{{\mathrm{1}}}}  \languagessym{,} \, .. \, \languagessym{,}  \languagesmv{x_{\languagesmv{k}}}  \languagessym{)}  \languagessym{=}  \languagesmv{f}  \languagessym{(}  \languagesmv{y_{{\mathrm{1}}}}  \languagessym{,} \, .. \, \languagessym{,}  \languagesmv{y_{\languagesmv{k}}}  \languagessym{)} \, \languageskw{in} \, e'$
        \State \Return \Call{InsertSwaps}{$e, \Gamma_{{\mathrm{1}}} \cup \Gamma_{{\mathrm{2}}}, \Psi$}
      \EndIf
      \State (The other cases are omit due to limitations of space.)
      \State 
    \EndFunction

    \Function{InsertSwaps}{$e, \Gamma, \Psi$}
      \If{$\Psi = \epsilon$}
      \State \Return $e$
      \ElsIf{$\Psi = (\alpha_{{\mathrm{1}}}, \alpha_{{\mathrm{2}}}), \Psi'$}
      \State $e' \gets$ \Call{InsertSwaps}{$e, (\alpha_{{\mathrm{1}}}\ \alpha_{{\mathrm{2}}})(\Gamma), \Psi'$}
      \State Find $\languagesmv{x_{{\mathrm{1}}}}  \languagessym{:}   \texttt{q}( \alpha_{{\mathrm{1}}} )   \languagessym{,}  \languagesmv{x_{{\mathrm{2}}}}  \languagessym{:}   \texttt{q}( \alpha_{{\mathrm{2}}} )  \in \Gamma_{\languagesmv{i}}$
      \State \Return $\languageskw{let} \, \languagessym{(}  \languagesmv{x_{{\mathrm{2}}}}  \languagessym{,}  \languagesmv{x_{{\mathrm{1}}}}  \languagessym{)}  \languagessym{=} \, \languageskw{swap} \, \languagessym{(}  \languagesmv{x_{{\mathrm{1}}}}  \languagessym{,}  \languagesmv{x_{{\mathrm{2}}}}  \languagessym{)} \, \languageskw{in} \, e'$
      \EndIf
    \EndFunction
  \end{algorithmic}
\end{algorithm}

\subsection{Correctness of the Algorithm}
To conclude this chapter, we prove the correctness of the qubit allocation algorithm, that is,
the algorithm translates a well-typed source program into a well-typed target program.

\begin{theorem}{(Type preservation)}\label{thm:alg-type-preserve}
  Let $G$ be a coupling graph.
  Suppose that $N  \vdash   \braket{  D  ,  e  } $, its derivation tree is $\Delta$ and $N \leq |V(G)|$.
  Then $\textsc{QubitAlloc}(\Delta, G)$ is derivable.
\end{theorem}

\begin{proof}
  Please see \cref{sec:appendix-alg}.
\end{proof}

\cref{thm:alg-type-preserve} states that the result of \textsc{QubitAllocExp} always satisfies connectivity constraints,
and thus it is sufficient to write a well-typed program in the source language.

Finally, we conjecture that our algorithm preserves the semantics of a given program.


\begin{definition}
  Density operators $\rho_1, \rho_2$ are \emph{isomorphic}
  if $\rho_1$ equals $\rho_2$ ignoring their labels and the order of tensor products.
  We write $\rho_1 \simeq \rho_2$ for isomorphic density operators $\rho_1, \rho_2$.
\end{definition}
For example, $\ketbra{01}{01}{x_1,x_2} \simeq \ketbra{10}{10}{x_2, x_1} \simeq \ketbra{10}{10}{y_1, y_2}$.

\begin{conjecture}{(Semantics preservation)}
  Suppose that $N  \vdash   \braket{  D  ,  e  } $ and its derivation tree is $\Delta$.
  Let $G$ be a coupling graph such that $N \leq |V(G)|$.
  We put $\Theta', D', e' \coloneqq \textsc{QubitAlloc}(\Delta, G)$.
  Let $\qclosure{X_1, \rho_1, e}, \qclosure{\rho_1', e'}$ be runtime states such that $\rho_1 \simeq \rho_1'$.
  If $\qclosure{X_1, \rho_1, e}  \rightarrow _{ D } ^* \qclosure{X_2, \rho_2, \languagessym{(}  \languagesmv{x_{{\mathrm{1}}}}  \languagessym{,} \, .. \, \languagessym{,}  \languagesmv{x_{\languagesmv{n}}}  \languagessym{)}}$,
  then there exists $\qclosure{\rho_2', \languagessym{(}  \languagesmv{y_{{\mathrm{1}}}}  \languagessym{,} \, .. \, \languagessym{,}  \languagesmv{y_{\languagesmv{m}}}  \languagessym{)}}$ such that $\qclosure{\rho_1', e'}  \rightarrow _{D', G}^* \qclosure{\rho_2', \languagessym{(}  \languagesmv{y_{{\mathrm{1}}}}  \languagessym{,} \, .. \, \languagessym{,}  \languagesmv{y_{\languagesmv{m}}}  \languagessym{)}}$ and $\rho_2 \simeq \rho_2'$.
  The opposite is also true.
\end{conjecture}

Bisimulation relation between the source and target language may be a key to proving this conjecture.

%% file: related_work.tex
\section{Related Work}\label{sec:related-work}

\paragraph{Type-Based Register Allocation.}
The idea of formalizing the process of qubit allocation by using type systems is influenced by Ohori's work~\cite{ohoriRegisterAllocationProof2004}.
He proposed a proof-theoretical framework for register allocation --- liveness analysis, inserting instructions, and optimization.
His framework is concerned with the number of registers, but not the locations of them
because the locations are not important in register allocation.
In contrast, qubit allocation is mainly concerned with the positions of qubits,
and thus the location-aware type system of the target language is suitable for qubit allocation.

\paragraph{Qubit Allocation.}
For the algorithm of qubit allocation, many studies have been conducted~\cite{siraichiQubitAllocation2018,siraichiQubitAllocationCombination2019,dengCodarContextualDurationAware2020,nishioExtractingSuccessIBM2020,muraliFormalConstraintbasedCompilation2019a,itokoOptimizationQuantumCircuit2020}.
Siraichi \etal proposed an optimal solution by dynamic programming and a heuristic solution by combining Subgraph Isomorphism and Token Swapping~\cite{siraichiQubitAllocationCombination2019}.
Nishio \etal developed a noise-aware compilation tool based on a heuristic search algorithm
for noisy quantum computers~\cite{nishioExtractingSuccessIBM2020}.
Deng \etal presented a context-sensitive and duration-aware remapping algorithm (\textsc{Codar}), which considers the parallelism of a program~\cite{dengCodarContextualDurationAware2020}.
Hietala \etal developed a verified optimizer for quantum circuits (\textsc{voqc}) by Coq
and implemented a simple qubit allocation algorithm~\cite{hietalaVerifiedOptimizerQuantum2021} in the framework.
The difference between theirs and ours is that they used a low-level quantum program, which does not have functions.
Furthermore, our framework currently is not concerned with the other constraints, such as quantum error.

\paragraph{Resource Aware Quantum Programming.}
Amy recently proposed metaQASM~\cite{amySizedTypesLowlevel2019},
an extension of the openQASM supporting the metaprogramming of quantum circuits.
MetaQASM uses \emph{sized types}~\cite{hughesProvingCorrectnessReactive1996} to analyze how many qubits a quantum circuit uses.
Our source language is quite similar to metaQASM for type-based resource analysis.
MetaQASM has a strong normalization property, while we can write infinite loops with recursive calls in our languages.


%% file: conclusion.tex
\section{Conclusion and Future Work}\label{sec:conclusion}

This paper presented a type-based framework of qubit allocation for a first-order quantum programming language.
The source language estimates the number of qubits required for a program to run by linear types.
The target language checks whether a given program satisfies connectivity constraints by qualified types.
The type-based qubit allocation algorithm translates a well-typed source program into a well-typed target program.
We also proved the type-safety of both languages and the type-preserving property of our algorithm.

We are trying to prove that the translation preserves the semantics of a given program.
We are also implementing our algorithm.
After we finish implementing the algorithm, we will find or make a benchmark for high-level quantum programs
because standard benchmarks rely on low-level quantum programs, such as~\cite{maslovReversibleLogicSynthesis}.
We shall also extend our language with higher-order functions.
Our previous work~\cite{wakizakaDependentTypeSystem2020} can verify connectivity constraints for higher-order programming languages,
while it does not take qubit allocation into account.
In higher-order programming languages, a function may be called with functions whose constraints are different,
and thus qubit allocation is made more complicated.
Beyond near-term quantum computers, we plan to develop a type-based compilation technique for fault-tolerant quantum computers.
The constraints in such computers are different from near-term quantum computers,
but we can nevertheless apply a type-based approach to those constraints similarly to this study.

%% file: appendix.tex
\section{The Proofs of the Source Language}\label{sec:src-proofs}

To prove the type soundness, we introduce the well-typedness of a runtime state in \cref{fig:src-typing-state}.
\begin{figure}[htbp]
  \infrule[]{
    \Theta  \vdash  D
    \andalso \Theta  \pipe  N  \pipe  \Gamma  \vdash  e  \languagessym{:}  T \\
    |X| \geq N
    \andalso \dom(\Gamma) \cap X = \emptyset
    \andalso \dom(\Gamma) \uplus X \subseteq \Var(\rho)
  }{
     N  \pipe  \Gamma  \vdash_{ D }  \qclosure{X, \rho, e}
  }
  \caption{The typing rule for a runtime state}
  \label{fig:src-typing-state}
\end{figure}

\begin{lemma}{(Progress)}\label{lem:src-progress}
If $ N  \pipe  \Gamma  \vdash_{ D }  \braket{X, \rho, e}$, then one of the following conditions hold:
  \begin{itemize}
  \item $e$ is $\languagessym{(}  \languagesmv{x_{{\mathrm{1}}}}  \languagessym{,} \, .. \, \languagessym{,}  \languagesmv{x_{\languagesmv{n}}}  \languagessym{)}$, or
  \item $\exists \braket{X', \rho', e'}$ . $\braket{X, \rho, e}  \rightarrow  \braket{X', \rho', e'}$.
  \end{itemize}
\end{lemma}

\begin{proof}
  By induction on the typing derivation of $e$.
  \begin{itemize}
  \item Case of (\textsc{T-Let}):
    \begin{gather*}
      \Theta  \pipe  N  \pipe  \Gamma_{{\mathrm{1}}}  \vdash  e_{{\mathrm{1}}}  \languagessym{:}  T_{{\mathrm{1}}}
      \andalso \Theta  \pipe  N  \pipe  \Gamma_{{\mathrm{1}}}  \languagessym{,}  \Gamma_{{\mathrm{2}}}  \vdash  \languageskw{let} \, \languagessym{(}  \languagesmv{x_{{\mathrm{1}}}}  \languagessym{,} \, .. \, \languagessym{,}  \languagesmv{x_{\languagesmv{n}}}  \languagessym{)}  \languagessym{=}  e_{{\mathrm{1}}} \, \languageskw{in} \, e_{{\mathrm{2}}}  \languagessym{:}  T
      \andalso \Gamma = \Gamma_{{\mathrm{1}}}  \languagessym{,}  \Gamma_{{\mathrm{2}}}
    \end{gather*}
    \begin{itemize}
    \item If $e_{{\mathrm{1}}} = \languagessym{(}  \languagesmv{y_{{\mathrm{1}}}}  \languagessym{,} \, .. \, \languagessym{,}  \languagesmv{y_{\languagesmv{n}}}  \languagessym{)}$,
      then $\braket{X, \rho, e}  \rightarrow  \braket{X, \rho, \languagessym{[}  \languagesmv{y_{{\mathrm{1}}}}  \slash  \languagesmv{x_{{\mathrm{1}}}}  \languagessym{,} \, .. \, \languagessym{,}  \languagesmv{y_{\languagesmv{n}}}  \slash  \languagesmv{x_{\languagesmv{n}}}  \languagessym{]} \, e_{{\mathrm{2}}}}$ by (\textsc{E-Let}).
    \item Otherwise: Obviously, $\Gamma_{{\mathrm{1}}} \subseteq \Gamma$,
      and thus $\Gamma_{{\mathrm{1}}} \uplus X \subseteq \Var(\rho)$ and $\Gamma_{{\mathrm{1}}} \cap X$.
      Therefore, this case follows the induction hypothesis.
    \end{itemize}
  \item The other cases are easy.
  \end{itemize}
\end{proof}

\begin{lemma}{(Substitution)}\label{lem:src-substitution}
  Suppose that variables $\languagesmv{y_{{\mathrm{1}}}}  \languagessym{,} \, .. \, \languagessym{,}  \languagesmv{y_{\languagesmv{n}}}$ are distinct,
  $\languagesmv{y_{{\mathrm{1}}}}  \languagessym{,} \, .. \, \languagessym{,}  \languagesmv{y_{\languagesmv{n}}} \not\in \dom(\Gamma)$,
  $\{\languagesmv{x_{{\mathrm{1}}}}  \languagessym{,} \, .. \, \languagessym{,}  \languagesmv{x_{\languagesmv{n}}}\} \subseteq \dom(\Gamma)$
  and $\Theta  \pipe  N  \pipe  \Gamma  \vdash  e  \languagessym{:}  T$.
  Then, $\Theta  \pipe  N  \pipe  \languagessym{[}  \languagesmv{y_{{\mathrm{1}}}}  \slash  \languagesmv{x_{{\mathrm{1}}}}  \languagessym{,} \, .. \, \languagessym{,}  \languagesmv{y_{\languagesmv{n}}}  \slash  \languagesmv{x_{\languagesmv{n}}}  \languagessym{]} \, \Gamma  \vdash  \languagessym{[}  \languagesmv{y_{{\mathrm{1}}}}  \slash  \languagesmv{x_{{\mathrm{1}}}}  \languagessym{,} \, .. \, \languagessym{,}  \languagesmv{y_{\languagesmv{n}}}  \slash  \languagesmv{x_{\languagesmv{n}}}  \languagessym{]} \, e  \languagessym{:}  T$.
\end{lemma}

\begin{proof}
  By induction on the typing derivation of $e$.
\end{proof}

\begin{lemma}{(Extend a function type environment)}\label{lem:src-ext-func-env}
  If $\Theta  \pipe  N  \pipe  \Gamma  \vdash  e  \languagessym{:}  T$,
  then $\Theta'  \pipe  N  \pipe  \Gamma  \vdash  e  \languagessym{:}  T$ for any $\Theta \subseteq \Theta'$.
\end{lemma}

\begin{proof}
  Straightforward induction on the typing derivation of $e$.
\end{proof}

\begin{lemma}\label{lem:src-incr-n}
  If $\Theta  \pipe  N  \pipe  \Gamma  \vdash  e  \languagessym{:}  T$,
  then $\Theta  \pipe  N'  \pipe  \Gamma  \vdash  e  \languagessym{:}  T$ for any $N' \geq N$.
\end{lemma}

\begin{proof}
  Straightforward induction on the typing derivation of $e$.
\end{proof}

\begin{lemma}{(Subject Reduction)}\label{lem:src-sr}
  If $ N  \pipe  \Gamma  \vdash_{ D }  \qclosure{X, \rho, e}$
    and $\qclosure{X, \rho, e}  \rightarrow  \qclosure{X', \rho', e'}$,
  then there exists $N'$ and $\Gamma'$ such that all the following conditions hold:
  \begin{itemize}
  \item $N + \dom(\Gamma) = N' + \dom(\Gamma')$, and
  \item $ N'  \pipe  \Gamma'  \vdash_{ D }  \qclosure{X', \rho', e'}$.
  \end{itemize}
\end{lemma}

\begin{proof}
  By induction on the derivation of $\braket{X, \rho, e}  \rightarrow  \braket{X', \rho', e'}$.
  \begin{itemize}
  \item Case of (\textsc{E-Let1}):
    \begin{gather}
      \braket{X, \rho, \languageskw{let} \, \languagessym{(}  \languagesmv{x_{{\mathrm{1}}}}  \languagessym{,} \, .. \, \languagessym{,}  \languagesmv{x_{\languagesmv{m}}}  \languagessym{)}  \languagessym{=}  \languagessym{(}  \languagesmv{y_{{\mathrm{1}}}}  \languagessym{,} \, .. \, \languagessym{,}  \languagesmv{y_{\languagesmv{m}}}  \languagessym{)} \, \languageskw{in} \, e_{{\mathrm{2}}}}
         \rightarrow  \braket{X, \rho, \languagessym{[}  \languagesmv{y_{{\mathrm{1}}}}  \slash  \languagesmv{x_{{\mathrm{1}}}}  \languagessym{,} \, .. \, \languagessym{,}  \languagesmv{y_{\languagesmv{m}}}  \slash  \languagesmv{x_{\languagesmv{m}}}  \languagessym{]} \, e_{{\mathrm{2}}}} \notag \\
      \Theta  \pipe  N  \pipe  \Gamma_{{\mathrm{2}}}  \languagessym{,}  \languagesmv{y_{{\mathrm{1}}}}  \languagessym{:}  \tau_{{\mathrm{1}}}  \languagessym{,} \, .. \, \languagessym{,}  \languagesmv{y_{\languagesmv{m}}}  \languagessym{:}  \tau_{\languagesmv{m}}  \vdash  \languageskw{let} \, \languagessym{(}  \languagesmv{x_{{\mathrm{1}}}}  \languagessym{,} \, .. \, \languagessym{,}  \languagesmv{x_{\languagesmv{m}}}  \languagessym{)}  \languagessym{=}  \languagessym{(}  \languagesmv{y_{{\mathrm{1}}}}  \languagessym{,} \, .. \, \languagessym{,}  \languagesmv{y_{\languagesmv{m}}}  \languagessym{)} \, \languageskw{in} \, e_{{\mathrm{2}}}  \languagessym{:}  T \notag \\
      \andalso \Theta  \pipe  N  \pipe  \Gamma_{{\mathrm{2}}}  \languagessym{,}  \languagesmv{x_{{\mathrm{1}}}}  \languagessym{:}  \tau_{{\mathrm{1}}}  \languagessym{,} \, .. \, \languagessym{,}  \languagesmv{x_{\languagesmv{m}}}  \languagessym{:}  \tau_{\languagesmv{m}}  \vdash  e_{{\mathrm{2}}}  \languagessym{:}  T \label{eq:src-sr-1}
    \end{gather}
    By Lemma \ref{lem:src-substitution} with (\ref{eq:src-sr-1}),
    \begin{gather*}
      \Theta  \pipe  N''  \pipe  \Gamma_{{\mathrm{2}}}  \languagessym{,}  \languagesmv{y_{{\mathrm{1}}}}  \languagessym{:}  \tau_{{\mathrm{1}}}  \languagessym{,} \, .. \, \languagessym{,}  \languagesmv{y_{\languagesmv{m}}}  \languagessym{:}  \tau_{\languagesmv{m}}  \vdash  \languagessym{[}  \languagesmv{y_{{\mathrm{1}}}}  \slash  \languagesmv{x_{{\mathrm{1}}}}  \languagessym{,} \, .. \, \languagessym{,}  \languagesmv{y_{\languagesmv{m}}}  \slash  \languagesmv{x_{\languagesmv{m}}}  \languagessym{]} \, e_{{\mathrm{2}}}  \languagessym{:}  T
    \end{gather*}
    Therefore, $ N  \pipe  \Gamma_{{\mathrm{2}}}  \languagessym{,}  \languagesmv{y_{{\mathrm{1}}}}  \languagessym{:}  \tau_{{\mathrm{1}}}  \languagessym{,} \, .. \, \languagessym{,}  \languagesmv{y_{\languagesmv{m}}}  \languagessym{:}  \tau_{\languagesmv{m}}  \vdash_{ D }  \qclosure{X, \rho, \languagessym{[}  \languagesmv{y_{{\mathrm{1}}}}  \slash  \languagesmv{x_{{\mathrm{1}}}}  \languagessym{,} \, .. \, \languagessym{,}  \languagesmv{y_{\languagesmv{m}}}  \slash  \languagesmv{x_{\languagesmv{m}}}  \languagessym{]} \, e_{{\mathrm{2}}}}$.
    The other conditions hold obviously,
    and thus $ N  \pipe  \Gamma_{{\mathrm{2}}}  \languagessym{,}  \languagesmv{y_{{\mathrm{1}}}}  \languagessym{:}  \tau_{{\mathrm{1}}}  \languagessym{,} \, .. \, \languagessym{,}  \languagesmv{y_{\languagesmv{m}}}  \languagessym{:}  \tau_{\languagesmv{m}}  \vdash_{ D }  \qclosure{X, \rho, \languagessym{[}  \languagesmv{y_{{\mathrm{1}}}}  \slash  \languagesmv{x_{{\mathrm{1}}}}  \languagessym{,} \, .. \, \languagessym{,}  \languagesmv{y_{\languagesmv{m}}}  \slash  \languagesmv{x_{\languagesmv{m}}}  \languagessym{]} \, e_{{\mathrm{2}}}}$.

  \item Case of (\textsc{E-Let2}):
    \begin{gather*}
      \qclosure{X, \rho, \languageskw{let} \, \languagessym{(}  \languagesmv{x_{{\mathrm{1}}}}  \languagessym{,} \, .. \, \languagessym{,}  \languagesmv{x_{\languagesmv{n}}}  \languagessym{)}  \languagessym{=}  e_{{\mathrm{1}}} \, \languageskw{in} \, e_{{\mathrm{2}}}}  \rightarrow _{ D }  \qclosure{X', \rho', \languageskw{let} \, \languagessym{(}  \languagesmv{x_{{\mathrm{1}}}}  \languagessym{,} \, .. \, \languagessym{,}  \languagesmv{x_{\languagesmv{n}}}  \languagessym{)}  \languagessym{=}  e'_{{\mathrm{1}}} \, \languageskw{in} \, e_{{\mathrm{2}}}} \\
      \qclosure{X, \rho, e_{{\mathrm{1}}}}  \rightarrow _{ D }  \braket{X', \rho', e'_{{\mathrm{1}}}} \\
      \Theta  \pipe  N  \pipe  \Gamma_{{\mathrm{1}}}  \vdash  e_{{\mathrm{1}}}  \languagessym{:}  \tau_{{\mathrm{1}}}  \languagessym{*} \, .. \, \languagessym{*}  \tau_{\languagesmv{n}}
      \andalso \Theta  \pipe  N'  \pipe  \Gamma_{{\mathrm{2}}}  \languagessym{,}  \languagesmv{x_{{\mathrm{1}}}}  \languagessym{:}  \tau_{{\mathrm{1}}}  \languagessym{,} \, .. \, \languagessym{,}  \languagesmv{x_{\languagesmv{n}}}  \languagessym{:}  \tau_{\languagesmv{n}}  \vdash  e_{{\mathrm{2}}}  \languagessym{:}  T \\
      N' = N + \dom(\Gamma_{{\mathrm{1}}}) - n
      \andalso \Theta  \pipe  N'  \pipe  \Gamma_{{\mathrm{2}}}  \vdash  e_{{\mathrm{2}}}  \languagessym{:}  T
    \end{gather*}
    By induction hypothesis with $e_{{\mathrm{1}}}$, there exists $N_{{\mathrm{3}}}$ and $\Gamma_{{\mathrm{3}}}$ such that
    \begin{gather*}
       N_{{\mathrm{3}}}  \pipe  \Gamma_{{\mathrm{3}}}  \vdash_{ D }  \qclosure{X', \rho', e'_{{\mathrm{1}}}}
      \andalso \Theta  \pipe  N_{{\mathrm{3}}}  \pipe  \Gamma_{{\mathrm{3}}}  \vdash  e'_{{\mathrm{1}}}  \languagessym{:}  \tau_{{\mathrm{1}}}  \languagessym{*} \, .. \, \languagessym{*}  \tau_{\languagesmv{n}}
      \andalso |X'| \geq N_{{\mathrm{3}}} \\
      \dom(\Gamma_{{\mathrm{3}}}) \cap X' = \emptyset
      \andalso \dom(\Gamma_{{\mathrm{3}}}) \uplus X' \subseteq \Var(\rho')
      \andalso N + \dom(\Gamma_{{\mathrm{1}}}) = N_{{\mathrm{3}}} + \dom(\Gamma_{{\mathrm{3}}})
    \end{gather*}
    Thus, $N' = N + \dom(\Gamma_{{\mathrm{1}}}) - n = N_{{\mathrm{3}}} + \dom(\Gamma_{{\mathrm{3}}}) - n$.
    By (\textsc{T-Let}),
    \begin{gather*}
      \Theta  \pipe  N_{{\mathrm{3}}}  \pipe  \Gamma_{{\mathrm{3}}}  \languagessym{,}  \Gamma_{{\mathrm{2}}}  \languagessym{,}  \languagesmv{x_{{\mathrm{1}}}}  \languagessym{:}  \tau_{{\mathrm{1}}}  \languagessym{,} \, .. \, \languagessym{,}  \languagesmv{x_{\languagesmv{n}}}  \languagessym{:}  \tau_{\languagesmv{n}}  \vdash  \languageskw{let} \, \languagessym{(}  \languagesmv{x_{{\mathrm{1}}}}  \languagessym{,} \, .. \, \languagessym{,}  \languagesmv{x_{\languagesmv{n}}}  \languagessym{)}  \languagessym{=}  e'_{{\mathrm{1}}} \, \languageskw{in} \, e_{{\mathrm{2}}}  \languagessym{:}  T.
    \end{gather*}
    We put $\Gamma' \coloneqq \Gamma_{{\mathrm{3}}}  \languagessym{,}  \Gamma_{{\mathrm{2}}}  \languagessym{,}  \languagesmv{x_{{\mathrm{1}}}}  \languagessym{:}  \tau_{{\mathrm{1}}}  \languagessym{,} \, .. \, \languagessym{,}  \languagesmv{x_{\languagesmv{n}}}  \languagessym{:}  \tau_{\languagesmv{n}}$.
    The relation $\dom(\Gamma_{{\mathrm{3}}}) \uplus X' \subseteq \Var(\rho')$ implies $\dom(\Gamma') \uplus X' \subseteq \Var(\rho')$.
    The equation $\dom(\Gamma_{{\mathrm{3}}}) \cap X' = \emptyset$ implies $\dom(\Gamma') \cap X' = \emptyset$.
    Therefore, $ N_{{\mathrm{3}}}  \pipe  \Gamma'  \vdash_{ D }  \qclosure{X', \rho', \languageskw{let} \, \languagessym{(}  \languagesmv{x_{{\mathrm{1}}}}  \languagessym{,} \, .. \, \languagessym{,}  \languagesmv{x_{\languagesmv{n}}}  \languagessym{)}  \languagessym{=}  e'_{{\mathrm{1}}} \, \languageskw{in} \, e_{{\mathrm{2}}}}$.

  \item Case of (\textsc{E-Init}):
    \begin{gather*}
      \qclosure{X, \rho, \languageskw{let} \, \languagesmv{x}  \languagessym{=} \, \languageskw{init} \, \languagessym{()} \, \languageskw{in} \, e'}  \rightarrow  \qclosure{X \setminus x', \rho', \languagessym{[}  \languagesmv{x'}  \slash  \languagesmv{x}  \languagessym{]} \, e'} \\
      \Theta  \pipe  N  \pipe  \Gamma  \vdash  \languageskw{let} \, \languagesmv{x}  \languagessym{=} \, \languageskw{init} \, \languagessym{()} \, \languageskw{in} \, e'  \languagessym{:}  T \\
      \Theta  \pipe  N  \languagessym{-}  \languagessym{1}  \pipe  \Gamma  \languagessym{,}  \languagesmv{x}  \languagessym{:}  \languageskw{qbit}  \vdash  e'  \languagessym{:}  T
    \end{gather*}
    By \cref{lem:src-substitution}, $\Theta  \pipe  N  \languagessym{-}  \languagessym{1}  \pipe  \Gamma  \languagessym{,}  \languagesmv{x'}  \languagessym{:}  \languageskw{qbit}  \vdash  \languagessym{[}  \languagesmv{x'}  \slash  \languagesmv{x}  \languagessym{]} \, e'  \languagessym{:}  T$.
    Obviously, the other conditions hold,
    and thus $ N  \languagessym{-}  \languagessym{1}  \pipe  \Gamma  \languagessym{,}  \languagesmv{x'}  \languagessym{:}  \languageskw{qbit}  \vdash_{ D }  \qclosure{X', \rho', e'}$.

  \item Case of (\textsc{E-Discard}): Similar to (\textsc{T-Init}).

  \item Case of (\textsc{E-Call}):
    \begin{gather*}
      \Theta  \languagessym{(}  \languagesmv{f}  \languagessym{)} =  \tau_{{\mathrm{1}}}  \languagessym{*} \, .. \, \languagessym{*}  \tau_{\languagesmv{m}}  \xrightarrow{ N' }  \tau'_{{\mathrm{1}}}  \languagessym{*} \, .. \, \languagessym{*}  \tau'_{\languagesmv{n}} 
      \andalso  \languagesmv{f}  \mapsto ( \languagesmv{y'_{{\mathrm{1}}}}  \languagessym{,} \, .. \, \languagessym{,}  \languagesmv{y'_{\languagesmv{m}}} )  e'  \in D \\
      \qclosure{X, \rho, \languageskw{let} \, \languagessym{(}  \languagesmv{x_{{\mathrm{1}}}}  \languagessym{,} \, .. \, \languagessym{,}  \languagesmv{x_{\languagesmv{n}}}  \languagessym{)}  \languagessym{=}  \languagesmv{f}  \languagessym{(}  \languagesmv{y_{{\mathrm{1}}}}  \languagessym{,} \, .. \, \languagessym{,}  \languagesmv{y_{\languagesmv{m}}}  \languagessym{)} \, \languageskw{in} \, e_{{\mathrm{2}}}}
         \rightarrow  \qclosure{X, \rho, \languageskw{let} \, \languagessym{(}  \languagesmv{x_{{\mathrm{1}}}}  \languagessym{,} \, .. \, \languagessym{,}  \languagesmv{x_{\languagesmv{n}}}  \languagessym{)}  \languagessym{=}  \languagessym{[}  \languagesmv{y_{{\mathrm{1}}}}  \slash  \languagesmv{y'_{{\mathrm{1}}}}  \languagessym{,} \, .. \, \languagessym{,}  \languagesmv{y_{\languagesmv{m}}}  \slash  \languagesmv{y'_{\languagesmv{m}}}  \languagessym{]} \, e' \, \languageskw{in} \, e_{{\mathrm{2}}}} \\
      \Theta  \pipe  N  \pipe  \Gamma_{{\mathrm{1}}}  \languagessym{,}  \languagesmv{x_{{\mathrm{1}}}}  \languagessym{:}  \tau'_{{\mathrm{1}}}  \languagessym{,} \, .. \, \languagessym{,}  \languagesmv{x_{\languagesmv{n}}}  \languagessym{:}  \tau'_{\languagesmv{n}}  \vdash  e_{{\mathrm{2}}}  \languagessym{:}  T \\
      \Theta  \pipe  N  \pipe  \Gamma_{{\mathrm{1}}}  \languagessym{,}  \languagesmv{y_{{\mathrm{1}}}}  \languagessym{:}  \tau_{{\mathrm{1}}}  \languagessym{,} \, .. \, \languagessym{,}  \languagesmv{y_{\languagesmv{m}}}  \languagessym{:}  \tau_{\languagesmv{m}}  \vdash  \languageskw{let} \, \languagessym{(}  \languagesmv{x_{{\mathrm{1}}}}  \languagessym{,} \, .. \, \languagessym{,}  \languagesmv{x_{\languagesmv{n}}}  \languagessym{)}  \languagessym{=}  \languagesmv{f}  \languagessym{(}  \languagesmv{y_{{\mathrm{1}}}}  \languagessym{,} \, .. \, \languagessym{,}  \languagesmv{y_{\languagesmv{m}}}  \languagessym{)} \, \languageskw{in} \, e_{{\mathrm{2}}}  \languagessym{:}  T
    \end{gather*}
    The judgement $\Theta  \vdash  D$ hold, and thus there exists $\Theta'$ such that
    $\Theta'  \pipe  N'  \pipe  \languagesmv{y'_{{\mathrm{1}}}}  \languagessym{:}  \tau_{{\mathrm{1}}}  \languagessym{,} \, .. \, \languagessym{,}  \languagesmv{y'_{\languagesmv{m}}}  \languagessym{:}  \tau_{\languagesmv{m}}  \vdash  e'  \languagessym{:}  \tau'_{{\mathrm{1}}}  \languagessym{*} \, .. \, \languagessym{*}  \tau'_{\languagesmv{n}}$.
    By \cref{lem:src-ext-func-env}, $\Theta  \pipe  N'  \pipe  \languagesmv{y'_{{\mathrm{1}}}}  \languagessym{:}  \tau_{{\mathrm{1}}}  \languagessym{,} \, .. \, \languagessym{,}  \languagesmv{y'_{\languagesmv{m}}}  \languagessym{:}  \tau_{\languagesmv{m}}  \vdash  e'  \languagessym{:}  \tau'_{{\mathrm{1}}}  \languagessym{*} \, .. \, \languagessym{*}  \tau'_{\languagesmv{n}}$.
    By \cref{lem:src-substitution}, $\Theta  \pipe  N'  \pipe  \languagesmv{y_{{\mathrm{1}}}}  \languagessym{:}  \tau_{{\mathrm{1}}}  \languagessym{,} \, .. \, \languagessym{,}  \languagesmv{y_{\languagesmv{m}}}  \languagessym{:}  \tau_{\languagesmv{m}}  \vdash  \languagessym{[}  \languagesmv{y_{{\mathrm{1}}}}  \slash  \languagesmv{y'_{{\mathrm{1}}}}  \languagessym{,} \, .. \, \languagessym{,}  \languagesmv{y_{\languagesmv{m}}}  \slash  \languagesmv{y'_{\languagesmv{m}}}  \languagessym{]} \, e'  \languagessym{:}  \tau'_{{\mathrm{1}}}  \languagessym{*} \, .. \, \languagessym{*}  \tau'_{\languagesmv{n}}$.
    By \cref{lem:src-incr-n}, $\Theta  \pipe  N  \pipe  \languagesmv{y_{{\mathrm{1}}}}  \languagessym{:}  \tau_{{\mathrm{1}}}  \languagessym{,} \, .. \, \languagessym{,}  \languagesmv{y_{\languagesmv{m}}}  \languagessym{:}  \tau_{\languagesmv{m}}  \vdash  \languagessym{[}  \languagesmv{y_{{\mathrm{1}}}}  \slash  \languagesmv{y'_{{\mathrm{1}}}}  \languagessym{,} \, .. \, \languagessym{,}  \languagesmv{y_{\languagesmv{m}}}  \slash  \languagesmv{y'_{\languagesmv{m}}}  \languagessym{]} \, e'  \languagessym{:}  \tau'_{{\mathrm{1}}}  \languagessym{*} \, .. \, \languagessym{*}  \tau'_{\languagesmv{n}}$.
    Therefore, by (\textsc{T-Let}), $\Theta  \pipe  N  \pipe  \Gamma_{{\mathrm{1}}}  \languagessym{,}  \languagesmv{y_{{\mathrm{1}}}}  \languagessym{:}  \tau_{{\mathrm{1}}}  \languagessym{,} \, .. \, \languagessym{,}  \languagesmv{y_{\languagesmv{m}}}  \languagessym{:}  \tau_{\languagesmv{m}}  \vdash  \languageskw{let} \, \languagessym{(}  \languagesmv{x_{{\mathrm{1}}}}  \languagessym{,} \, .. \, \languagessym{,}  \languagesmv{x_{\languagesmv{n}}}  \languagessym{)}  \languagessym{=}  \languagessym{[}  \languagesmv{y_{{\mathrm{1}}}}  \slash  \languagesmv{y'_{{\mathrm{1}}}}  \languagessym{,} \, .. \, \languagessym{,}  \languagesmv{y_{\languagesmv{m}}}  \slash  \languagesmv{y'_{\languagesmv{m}}}  \languagessym{]} \, e' \, \languageskw{in} \, e_{{\mathrm{2}}}  \languagessym{:}  T$.
    The other conditions hold obviously,
    and thus $ N  \pipe  \Gamma_{{\mathrm{1}}}  \languagessym{,}  \languagesmv{y_{{\mathrm{1}}}}  \languagessym{:}  \tau_{{\mathrm{1}}}  \languagessym{,} \, .. \, \languagessym{,}  \languagesmv{y_{\languagesmv{m}}}  \languagessym{:}  \tau_{\languagesmv{m}}  \vdash_{ D }  \qclosure{X, \rho, \languageskw{let} \, \languagessym{(}  \languagesmv{x_{{\mathrm{1}}}}  \languagessym{,} \, .. \, \languagessym{,}  \languagesmv{x_{\languagesmv{n}}}  \languagessym{)}  \languagessym{=}  \languagessym{[}  \languagesmv{y_{{\mathrm{1}}}}  \slash  \languagesmv{y'_{{\mathrm{1}}}}  \languagessym{,} \, .. \, \languagessym{,}  \languagesmv{y_{\languagesmv{m}}}  \slash  \languagesmv{y'_{\languagesmv{m}}}  \languagessym{]} \, e' \, \languageskw{in} \, e_{{\mathrm{2}}}}$.
  \end{itemize}
\end{proof}

Now, we can show \cref{thm:src-soundness} from \cref{lem:src-progress,lem:src-sr}.

\begin{lemma}\label{lem:domG-eq-FV}
  If $\Theta  \pipe  N  \pipe  \Gamma  \vdash  e  \languagessym{:}  T$,
  then $\dom(\Gamma) = \FV(e)$.
\end{lemma}

\begin{proof}
  Straightforward induction on the typing derivation.
\end{proof}

%% file: appendix-tgt.tex
\section{The Target Language}\label{sec:appendix-tgt}

\subsection{The full definition}

\begin{definition}{(Typing rules)}
  The typing rules for the target language are given in \cref{fig:tgt-full-typing}.
\end{definition}

\begin{figure}[tb]
  \fbox{$\TTJudgeExp$}

  \TTReturn

  \TTInit

  \TTSwap

  \TTCnot

  \TTLet

  \TTIf

  \TTCall

  \fbox{$\TTJudgeFunDef$}

  \TTFunDef

  \TTProg

  \caption{The typing rules of the target language}
  \label{fig:tgt-full-typing}
\end{figure}

\begin{definition}{(Operational Semantics)}
  The transition rules on runtime states are given in \cref{fig:tgt-full-semantics}.
\end{definition}

\begin{figure}[htbp]
  \fbox{$\braket{\rho, e}  \rightarrow _{D, G} \braket{\rho', e'}$}

  \TEInit

  \TELetI

  \TELetII

  \TESwap

  \TECnot

  \TECall

  \caption{The operational semantics of the target language}
  \label{fig:tgt-full-semantics}
\end{figure}

\subsection{Proofs}

First, we give the proof of the progress lemma.

\begin{lemma}{(Progress)}\label{lem:tgt-progress}
  If $ \Phi  \pipe  \Gamma  \vdash_{ D }  \qclosure{\rho, e}$,
  then one of the following conditions holds:
  \begin{itemize}
  \item $e = \languagessym{(}  \languagesmv{x_{{\mathrm{1}}}}  \languagessym{,} \, .. \, \languagessym{,}  \languagesmv{x_{\languagesmv{n}}}  \languagessym{)}$.
  \item $\exists \qclosure{\rho', e'}$ . $\qclosure{\rho, e}  \rightarrow _{D, G} \qclosure{\rho', e'}$.
  \end{itemize}
\end{lemma}

\begin{proof}
  By the induction on the typing derivation of $e$.
\end{proof}

\begin{lemma}{(Substitution variables)}\label{lem:tgt-subst-var}
  Suppose that $\Theta  \pipe  \Phi  \pipe  \Gamma  \languagessym{,}  \languagesmv{x_{{\mathrm{1}}}}  \languagessym{:}  \tau_{{\mathrm{1}}}  \languagessym{,} \, .. \, \languagessym{,}  \languagesmv{x_{\languagesmv{n}}}  \languagessym{:}  \tau_{\languagesmv{n}}  \vdash  e  \languagessym{:}  T$ and $ \vdash _{ \text{WF} }  \Phi  \pipe  \Gamma  \languagessym{,}  \languagesmv{y_{{\mathrm{1}}}}  \languagessym{:}  \tau_{{\mathrm{1}}}  \languagessym{,} \, .. \, \languagessym{,}  \languagesmv{y_{\languagesmv{n}}}  \languagessym{:}  \tau_{\languagesmv{n}} $.
  Then $\Theta  \pipe  \Phi  \pipe  \Gamma  \languagessym{,}  \languagesmv{y_{{\mathrm{1}}}}  \languagessym{:}  \tau_{{\mathrm{1}}}  \languagessym{,} \, .. \, \languagessym{,}  \languagesmv{y_{\languagesmv{n}}}  \languagessym{:}  \tau_{\languagesmv{n}}  \vdash  \languagessym{[}  \languagesmv{y_{{\mathrm{1}}}}  \slash  \languagesmv{x_{{\mathrm{1}}}}  \languagessym{,} \, .. \, \languagessym{,}  \languagesmv{y_{\languagesmv{n}}}  \slash  \languagesmv{x_{\languagesmv{n}}}  \languagessym{]} \, e  \languagessym{:}  T$.
\end{lemma}

\begin{proof}
  Straightforward induction on the typing dervivation of $e$.
\end{proof}

\begin{lemma}\label{lem:tgt-rename-qidx}
  If $\Theta  \pipe  \Phi  \pipe  \Gamma  \vdash  e  \languagessym{:}  T$,
  then $\Theta  \pipe  \Phi  \pipe   \sigma _{ \alpha }  \, \Gamma  \vdash  e  \languagessym{:}   \sigma _{ \alpha }  \, T$
  for any $ \sigma _{ \alpha }  = \languagessym{[}  \alpha'_{{\mathrm{1}}}  \slash  \alpha_{{\mathrm{1}}}  \languagessym{,} \, .. \, \languagessym{,}  \alpha'_{\languagesmv{n}}  \slash  \alpha_{\languagesmv{n}}  \languagessym{]}$
  such that $\forall \languagesmv{x}  \languagessym{:}   \texttt{q}( \alpha_{{\mathrm{1}}} )   \languagessym{,}  \languagesmv{y}  \languagessym{:}   \texttt{q}( \alpha_{{\mathrm{2}}} )  \in  \sigma _{ \alpha }  \, \Gamma .\ x \neq y \Rightarrow \alpha_{{\mathrm{1}}} \neq \alpha_{{\mathrm{2}}}$.
\end{lemma}

\begin{proof}
  Straightforward induction on the typing dervivation of $e$.
\end{proof}

\begin{lemma}\label{lem:tgt-ext-O}
  If $\Theta  \pipe  \Phi  \pipe  \Gamma  \vdash  e  \languagessym{:}  T$,
  then $\Theta'  \pipe  \Phi  \pipe  \Gamma  \vdash  e  \languagessym{:}  T$ for any $\Theta \subseteq \Theta'$.
\end{lemma}

\begin{proof}
  Straightforward induction on the typing derivation of $e$.
\end{proof}


Finally, we give the proof of the subject reduction lemma.

\begin{lemma}{(Subject reduction)}\label{lem:tgt-SR}
  If $ \Phi  \pipe  \Gamma  \vdash_{ D }  \qclosure{\rho, e}$ and $\qclosure{\rho, e}  \rightarrow _{D, G} \qclosure{\rho', e'}$,
  then $ \Phi  \pipe  \Gamma  \vdash_{ D }  \qclosure{\rho', e'}$.
\end{lemma}

\begin{proof}
  By induction on the derivation of $\qclosure{\rho, e}  \rightarrow _{D, G} \qclosure{\rho', e'}$.
  \begin{itemize}
  \item Case of (\textsc{E-Init}), (\textsc{E-Let1}), (\textsc{E-Cnot}):
    Immediately follows \cref{lem:tgt-subst-var}.

  \item Case of (\textsc{E-Let2}):
    \begin{gather*}
      \qclosure{\rho, \languageskw{let} \, \languagessym{(}  \languagesmv{x_{{\mathrm{1}}}}  \languagessym{,} \, .. \, \languagessym{,}  \languagesmv{x_{\languagesmv{m}}}  \languagessym{)}  \languagessym{=}  e_{{\mathrm{1}}} \, \languageskw{in} \, e_{{\mathrm{2}}}}  \rightarrow _{D, G} \qclosure{\rho', \languageskw{let} \, \languagessym{(}  \languagesmv{x_{{\mathrm{1}}}}  \languagessym{,} \, .. \, \languagessym{,}  \languagesmv{x_{\languagesmv{n}}}  \languagessym{)}  \languagessym{=}  e'_{{\mathrm{1}}} \, \languageskw{in} \, e_{{\mathrm{2}}}} \\
      \Theta  \pipe  \Phi  \pipe  \Gamma_{{\mathrm{1}}}  \vdash  e_{{\mathrm{1}}}  \languagessym{:}  \tau'_{{\mathrm{1}}}  \languagessym{*} \, .. \, \languagessym{*}  \tau'_{\languagesmv{m}} \\
      \Theta  \pipe  \Phi  \pipe  \Gamma_{{\mathrm{2}}}  \languagessym{,}  \languagesmv{x_{{\mathrm{1}}}}  \languagessym{:}  \tau'_{{\mathrm{1}}}  \languagessym{,} \, .. \, \languagessym{,}  \languagesmv{x_{\languagesmv{m}}}  \languagessym{:}  \tau'_{\languagesmv{m}}  \vdash  e_{{\mathrm{2}}}  \languagessym{:}  \tau_{{\mathrm{1}}}  \languagessym{*} \, .. \, \languagessym{*}  \tau_{\languagesmv{n}}
    \end{gather*}
    By the induction hypothesis, $ \Phi  \pipe  \Gamma_{{\mathrm{1}}}  \vdash_{ D }  \qclosure{\rho', e'_{{\mathrm{1}}}}$.
    That is, $\Theta  \pipe  \Phi  \pipe  \Gamma_{{\mathrm{1}}}  \vdash  e'_{{\mathrm{1}}}  \languagessym{:}  \tau'_{{\mathrm{1}}}  \languagessym{*} \, .. \, \languagessym{*}  \tau'_{\languagesmv{m}}$.
    By (\textsc{T-Let}), $\Theta  \pipe  \Phi  \pipe  \Gamma_{{\mathrm{1}}}  \languagessym{,}  \Gamma_{{\mathrm{2}}}  \vdash  \languageskw{let} \, \languagessym{(}  \languagesmv{x_{{\mathrm{1}}}}  \languagessym{,} \, .. \, \languagessym{,}  \languagesmv{x_{\languagesmv{m}}}  \languagessym{)}  \languagessym{=}  e_{{\mathrm{1}}} \, \languageskw{in} \, e_{{\mathrm{2}}}  \languagessym{:}  \tau_{{\mathrm{1}}}  \languagessym{*} \, .. \, \languagessym{*}  \tau_{\languagesmv{n}}$.
    Therefore, $ \Phi  \pipe  \Gamma_{{\mathrm{1}}}  \languagessym{,}  \Gamma_{{\mathrm{2}}}  \vdash_{ D }  \qclosure{\rho', \languageskw{let} \, \languagessym{(}  \languagesmv{x_{{\mathrm{1}}}}  \languagessym{,} \, .. \, \languagessym{,}  \languagesmv{x_{\languagesmv{m}}}  \languagessym{)}  \languagessym{=}  e'_{{\mathrm{1}}} \, \languageskw{in} \, e_{{\mathrm{2}}}}$.

  \item Case of (\textsc{E-Call}):
    \begin{gather*}
       \languagesmv{f}  \mapsto ( \languagesmv{y'_{{\mathrm{1}}}}  \languagessym{,} \, .. \, \languagessym{,}  \languagesmv{y'_{\languagesmv{l}}} )  e_{{\mathrm{1}}}  \in D
        \quad\quad   \sigma _{ \languagesmv{x} }  = \languagessym{[}  \languagesmv{y_{{\mathrm{1}}}}  \slash  \languagesmv{y'_{{\mathrm{1}}}}  \languagessym{,} \, .. \, \languagessym{,}  \languagesmv{y_{\languagesmv{l}}}  \slash  \languagesmv{y'_{\languagesmv{l}}}  \languagessym{]} \\
      \qclosure{\rho, \languageskw{let} \, \languagessym{(}  \languagesmv{x_{{\mathrm{1}}}}  \languagessym{,} \, .. \, \languagessym{,}  \languagesmv{x_{\languagesmv{m}}}  \languagessym{)}  \languagessym{=}  \languagesmv{f}  \languagessym{(}  \languagesmv{y_{{\mathrm{1}}}}  \languagessym{,} \, .. \, \languagessym{,}  \languagesmv{y_{\languagesmv{l}}}  \languagessym{)} \, \languageskw{in} \, e_{{\mathrm{2}}}}
         \rightarrow  \qclosure{\rho, \languageskw{let} \, \languagessym{(}  \languagesmv{x_{{\mathrm{1}}}}  \languagessym{,} \, .. \, \languagessym{,}  \languagesmv{x_{\languagesmv{m}}}  \languagessym{)}  \languagessym{=}   \sigma _{ \languagesmv{x} }  \, e_{{\mathrm{1}}} \, \languageskw{in} \, e_{{\mathrm{2}}}} \\
      \Theta  \languagessym{(}  \languagesmv{f}  \languagessym{)} =  \forall   \overline{ \alpha }   .   \Phi'  \Rightarrow  \tau''_{{\mathrm{1}}}  \languagessym{*} \, .. \, \languagessym{*}  \tau''_{\languagesmv{m}}  \rightarrow  \tau'''_{{\mathrm{1}}}  \languagessym{*} \, .. \, \languagessym{*}  \tau'''_{\languagesmv{l}}   \\
      \andalso  \sigma _{ \alpha }  = \languagessym{[}   \overline{ \alpha' } / \overline{ \alpha }   \languagessym{]}
      \andalso  \sigma _{ \alpha } \Phi' \subseteq \Phi
      \andalso \forall i \in \{1, \dots, m\}.  \sigma _{ \alpha }  \tau''_{\languagesmv{i}} = \tau'_{\languagesmv{i}} \\
      \Theta  \pipe  \Phi  \pipe  \Gamma  \languagessym{,}  \languagesmv{x_{{\mathrm{1}}}}  \languagessym{:}   \sigma _{ \alpha }  \, \tau'''_{{\mathrm{1}}}  \languagessym{,} \, .. \, \languagessym{,}  \languagesmv{x_{\languagesmv{m}}}  \languagessym{:}   \sigma _{ \alpha }  \, \tau'''_{\languagesmv{m}}  \vdash  e_{{\mathrm{2}}}  \languagessym{:}  \tau_{{\mathrm{1}}}  \languagessym{*} \, .. \, \languagessym{*}  \tau_{\languagesmv{n}}
    \end{gather*}
    For some $\Theta' \subseteq \Theta$,
    \begin{gather*}
      \Theta'  \languagessym{,}  \languagesmv{f}  \languagessym{:}   \Phi  \Rightarrow  \tau''_{{\mathrm{1}}}  \languagessym{*} \, .. \, \languagessym{*}  \tau''_{\languagesmv{l}}  \rightarrow  \tau'''_{{\mathrm{1}}}  \languagessym{*} \, .. \, \languagessym{*}  \tau'''_{\languagesmv{m}}   \pipe  \Phi  \pipe  \languagesmv{y'_{{\mathrm{1}}}}  \languagessym{:}  \tau''_{{\mathrm{1}}}  \languagessym{,} \, .. \, \languagessym{,}  \languagesmv{y'_{\languagesmv{l}}}  \languagessym{:}  \tau''_{\languagesmv{l}}  \vdash  e_{{\mathrm{2}}}  \languagessym{:}  \tau'''_{{\mathrm{1}}}  \languagessym{*} \, .. \, \languagessym{*}  \tau'''_{\languagesmv{m}}
    \end{gather*}
    By \cref{lem:tgt-ext-O}, $\Theta  \pipe  \Phi  \pipe  \languagesmv{y'_{{\mathrm{1}}}}  \languagessym{:}  \tau''_{{\mathrm{1}}}  \languagessym{,} \, .. \, \languagessym{,}  \languagesmv{y'_{\languagesmv{l}}}  \languagessym{:}  \tau''_{\languagesmv{l}}  \vdash  e_{{\mathrm{2}}}  \languagessym{:}  \tau'''_{{\mathrm{1}}}  \languagessym{*} \, .. \, \languagessym{*}  \tau'''_{\languagesmv{m}}$.
    By \cref{lem:tgt-subst-var}, $\Theta  \pipe  \Phi  \pipe  \languagesmv{y_{{\mathrm{1}}}}  \languagessym{:}  \tau''_{{\mathrm{1}}}  \languagessym{,} \, .. \, \languagessym{,}  \languagesmv{y_{\languagesmv{l}}}  \languagessym{:}  \tau''_{\languagesmv{l}}  \vdash   \sigma _{ \languagesmv{x} }  \, e_{{\mathrm{2}}}  \languagessym{:}  \tau'''_{{\mathrm{1}}}  \languagessym{*} \, .. \, \languagessym{*}  \tau'''_{\languagesmv{m}}$.
    By \cref{lem:tgt-rename-qidx}, $\Theta  \pipe  \Phi  \pipe  \languagesmv{y_{{\mathrm{1}}}}  \languagessym{:}  \tau'_{{\mathrm{1}}}  \languagessym{,} \, .. \, \languagessym{,}  \languagesmv{y_{\languagesmv{l}}}  \languagessym{:}  \tau'_{\languagesmv{l}}  \vdash   \sigma _{ \languagesmv{x} }  \, e_{{\mathrm{2}}}  \languagessym{:}   \sigma _{ \alpha }  \, \tau'''_{{\mathrm{1}}}  \languagessym{*} \, .. \, \languagessym{*}   \sigma _{ \alpha }  \, \tau'''_{\languagesmv{m}}$.
    Therefore, by (\textsc{T-Let}),
    \begin{gather*}
      \Theta  \pipe  \Phi  \pipe  \Gamma  \languagessym{,}  \languagesmv{y_{{\mathrm{1}}}}  \languagessym{:}  \tau'_{{\mathrm{1}}}  \languagessym{,} \, .. \, \languagessym{,}  \languagesmv{y_{\languagesmv{l}}}  \languagessym{:}  \tau'_{\languagesmv{l}}  \vdash  \languageskw{let} \, \languagessym{(}  \languagesmv{x_{{\mathrm{1}}}}  \languagessym{,} \, .. \, \languagessym{,}  \languagesmv{x_{\languagesmv{m}}}  \languagessym{)}  \languagessym{=}   \sigma _{ \languagesmv{x} }  \, e_{{\mathrm{1}}} \, \languageskw{in} \, e_{{\mathrm{2}}}  \languagessym{:}  \tau_{{\mathrm{1}}}  \languagessym{*} \, .. \, \languagessym{*}  \tau_{\languagesmv{n}}
    \end{gather*}

  \end{itemize}
\end{proof}

%% file: appendix-alg.tex
\section{The Proofs of Algorithms}\label{sec:appendix-alg}

First, we give the full definition of \textsc{QubitAllocExp} in \cref{alg:appendix-qalloc-exp}.

\begin{algorithm}[tbp]
  \caption{Qubit Allocation for Expressions}
  \label{alg:appendix-qalloc-exp}
  \begin{algorithmic}[1]
    \Function{QubitAllocExp}{$e, \Theta, \Phi, \Gamma_{{\mathrm{1}}}, \Gamma_{{\mathrm{2}}}$}
      \If {$e \equiv \languagessym{(}  \languagesmv{x_{{\mathrm{1}}}}  \languagessym{,} \, .. \, \languagessym{,}  \languagesmv{x_{\languagesmv{n}}}  \languagessym{)}$}
        \State (The expected return type is $ \texttt{q}( \alpha_{{\mathrm{1}}} )   \languagessym{*} \, .. \, \languagessym{*}   \texttt{q}( \alpha_{\languagesmv{k}} ) $)
        \Assert{$\Gamma_{{\mathrm{1}}}, \Gamma_{{\mathrm{2}}} = \languagesmv{x_{{\mathrm{1}}}}  \languagessym{:}   \texttt{q}( \beta_{{\mathrm{1}}} )   \languagessym{,} \, .. \, \languagessym{,}  \languagesmv{x_{\languagesmv{k}}}  \languagessym{:}   \texttt{q}( \beta_{\languagesmv{k}} ) $}
        \State $\Psi \gets$ \Call{TokenSwapping}{$( \mathit{QV}( \Gamma_{{\mathrm{1}}} ) , \Phi), \{\alpha_{\languagesmv{i}} \mapsto \beta_{\languagesmv{i}}\}_{i=1}^k$}
        \State \Return \Call{InsertSwaps}{$\languagessym{(}  \languagesmv{x_{{\mathrm{1}}}}  \languagessym{,} \, .. \, \languagessym{,}  \languagesmv{x_{\languagesmv{k}}}  \languagessym{)}, \Gamma_{{\mathrm{1}}} \cup \Gamma_{{\mathrm{2}}}, \Psi$}

      \ElsIf {$e \equiv \languageskw{let} \, \languagesmv{x}  \languagessym{=} \, \languageskw{init} \, \languagessym{()} \, \languageskw{in} \, e'$}
        \State Take $\languagesmv{x'}  \languagessym{:}   \texttt{q}( \alpha )  \in \Gamma_{{\mathrm{2}}}$
        \State \Return \Call{QubitAllocExp}{$[x' / x]\Delta', \Theta, \Phi, \Gamma_{{\mathrm{1}}} \uplus \{\languagesmv{x'}  \languagessym{:}   \texttt{q}( \alpha ) \}, \Gamma_{{\mathrm{2}}} \backslash \languagesmv{x'}$}

      \ElsIf {$e \equiv \languageskw{discard} \, \languagesmv{x}  \languagessym{;}  e'$}
        \State $e' \gets$ \Call{QubitAllocExp}{$\Delta', \Theta, \Phi, \Gamma_{{\mathrm{1}}}, \Gamma_{{\mathrm{2}}}$}
        \State \Return $\languageskw{init} \, \languagesmv{x}  \languagessym{;}  e'$

      \ElsIf {$e \equiv \languageskw{let} \, \languagessym{(}  \languagesmv{x_{{\mathrm{1}}}}  \languagessym{,} \, .. \, \languagessym{,}  \languagesmv{x_{\languagesmv{n}}}  \languagessym{)}  \languagessym{=}  \languagesmv{f}  \languagessym{(}  \languagesmv{y_{{\mathrm{1}}}}  \languagessym{,} \, .. \, \languagessym{,}  \languagesmv{y_{\languagesmv{m}}}  \languagessym{)} \, \languageskw{in} \, e'$}
        \Assert{$\Theta  \languagessym{(}  \languagesmv{f}  \languagessym{)} =  \forall   \overline{ \alpha }   .   \Phi'  \Rightarrow  \tau_{{\mathrm{1}}}  \languagessym{*} \, .. \, \languagessym{*}  \tau_{\languagesmv{k}}  \rightarrow  \tau'_{{\mathrm{1}}}  \languagessym{*} \, .. \, \languagessym{*}  \tau'_{\languagesmv{k}}  $}
        \State $G \gets ( \mathit{QV}( \Gamma_{{\mathrm{1}}} ) , \Phi); G_f \gets ( \overline{ \alpha } , \Phi')$
        \State $\phi \gets$ \Call{SubgraphIsomorphism}{$G, G_f$}
        \State Take $y_{m+1} : \texttt{q}(\beta_{m+1}), \dots, y_k :  \texttt{q}( \beta_{\languagesmv{k}} )  \in \Gamma_{{\mathrm{2}}}$
        \Assert{$\tau_{{\mathrm{1}}} =  \texttt{q}( \alpha_{{\mathrm{1}}} )  \land \cdots \land \tau_{\languagesmv{k}} =  \texttt{q}( \alpha_{\languagesmv{k}} )  \land \languagesmv{y_{{\mathrm{1}}}}  \languagessym{:}   \texttt{q}( \beta_{{\mathrm{1}}} )   \languagessym{,} \, .. \, \languagessym{,}  \languagesmv{y_{\languagesmv{m}}}  \languagessym{:}   \texttt{q}( \beta_{\languagesmv{m}} )  \in \Gamma_{{\mathrm{1}}}$}
        \State $\Psi \gets$ \Call{TokenSwapping}{$G, \{\beta_{\languagesmv{i}} \mapsto \phi (\alpha_{\languagesmv{i}})\}_{i=1}^k$}
        \State $ \sigma _{ \alpha }  \gets \languagessym{[}   \phi ( \alpha_{{\mathrm{1}}} )   \slash  \alpha_{{\mathrm{1}}}  \languagessym{,} \, .. \, \languagessym{,}   \phi ( \alpha_{\languagesmv{k}} )   \slash  \alpha_{\languagesmv{k}}  \languagessym{]}$
        \State $\Gamma_{{\mathrm{3}}} \gets (\Psi(\Gamma_{{\mathrm{1}}}) \backslash \{\languagesmv{y_{{\mathrm{1}}}}  \languagessym{,} \, .. \, \languagessym{,}  \languagesmv{y_{\languagesmv{m}}}\}) \cup \{\languagesmv{x_{{\mathrm{1}}}}  \languagessym{:}   \sigma _{ \alpha }  \, \tau'_{{\mathrm{1}}}  \languagessym{,} \, .. \, \languagessym{,}  \languagesmv{x_{\languagesmv{n}}}  \languagessym{:}   \sigma _{ \alpha }  \, \tau'_{\languagesmv{n}}\}$
        \State $\Gamma_{{\mathrm{4}}} \gets (\Psi(\Gamma_{{\mathrm{2}}}) \backslash \{y_{m+1}, \dots, y_k\}) \cup \{x_{n+1} :  \sigma _{ \alpha } \tau_{n+1}', \dots, \languagesmv{x_{\languagesmv{k}}}  \languagessym{:}   \sigma _{ \alpha }  \, \tau'_{\languagesmv{k}}\}$, where $x_{n+1}, \dots, x_k$ are fresh variables.
        \State $e' \gets$ \Call{QubitAllocExp}{$e', \Theta, \Phi, \Gamma_{{\mathrm{3}}}, \Gamma_{{\mathrm{4}}}$}
        \State $e \gets \languageskw{let} \, \languagessym{(}  \languagesmv{x_{{\mathrm{1}}}}  \languagessym{,} \, .. \, \languagessym{,}  \languagesmv{x_{\languagesmv{k}}}  \languagessym{)}  \languagessym{=}  \languagesmv{f}  \languagessym{(}  \languagesmv{y_{{\mathrm{1}}}}  \languagessym{,} \, .. \, \languagessym{,}  \languagesmv{y_{\languagesmv{k}}}  \languagessym{)} \, \languageskw{in} \, e'$
        \State \Return \Call{InsertSwaps}{$e, \Gamma_{{\mathrm{1}}} \cup \Gamma_{{\mathrm{2}}}, \Psi$}

      \ElsIf{$e \equiv \languageskw{let} \, \languagessym{(}  \languagesmv{x_{{\mathrm{1}}}}  \languagessym{,}  \languagesmv{x_{{\mathrm{2}}}}  \languagessym{)}  \languagessym{=} \, \languageskw{cnot} \, \languagessym{(}  \languagesmv{y_{{\mathrm{1}}}}  \languagessym{,}  \languagesmv{y_{{\mathrm{2}}}}  \languagessym{)} \, \languageskw{in} \, e'$}
        \State $(\Gamma  \languagessym{(}  \languagesmv{y_{{\mathrm{1}}}}  \languagessym{)} =  \texttt{q}( \alpha_{{\mathrm{1}}} ) , \Gamma  \languagessym{(}  \languagesmv{y_{{\mathrm{2}}}}  \languagessym{)} =  \texttt{q}( \alpha_{{\mathrm{2}}} ) )$
        \State $\beta_{{\mathrm{1}}}  \languagessym{,} \, .. \, \languagessym{,}  \beta_{\languagesmv{L}} \gets$ \Call{ShortestPath}{$( \mathit{QV}( \Gamma_{{\mathrm{1}}}  \languagessym{,}  \Gamma_{{\mathrm{2}}} ) , \Phi), \alpha_{{\mathrm{1}}}, \alpha_{{\mathrm{2}}}$}
        \Assert{$\beta_{{\mathrm{1}}} = \alpha_{{\mathrm{1}}} \land \beta_{\languagesmv{L}} = \alpha_{{\mathrm{2}}}$}
        \State $\Psi \gets (\beta_{{\mathrm{1}}}  \languagessym{,}  \beta_{{\mathrm{2}}}), \dots, (\beta_{L-2}, \beta_{{\languagesmv{L}-1}})$
        \State $e_{{\mathrm{2}}} \gets$ \Call{QubitAllocExp}{$e', \Theta, \Phi,  \Psi (  \languagessym{[}  \languagesmv{x_{{\mathrm{1}}}}  \slash  \languagesmv{y_{{\mathrm{1}}}}  \languagessym{,}  \languagesmv{x_{{\mathrm{2}}}}  \slash  \languagesmv{y_{{\mathrm{2}}}}  \languagessym{]} \, \Gamma_{{\mathrm{1}}}  ) ,  \Psi (  \Gamma_{{\mathrm{2}}}  ) $}
        \State \Return \Call{InsertSwaps}{$e_{{\mathrm{2}}}, \Gamma_{{\mathrm{1}}} \cup \Gamma_{{\mathrm{2}}}$}
      \ElsIf{$e \equiv \languageskw{if} \, \languagesmv{x} \, \languageskw{then} \, e_{{\mathrm{1}}} \, \languageskw{else} \, e_{{\mathrm{2}}}$}
        \State $e'_{{\mathrm{1}}} \gets$ \Call{QubitAllocExp}{$e_{{\mathrm{1}}}, \Theta, \Phi, \Gamma_{{\mathrm{1}}}, \Gamma_{{\mathrm{2}}}$}
        \State $e'_{{\mathrm{2}}} \gets$ \Call{QubitAllocExp}{$e_{{\mathrm{2}}}, \Theta, \Phi, \Gamma_{{\mathrm{1}}}, \Gamma_{{\mathrm{2}}}$}
        \State \Return $\languageskw{if} \, \languagesmv{x} \, \languageskw{then} \, e'_{{\mathrm{1}}} \, \languageskw{else} \, e'_{{\mathrm{2}}}$
      \ElsIf{$e \equiv \languageskw{let} \, \languagessym{(}  \languagesmv{x_{{\mathrm{1}}}}  \languagessym{,} \, .. \, \languagessym{,}  \languagesmv{x_{\languagesmv{n}}}  \languagessym{)}  \languagessym{=}  e_{{\mathrm{1}}} \, \languageskw{in} \, e_{{\mathrm{2}}}$}
        \State $\Gamma'_{{\mathrm{1}}} \gets \{ \languagesmv{x}  \languagessym{:}  \Gamma_{{\mathrm{1}}}  \languagessym{(}  \languagesmv{x}  \languagessym{)} \pipe x \in \FV(e_{{\mathrm{1}}}) \}$
        \State $e'_{{\mathrm{1}}} \gets$ \Call{QubitAllocExp}{$e_{{\mathrm{1}}}, \Theta, \Phi, \Gamma'_{{\mathrm{1}}}, \Gamma_{{\mathrm{2}}}$}
        \State (The return type of $e'_{{\mathrm{1}}}$ is $\tau_{{\mathrm{1}}}  \languagessym{*} \, .. \, \languagessym{*}  \tau_{\languagesmv{k}}$)
        \State $\Gamma''_{{\mathrm{1}}} \gets \{ \languagesmv{x_{{\mathrm{1}}}}  \languagessym{:}  \tau_{{\mathrm{1}}}  \languagessym{,} \, .. \, \languagessym{,}  \languagesmv{x_{\languagesmv{n}}}  \languagessym{:}  \tau_{\languagesmv{n}} \} \cup (\Gamma_{{\mathrm{1}}} \setminus \Gamma'_{{\mathrm{1}}})$
        \State $\Gamma'_{{\mathrm{2}}} \gets x_{n+1} : \tau_{n+1}, \dots, \languagesmv{x_{\languagesmv{k}}}  \languagessym{:}  \tau_{\languagesmv{k}}$, where $\languagesmv{x}_{n+1}, \dots, \languagesmv{x_{\languagesmv{k}}}$ are fresh.
        \State $e'_{{\mathrm{2}}} \gets$ \Call{QubitAllocExp}{$e_{{\mathrm{2}}}, \Theta, \Phi, \Gamma''_{{\mathrm{1}}}, \Gamma'_{{\mathrm{2}}}$}
        \State \Return $\languageskw{let} \, \languagessym{(}  \languagesmv{x_{{\mathrm{1}}}}  \languagessym{,} \, .. \, \languagessym{,}  \languagesmv{x_{\languagesmv{k}}}  \languagessym{)}  \languagessym{=}  e'_{{\mathrm{1}}} \, \languageskw{in} \, e'_{{\mathrm{2}}}$
      \EndIf
    \EndFunction
  \end{algorithmic}
\end{algorithm}

Before proving the correctness of our algorithm. we give some auxiliary notetions and definitions.
\begin{definition}{(Called function names)}
  $\CFN(e)$ denotes the set of function names called in $e$.
\end{definition}

\begin{definition}{(Free variables)}
  $\FV(e)$ denotes the free variables in $e$:
\end{definition}

\begin{definition}{(Consistent function types)}
  A function type $ \forall   \overline{ \alpha }   .   \Phi  \Rightarrow  \tau_{{\mathrm{1}}}  \languagessym{*} \, .. \, \languagessym{*}  \tau_{\languagesmv{k}}  \rightarrow  \tau'_{{\mathrm{1}}}  \languagessym{*} \, .. \, \languagessym{*}  \tau'_{\languagesmv{k}}  $ in the target language
  is \emph{consistent} with a function type $ \tau_{{\mathrm{1}}}  \languagessym{*} \, .. \, \languagessym{*}  \tau_{\languagesmv{n}}  \xrightarrow{ N }  \tau'_{{\mathrm{1}}}  \languagessym{*} \, .. \, \languagessym{*}  \tau'_{\languagesmv{m}} $ in the source language
  if $k = N + n$ and $( \overline{ \alpha } , \Phi)$ is a connected graph.
\end{definition}

We proceed to proofs of lemmas for type-preseving property.

\begin{lemma}\label{lem:alg-insert-swaps}
  Let $\Psi$ be a sequence of pairs of qidx such that
  for all $(\alpha_{{\mathrm{1}}}, \alpha_{{\mathrm{2}}}) \in \Psi$, $ \alpha_{{\mathrm{1}}}  \sim  \alpha_{{\mathrm{2}}}  \in \Phi$.
  Suppose that $\Theta  \pipe  \Phi  \pipe   \Psi (  \Gamma  )   \vdash  e  \languagessym{:}  T$.
  Then, \textsc{InsertSwaps}($e, \Gamma, \Psi$) is well typed under the contexts $\Theta, \Phi, \Gamma$.
\end{lemma}

\begin{proof}
  By induction on the length of $\Psi$.
  The case of $\Psi = \epsilon$ is trivial.
  If $\Psi = (\alpha_{{\mathrm{1}}}, \alpha_{{\mathrm{2}}}), \Psi'$,
  $e'$ is well typed under the contexts $\Theta, \Phi, (\alpha_{{\mathrm{1}}}\ \alpha_{{\mathrm{2}}})(\Gamma)$ by the induction hypothesis
  because $\Psi'((\alpha_{{\mathrm{1}}}\ \alpha_{{\mathrm{2}}})(\Gamma)) = \Psi(\Gamma)$.
  From the assumption of this lemma, $ \alpha_{{\mathrm{1}}}  \sim  \alpha_{{\mathrm{2}}}  \in \Phi$.
  Therefore, by (\textsc{T-Swap}), $\languageskw{let} \, \languagessym{(}  \languagesmv{x_{{\mathrm{2}}}}  \languagessym{,}  \languagesmv{x_{{\mathrm{1}}}}  \languagessym{)}  \languagessym{=} \, \languageskw{swap} \, \languagessym{(}  \languagesmv{x_{{\mathrm{1}}}}  \languagessym{,}  \languagesmv{x_{{\mathrm{2}}}}  \languagessym{)} \, \languageskw{in} \, e'$ is well typed under the contexts $\Theta, \Phi, \Gamma$.
\end{proof}

\begin{lemma}\label{lem:alg-qalloc-exp}
  Suppose that $\Theta  \pipe  N  \pipe  \Gamma  \vdash  e  \languagessym{:}  T$.
  Let $\Theta', \Phi, \Gamma_{{\mathrm{1}}}, \Gamma_{{\mathrm{2}}}$ satisfy all the following conditions:
  \begin{itemize}
  \item $\dom(\Gamma_{{\mathrm{1}}}) \cap \dom(\Gamma_{{\mathrm{2}}}) = \emptyset$,
  \item $ \mathit{QV}( \Gamma_{{\mathrm{1}}} )  \cap  \mathit{QV}( \Gamma_{{\mathrm{2}}} )  = \emptyset$,
  \item $| \mathit{QV}( \Gamma_{{\mathrm{2}}} ) | \geq N$,
  \item $\dom(\Gamma) = \dom(\Gamma_{{\mathrm{1}}})$.
  \item $\dom(\Theta) = \dom(\Theta') \land \forall f \in \dom(\Theta').\ \Theta'  \languagessym{(}  \languagesmv{f}  \languagessym{)}$ is consistent with $\Theta  \languagessym{(}  \languagesmv{f}  \languagessym{)}$,
  \item $( \mathit{QV}( \Gamma_{{\mathrm{1}}}  \languagessym{,}  \Gamma_{{\mathrm{2}}} ) , \Phi)$ is a connected graph, and
  \item For all $f \in \CFN(e)$, if $\Theta  \languagessym{(}  \languagesmv{f}  \languagessym{)} =  \forall   \overline{ \alpha }   .   \Phi'  \Rightarrow  \tau_{{\mathrm{1}}}  \languagessym{*} \, .. \, \languagessym{*}  \tau_{\languagesmv{n}}  \rightarrow  \tau'_{{\mathrm{1}}}  \languagessym{*} \, .. \, \languagessym{*}  \tau'_{\languagesmv{m}}  $,
    then there exists $G \subseteq ( \mathit{QV}( \Gamma_{{\mathrm{1}}}  \languagessym{,}  \Gamma_{{\mathrm{2}}} ) , \Phi)$ such that $G \simeq ( \overline{ \alpha } , \Phi')$.
  \end{itemize}
  Then \textsc{QubitAllocExp}($e, \Theta', \Phi, \Gamma_{{\mathrm{1}}}, \Gamma_{{\mathrm{2}}})$ is well typed under $\Theta', \Phi, \Gamma_{{\mathrm{1}}}  \languagessym{,}  \Gamma_{{\mathrm{2}}}$,
  and the size of its return type equals $|\dom(\Gamma_{{\mathrm{1}}})| + |\dom(\Gamma_{{\mathrm{2}}})|$.
\end{lemma}

\begin{proof}
  By induction on the typing derivation tree of $\Theta  \pipe  N  \pipe  \Gamma  \vdash  e  \languagessym{:}  T$.
  \begin{itemize}
  \item Case of (\textsc{T-Return}):
    Immediately follows \cref{lem:alg-insert-swaps}.

  \item Case of (\textsc{T-Init}):
    \begin{gather*}
      \Theta  \pipe  N  \languagessym{-}  \languagessym{1}  \pipe  \Gamma  \languagessym{,}  \languagesmv{x}  \languagessym{:}   \texttt{q}( \alpha )   \vdash  e'  \languagessym{:}  T
    \end{gather*}
    There exists $\languagesmv{x'}  \languagessym{:}   \texttt{q}( \alpha ) $ in $\Gamma_{{\mathrm{2}}}$ because $N \geq 1$.
    By \cref{lem:src-substitution}, $\Theta  \pipe  N  \languagessym{-}  \languagessym{1}  \pipe  \Gamma  \languagessym{,}  \languagesmv{x'}  \languagessym{:}   \texttt{q}( \alpha )   \vdash  \languagessym{[}  \languagesmv{x'}  \slash  \languagesmv{x}  \languagessym{]} \, e'  \languagessym{:}  T$.
    By the induction hypothesis,
    $e' = \textsc{QubitAllocExp}(\languagessym{[}  \languagesmv{x'}  \slash  \languagesmv{x}  \languagessym{]} \, e', \Theta', \Phi, \Gamma_{{\mathrm{1}}} \uplus \{\languagesmv{x'}  \languagessym{:}   \texttt{q}( \alpha ) \}, \Gamma_{{\mathrm{2}}} \backslash \languagesmv{x'})$
    is well typed.
    Therefore, $\languageskw{init} \, \languagesmv{x'}  \languagessym{;}  e'$ is well typed under $\Theta', \Phi, \Gamma_{{\mathrm{1}}}  \languagessym{,}  \Gamma_{{\mathrm{2}}}$ by (\textsc{T-Init}).

  \item Case of (\textsc{T-Call}):
    By the induction hypothesis, $e'$ is well typed under the contexts $\Theta, \Phi, \Gamma_{{\mathrm{3}}}  \languagessym{,}  \Gamma_{{\mathrm{4}}}$.
    For all $ \alpha_{{\mathrm{1}}}  \sim  \alpha_{{\mathrm{2}}}  \in \Phi'$, $\phi(\alpha_{{\mathrm{1}}}) \sim \phi(\alpha_{{\mathrm{2}}}) \in \Phi$
    because $\phi(G_f) \simeq G_f$.
    Thus, $ \sigma _{ \alpha }  \, \Phi' \subseteq \Phi$.
    By (\textsc{T-Call}), $\languageskw{let} \, \languagessym{(}  \languagesmv{x_{{\mathrm{1}}}}  \languagessym{,} \, .. \, \languagessym{,}  \languagesmv{x_{\languagesmv{k}}}  \languagessym{)}  \languagessym{=}  \languagesmv{f}  \languagessym{(}  \languagesmv{y_{{\mathrm{1}}}}  \languagessym{,} \, .. \, \languagessym{,}  \languagesmv{y_{\languagesmv{k}}}  \languagessym{)} \, \languageskw{in} \, e'$ is well typed under the contexts $\Theta, \Phi, \Psi(\Gamma_{{\mathrm{1}}}, \Gamma_{{\mathrm{2}}})$.
    Therefore, by \cref{lem:alg-insert-swaps}, \textsc{InsertSwaps}($e, \Gamma_{{\mathrm{1}}} \cup \Gamma_{{\mathrm{2}}}, \Psi$) is well typed under the contexts $\Theta, \Phi, \Gamma_{{\mathrm{1}}}  \languagessym{,}  \Gamma_{{\mathrm{2}}}$.
    Moreover, $|\dom(\Gamma_{{\mathrm{3}}})| = |\dom(\Gamma_{{\mathrm{1}}})| - m + n$ and $|\dom(\Gamma_{{\mathrm{4}}})| = |\dom(\Gamma_{{\mathrm{2}}})| - (k - m) + (k - n)$.
    Thus, $|\dom(\Gamma_{{\mathrm{3}}})| + |\dom(\Gamma_{{\mathrm{4}}})| = |\dom(\Gamma_{{\mathrm{1}}})| + |\dom(\Gamma_{{\mathrm{2}}})|$.

  \item Case of (\textsc{T-Let}):
    \begin{gather*}
      e \equiv \languageskw{let} \, \languagessym{(}  \languagesmv{x_{{\mathrm{1}}}}  \languagessym{,} \, .. \, \languagessym{,}  \languagesmv{x_{\languagesmv{n}}}  \languagessym{)}  \languagessym{=}  e_{{\mathrm{1}}} \, \languageskw{in} \, e_{{\mathrm{2}}}
      \andalso N' = N + |\dom(\Gamma_{{\mathrm{3}}})| - n \\
      \Theta  \pipe  N  \pipe  \Gamma_{{\mathrm{3}}}  \vdash  e_{{\mathrm{1}}}  \languagessym{:}  \tau_{{\mathrm{1}}}  \languagessym{*} \, .. \, \languagessym{*}  \tau_{\languagesmv{n}}
      \andalso \Theta  \pipe  N'  \pipe  \Gamma_{{\mathrm{4}}}  \languagessym{,}  \languagesmv{x_{{\mathrm{1}}}}  \languagessym{:}  \tau_{{\mathrm{1}}}  \languagessym{,} \, .. \, \languagessym{,}  \languagesmv{x_{\languagesmv{n}}}  \languagessym{:}  \tau_{\languagesmv{n}}  \vdash  e_{{\mathrm{2}}}  \languagessym{:}  T \\
      \Gamma = \Gamma_{{\mathrm{3}}}  \languagessym{,}  \Gamma_{{\mathrm{4}}}
      \andalso \Gamma'_{{\mathrm{1}}} = \{\languagesmv{x}  \languagessym{:}  \Gamma_{{\mathrm{1}}}  \languagessym{(}  \languagesmv{x}  \languagessym{)} \pipe x \in \FV(e_{{\mathrm{1}}})\} \\
      \Gamma''_{{\mathrm{1}}} = \{\languagesmv{x_{{\mathrm{1}}}}  \languagessym{:}  \tau_{{\mathrm{1}}}  \languagessym{,} \, .. \, \languagessym{,}  \languagesmv{x_{\languagesmv{n}}}  \languagessym{:}  \tau_{\languagesmv{n}}\} \cup (\Gamma_{{\mathrm{1}}} \backslash \Gamma'_{{\mathrm{1}}}) \\
      \Gamma'_{{\mathrm{2}}} = \{x_{n+1} : \tau_{n+1}, \dots, \languagesmv{x_{\languagesmv{k}}}  \languagessym{:}  \tau_{\languagesmv{k}}\},~ \text{where}~ x_{n+1}, \dots, x_k~ \text{are fresh}
    \end{gather*}
    By \cref{lem:domG-eq-FV}, $\FV(e_{{\mathrm{1}}}) = \dom(\Gamma_{{\mathrm{3}}}) = \dom(\Gamma'_{{\mathrm{1}}}) \subseteq \dom(\Gamma_{{\mathrm{1}}})$.
    Thus, $\dom(\Gamma'_{{\mathrm{1}}}) \cap \dom(\Gamma_{{\mathrm{2}}}) = \emptyset$ because $\dom(\Gamma_{{\mathrm{1}}}) \cap \dom(\Gamma_{{\mathrm{2}}}) = \emptyset$.
    Moreover, by the induction hypothesis, $e'_{{\mathrm{1}}}$ is well typed under $\Theta, \Phi, \Gamma'_{{\mathrm{1}}}  \languagessym{,}  \Gamma_{{\mathrm{2}}}$
    and $k = |\dom(\Gamma'_{{\mathrm{1}}})| + |\dom(\Gamma_{{\mathrm{2}}})|$.
    Now, the inequation
    \begin{align*}
      |\dom(\Gamma'_{{\mathrm{2}}})| - N' &= k - N - |\dom(\Gamma_{{\mathrm{3}}})| \\
      &\geq k - |\dom(\Gamma_{{\mathrm{2}}})| - |\dom(\Gamma_{{\mathrm{3}}})| \quad\quad \text{from the assumption of this lemma} \\
      &= k - |\dom(\Gamma_{{\mathrm{2}}})| - |\dom(\Gamma'_{{\mathrm{1}}})| \quad\quad \text{from the definition of}~ \Gamma'_{{\mathrm{1}}} \\
      &= 0
    \end{align*}
    holds. Therefore, by the induction hypothesis, $e'_{{\mathrm{2}}}$ is well typed under $\Theta, \Phi, \Gamma''_{{\mathrm{1}}}  \languagessym{,}  \Gamma'_{{\mathrm{2}}}$
    and the size of its return type equals $|\dom(\Gamma''_{{\mathrm{1}}})| + |\dom(\Gamma'_{{\mathrm{2}}})| = |\dom(\Gamma_{{\mathrm{1}}})| + |\dom(\Gamma_{{\mathrm{2}}})|$.
    The equation $(\Gamma'_{{\mathrm{1}}}  \languagessym{,}  \Gamma_{{\mathrm{2}}}  \languagessym{,}  \Gamma''_{{\mathrm{1}}}  \languagessym{,}  \Gamma_{{\mathrm{2}}}) \setminus \{\languagesmv{x_{{\mathrm{1}}}}  \languagessym{,} \, .. \, \languagessym{,}  \languagesmv{x_{\languagesmv{k}}}\} = \Gamma_{{\mathrm{1}}}  \languagessym{,}  \Gamma_{{\mathrm{2}}}$ holds.
    Therefore, by (\textsc{T-Let}), $\languageskw{let} \, \languagessym{(}  \languagesmv{x_{{\mathrm{1}}}}  \languagessym{,} \, .. \, \languagessym{,}  \languagesmv{x_{\languagesmv{k}}}  \languagessym{)}  \languagessym{=}  e'_{{\mathrm{1}}} \, \languageskw{in} \, e'_{{\mathrm{2}}}$ is well typed under $\Theta, \Phi, \Gamma_{{\mathrm{1}}}  \languagessym{,}  \Gamma_{{\mathrm{2}}}$,
    and the size of its return type equals $|\dom(\Gamma_{{\mathrm{1}}})| + |\dom(\Gamma_{{\mathrm{2}}})|$.

  \item The other cases are easy or similar to above cases.
  \end{itemize}
\end{proof}

\begin{lemma}\label{lem:alg-qalloc-func}
  Suppose that $\Theta  \vdash  D$
  and a graph assignment $S$ satisfies the condition
  $\forall  \languagesmv{f}  \mapsto ( \languagesmv{x_{{\mathrm{1}}}}  \languagessym{,} \, .. \, \languagessym{,}  \languagesmv{x_{\languagesmv{n}}} )  e  \in D.\ \forall g \in \CFN(e).\ \exists G' \subset S(f).\ G' \simeq S(g)$.
  Then \textsc{QubitAllocFunc}($\Theta, D, S$) returns $\Theta'$ and $D'$,
  and $\Theta'  \vdash  D'$.
\end{lemma}

\begin{proof}
  By induction on the derivation tree of $\Theta  \vdash  D$.
  The base case, $\Theta = \bullet$ and $D = \bullet$, is trivial.
  We assume that $\Theta = \Theta'  \languagessym{,}  \languagesmv{f}  \languagessym{:}   \tau_{{\mathrm{1}}}  \languagessym{*} \, .. \, \languagessym{*}  \tau_{\languagesmv{n}}  \xrightarrow{ N }  \tau'_{{\mathrm{1}}}  \languagessym{*} \, .. \, \languagessym{*}  \tau'_{\languagesmv{m}} $
  and $D = D'  \languagessym{,}   \languagesmv{f}  \mapsto ( \languagesmv{x_{{\mathrm{1}}}}  \languagessym{,} \, .. \, \languagessym{,}  \languagesmv{x_{\languagesmv{n}}} )  e $.
  By the induction hypothesis, $\textsc{QubitAllocFunc}(\Theta', D', S)$ outputs $\Theta_{{\mathrm{1}}}, D_{{\mathrm{1}}}$
  and $\Theta_{{\mathrm{1}}}  \vdash  D_{{\mathrm{1}}}$.
  By \cref{lem:alg-qalloc-exp}, $e'$ is well typed under the contexts $\Theta_{{\mathrm{2}}}, \Phi, \Gamma_{{\mathrm{1}}}  \languagessym{,}  \Gamma_{{\mathrm{2}}}$.
  Therefore, by (\textsc{T-FunDef}) and (\textsc{T-FunDecl}), $\Theta_{{\mathrm{2}}}  \vdash  D_{{\mathrm{2}}}$.
\end{proof}

Finally, we prove the correctness of our algorithm.
\begin{proof}[Proof of \cref{thm:alg-type-preserve}]
  We can prove easily by \cref{lem:alg-qalloc-exp,lem:alg-qalloc-func}.
\end{proof}

%% file: main.bbl
\begin{thebibliography}{10}

\bibitem{qiskitteamIBMQXBackend}
{{IBM QX}} backend information.
\newblock \url{https://github.com/Qiskit/ibmq-device-information}.

\bibitem{siraichiQubitAllocation2018}
Marcos~Yukio Siraichi, Vin{\'i}cius~Fernandes dos Santos, Sylvain Collange, and
  Fernando Magno~Quintao Pereira.
\newblock Qubit allocation.
\newblock In {\em Proceedings of the 2018 {{International Symposium}} on {{Code
  Generation}} and {{Optimization}}}, {{CGO}} 2018, pp. 113--125, {New York,
  NY, USA}, February 2018. {Association for Computing Machinery}.

\bibitem{siraichiQubitAllocationCombination2019}
Marcos~Yukio Siraichi, Vin{\'i}cius~Fernandes dos Santos, Caroline Collange,
  and Fernando Magno~Quint{\~a}o Pereira.
\newblock Qubit allocation as a combination of subgraph isomorphism and token
  swapping.
\newblock {\em Proceedings of the ACM on Programming Languages}, Vol.~3, No.
  OOPSLA, pp. 120:1--120:29, October 2019.

\bibitem{dengCodarContextualDurationAware2020}
Haowei Deng, Yu~Zhang, and Quanxi Li.
\newblock Codar: A contextual duration-aware qubit mapping for various {{NISQ}}
  devices.
\newblock In {\em 2020 57th {{ACM}}/{{IEEE Design Automation Conference}}
  ({{DAC}})}, pp. 1--6, July 2020.

\bibitem{nishioExtractingSuccessIBM2020}
Shin Nishio, Yulu Pan, Takahiko Satoh, Hideharu Amano, and Rodney~Van Meter.
\newblock Extracting success from ibm's 20-qubit machines using error-aware
  compilation.
\newblock {\em ACM Journal on Emerging Technologies in Computing Systems},
  Vol.~16, No.~3, pp. 32:1--32:25, May 2020.

\bibitem{muraliFormalConstraintbasedCompilation2019a}
Prakash Murali, Ali {Javadi-Abhari}, Frederic~T. Chong, and Margaret Martonosi.
\newblock Formal constraint-based compilation for noisy intermediate-scale
  quantum systems.
\newblock {\em Microprocessors and Microsystems}, Vol.~66, pp. 102--112, April
  2019.

\bibitem{turnerOnceType1995}
David~N. Turner, Philip Wadler, and Christian Mossin.
\newblock Once upon a type.
\newblock In {\em Proceedings of the Seventh International Conference on
  {{Functional}} Programming Languages and Computer Architecture}, {{FPCA}}
  '95, pp. 1--11, {New York, NY, USA}, October 1995. {Association for Computing
  Machinery}.

\bibitem{jonesTheoryQualifiedTypes1992}
Mark~P. Jones.
\newblock A theory of qualified types.
\newblock In Bernd {Krieg-Br{\"u}ckner}, editor, {\em {{ESOP}} '92}, Lecture
  {{Notes}} in {{Computer Science}}, pp. 287--306, {Berlin, Heidelberg}, 1992.
  {Springer}.

\bibitem{nielsenQuantumComputationQuantum2011}
Michael~A. Nielsen and Isaac~L. Chuang.
\newblock {\em Quantum Computation and Quantum Information: 10th Anniversary
  Edition}.
\newblock {Cambridge University Press}, {USA}, tenth edition, 2011.

\bibitem{hopcroftAlgorithm447Efficient1973}
John Hopcroft and Robert Tarjan.
\newblock Algorithm 447: Efficient algorithms for graph manipulation.
\newblock {\em Communications of the ACM}, Vol.~16, No.~6, pp. 372--378,
  January 1973.

\bibitem{yamanakaSwappingLabeledTokens2014}
Katsuhisa Yamanaka, Erik~D. Demaine, Takehiro Ito, Jun Kawahara, Masashi
  Kiyomi, Yoshio Okamoto, Toshiki Saitoh, Akira Suzuki, Kei Uchizawa, and
  Takeaki Uno.
\newblock Swapping labeled tokens on graphs.
\newblock In Alfredo Ferro, Fabrizio Luccio, and Peter Widmayer, editors, {\em
  Fun with {{Algorithms}}}, Lecture {{Notes}} in {{Computer Science}}, pp.
  364--375, {Cham}, 2014. {Springer International Publishing}.

\bibitem{miltzowApproximationHardnessToken2016}
Tillmann Miltzow, Lothar Narins, Yoshio Okamoto, G{\"u}nter Rote, Antonis
  Thomas, and Takeaki Uno.
\newblock Approximation and hardness for token swapping.
\newblock {\em arXiv:1602.05150 [cs]}, August 2016.

\bibitem{ohoriRegisterAllocationProof2004}
Atsushi Ohori.
\newblock Register allocation by proof transformation.
\newblock {\em Science of Computer Programming}, Vol.~50, No.~1, pp. 161--187,
  March 2004.

\bibitem{itokoOptimizationQuantumCircuit2020}
Toshinari Itoko, Rudy Raymond, Takashi Imamichi, and Atsushi Matsuo.
\newblock Optimization of quantum circuit mapping using gate transformation and
  commutation.
\newblock {\em Integration}, Vol.~70, pp. 43--50, January 2020.

\bibitem{hietalaVerifiedOptimizerQuantum2021}
Kesha Hietala, Robert Rand, Shih-Han Hung, Xiaodi Wu, and Michael Hicks.
\newblock A verified optimizer for quantum circuits.
\newblock {\em Proceedings of the ACM on Programming Languages}, Vol.~5, No.
  POPL, pp. 37:1--37:29, January 2021.

\bibitem{amySizedTypesLowlevel2019}
Matthew Amy.
\newblock Sized types for low-level quantum metaprogramming.
\newblock {\em arXiv:1908.02644 [quant-ph]}, Vol. 11497, pp. 87--107, 2019.

\bibitem{hughesProvingCorrectnessReactive1996}
John Hughes, Lars Pareto, and Amr Sabry.
\newblock Proving the correctness of reactive systems using sized types.
\newblock In {\em Proceedings of the 23rd {{ACM SIGPLAN-SIGACT}} Symposium on
  {{Principles}} of Programming Languages}, {{POPL}} '96, pp. 410--423, {New
  York, NY, USA}, January 1996. {Association for Computing Machinery}.

\bibitem{maslovReversibleLogicSynthesis}
Dmitri Maslov.
\newblock Reversible logic synthesis benchmarks page.
\newblock \url{https://reversiblebenchmarks.github.io}.

\bibitem{wakizakaDependentTypeSystem2020}
Ryo Wakizaka and Atsushi Igarashi.
\newblock A dependent type system for qubit connectivity verification
  (japanese).
\newblock {\em Proceeding of the JSSST Annual Conference (Web)}, Vol. 37th,
  p.~35, 2020.

\end{thebibliography}
